\pdfoutput=1

\documentclass{article}

\usepackage[a4paper]{geometry}	% Page layout
\usepackage{lmodern,microtype}	% Allowing modern fonts and more flexibility
\usepackage{tikz}				% For pictures
\usepackage[english]{babel}		% Language
\usepackage{color}				% Colours

\usepackage[nobreak]{cite} %Turns a citation that looks like [3,4,5] into [3-5]. The nobreak option prevents cite from breaking at the end of a line.

\usepackage{hyperref}
\hypersetup{
    colorlinks,
    citecolor=red,
    filecolor=black,
    linkcolor=darkblue,
    urlcolor=black
}

\usepackage{amssymb,amsthm,amsmath,bm}

\definecolor{darkblue}{rgb}{0,0,.8}

\usetikzlibrary{arrows,decorations,calc}
\usetikzlibrary{decorations.pathmorphing,patterns,decorations.pathreplacing,decorations.markings}

\numberwithin{equation}{section} %This next line must appear before cleveref for the referencing to be correct

\usepackage[capitalise,noabbrev]{cleveref}

\newtheorem{theorem}{Theorem}[section]
\newtheorem{lemma}[theorem]{Lemma}
\newtheorem{proposition}[theorem]{Proposition}
\newtheorem{corollary}[theorem]{Corollary}
\newtheorem{conjecture}[theorem]{Conjecture}

\newcommand{\rrangle}{\rangle\hspace{-.05cm}\rangle}
\newcommand{\llangle}{\langle\hspace{-.05cm}\langle}

\newcommand{\chit}{\protect\raisebox{0.25ex}{$\chi$}}

\renewcommand{\ge}{\geqslant}
\renewcommand{\le}{\leqslant}

\newcommand{\drawvertex}[3]
{
 
  \ifnum#1=1
   \draw[xshift=#2 cm,yshift=#3 cm,postaction={on each segment={mid arrow}}] (0,0)--
(.5,0)--(1,0);
  \draw[xshift=#2 cm,yshift=#3 cm,postaction={on each segment={mid arrow}}] (.5,-.
5)--(.5,0)--(.5,.5);
  \fi
   \ifnum#1=2
      \draw[xshift=#2 cm,yshift=#3 cm,postaction={on each segment={mid arrow}}] 
(1,0)--(.5,0)--(0,0);
  \draw[xshift=#2 cm,yshift=#3 cm,postaction={on each segment={mid arrow}}] (.
5,.5)--(.5,0)--(.5,-.5); 
  \fi
\ifnum#1=3
      \draw[xshift=#2 cm,yshift=#3 cm,postaction={on each segment={mid arrow}}] 
(1,0)--(.5,0)--(0,0);
  \draw[xshift=#2 cm,yshift=#3 cm,postaction={on each segment={mid arrow}}] (.5,-.
5)--(.5,0)--(.5,.5);   \fi
\ifnum#1=4
      \draw[xshift=#2 cm,yshift=#3 cm,postaction={on each segment={mid arrow}}] 
(0,0)--(.5,0)--(1,0);
  \draw[xshift=#2 cm,yshift=#3 cm,postaction={on each segment={mid arrow}}] (.
5,.5)--(.5,0)--(.5,-.5);  
  \fi
\ifnum#1=5
      \draw[xshift=#2 cm,yshift=#3 cm,postaction={on each segment={mid arrow}}] 
(0,0)--(.5,0)--(.5,.5);
  \draw[xshift=#2 cm,yshift=#3 cm,postaction={on each segment={mid arrow}}] (1,0)--
(.5,0)--(.5,-.5);  
  \fi
\ifnum#1=6
      \draw[xshift=#2 cm,yshift=#3 cm,postaction={on each segment={mid arrow}}] (.
5,.5)--(.5,0)--(0,0);
  \draw[xshift=#2 cm,yshift=#3 cm,postaction={on each segment={mid arrow}}] (.5,-.
5)--(.5,0)--(1,0);  
  \fi
}

\renewcommand{\i}{\text{i}}
\newcommand{\diff}{\text{d}}
\newcommand{\Det}[2]{\det_{#1}^{#2}}
\newcommand{\Pf}[2]{\mathop{\mathrm{pf}}_{#1}^{#2}}

\tikzset{
  on each segment/.style={
    decorate,
    decoration={
      show path construction,
      moveto code={},
      lineto code={
        \path [#1]
        (\tikzinputsegmentfirst) -- (\tikzinputsegmentlast);
      },
      curveto code={
        \path [#1] (\tikzinputsegmentfirst)
        .. controls
        (\tikzinputsegmentsupporta) and (\tikzinputsegmentsupportb)
        ..
        (\tikzinputsegmentlast);
      },
      closepath code={
        \path [#1]
        (\tikzinputsegmentfirst) -- (\tikzinputsegmentlast);
      },
    },
  },
  mid arrow/.style={postaction={decorate,decoration={
        markings,
        mark=at position .625 with {\arrow[#1]{stealth}}
      }}},
}

\title{ \bf Symmetry classes of alternating sign matrices\\ in the nineteen-vertex model} 
\author{\normalsize \textsc{Christian Hagendorf}, 
 \textsc{Alexi Morin-Duchesne}
\medskip\\
{\normalsize
  \begin{minipage}{\textwidth}
  \begin{center}
  \textit{
   Universit\'e Catholique de Louvain\\
  Institut de Recherche en Math\'ematique et Physique\\
  Chemin du Cyclotron 2, 1348 Louvain-la-Neuve, Belgium}\\
  \medskip
  \href{mailto:christian.hagendorf@uclouvain.be}{\normalsize 
\texttt{christian.hagendorf@uclouvain.be}},
\href{mailto:alexi.morin-duchesne@uclouvain.be}{\normalsize 
\texttt{alexi.morin-duchesne@uclouvain.be}}
  \end{center}
  \end{minipage}
}
}
\date{}

\begin{document}
\maketitle
\vspace{0.5cm}

\begin{abstract}

The nineteen-vertex model on a periodic lattice with an anti-diagonal twist is investigated. 
Its inhomogeneous transfer matrix is shown to have a simple eigenvalue,
with the corresponding eigenstate displaying intriguing combinatorial features.
Similar results were
previously found for the same model with a diagonal twist. 
The eigenstate for the anti-diagonal twist is explicitly constructed using the quantum separation of variables technique. 
A number of sum rules and special components are computed and expressed in terms of Kuperberg's determinants for partition functions of the inhomogeneous six-vertex model.
The computations of some components of the special eigenstate for the diagonal twist are also presented. In the homogeneous limit, the special eigenstates become eigenvectors of the Hamiltonians of the
integrable spin-one XXZ chain with twisted boundary conditions. Their sum rules and special components for both twists are expressed in terms of generating functions arising in the weighted enumeration of various symmetry classes of alternating sign matrices (ASMs). These include half-turn symmetric ASMs, quarter-turn symmetric ASMs, vertically symmetric ASMs, vertically and horizontally perverse ASMs and double U-turn ASMs. As side results, new determinant and pfaffian formulas for the weighted enumeration of various symmetry classes of alternating sign matrices are obtained.

\medskip

\noindent Keywords: Integrable spin-one XXZ chain, alternating sign matrices, nineteen-vertex model, quantum inverse scattering method.
\end{abstract}
\newpage

\tableofcontents
\newpage

%%%%%%%%%%%%%%%%%%%%%
%
\section{Introduction}
%
%%%%%%%%%%%%%%%%%%%%%

The seminal work of Razumov and Stroganov \cite{razumov:00} revealed a remarkable combinatorial structure of the periodic Heisenberg
XXZ spin chain with anisotropy $\Delta=-1/2$. They investigated the ground state of the Hamiltonian
for chains of odd length $N=2n+1$ and observed that in a suitable normalisation, many components and scalar products 
are given by integer sequences in $n$ which appear in the enumeration of alternating sign matrices (ASMs) and plane partitions \cite{bressoudbook}. 
Similar connections were found for chains of even length with
 twisted boundary conditions, as well as for open chains with boundary magnetic fields \cite{razumov:01,batchelor:01,degier:02}. In all
these cases, some of the sequences that arise enumerate ASMs invariant under certain symmetries, such as reflections and rotations. Subsequently, the ground state of the Hamiltonian of the dense $O(1)$ loop model was also found to involve sequences enumerating ASMs or symmetry classes thereof \cite{mitra:04,mitra:04_2}.

 A major leap forward in understanding and proving these observations
 was made by Di Francesco and Zinn-Justin \cite{difrancesco:05_3}. Their idea was to exploit the connection between one-dimensional quantum systems
 and two-dimensional models of classical statistical mechanics, specifically the six-vertex model and the dense $O(1)$ loop model. By considering inhomogeneous versions of these models, they reformulated the problem of obtaining the 
Hamiltonian's ground state eigenvector into that of finding a special eigenvector of the inhomogeneous transfer matrix. 
 The powerful tools of quantum integrability make this computation feasible,
 with the results for the Hamiltonian recovered in the homogeneous limit.
 This approach led to the rigorous proofs of
 a large number of properties of the XXZ ground states at $\Delta=-1/2$, 
 namely exact finite-size sum rules \cite{difrancesco:06}, 
integral formulas for all components \cite{razumov:07}, 
 as well as exact results for correlation functions \cite{kitanine:02,cantini:12_1}.
 Higher-spin \cite{zinn:09,fonseca:12} and higher-rank systems \cite{difrancesco:05_4} were addressed along similar lines, confirming the existence of so-called combinatorial points, 
 namely special
values of the anisotropy parameter where the ground states exhibit a relation to combinatorial problems. 

In addition to a rich combinatorial structure, the XXZ chain at $\Delta=-1/2$ possesses another interesting feature: an exact lattice supersymmetry \cite{yang:04,veneziano:06}. 
Its Hamiltonian can be written as the anti-commutator of a nilpotent operator and its adjoint. The distinctive feature of the above-mentioned XXZ ground states is that they are so-called supersymmetry singlets, which 
play a special role in supersymmetric theories \cite{witten:82}. The coincidence between the connection to combinatorics
and the supersymmetric structure is still not well understood, but nonetheless
suggests lattice supersymmetry as a heuristic tool to detect
combinatorial features of other spin-chain models, for instance at higher spin. 

In \cite{hagendorf:13}, it was shown that,
for a one-parameter family of twisted boundary conditions,
 such a lattice supersymmetry is present for the integrable spin-one XXZ chain \cite{zamolodchikov:81,fateev:81}, irrespectively of the value taken by the anisotropy parameter. 
This suggested some eigenvectors of the spin-chain Hamiltonian could again display interesting combinatorial structures. 
Indeed, for one specific twist, it was observed in \cite{hagendorf:13} that the Hamiltonian possesses a special eigenvector which is also a supersymmetry singlet. In contrast with the spin-$\frac12$ chain, this occurs \textit{for chains of any length and for all values of the anisotropy parameter}. 
In a suitable normalisation, some of its components and certain scalar products  
were conjectured in \cite{hagendorf:13} to coincide with polynomials in the anisotropy parameter which are known generating functions appearing in a particular type of weighted enumeration of ASMs \cite{robbins:00,kuperberg:02}. 
The two-dimensional lattice model underlying the spin chain is a nineteen-vertex model built from the fusion of the six-vertex model \cite{kulish:81,kulish:82,kirillov:87}.
Inspired by the ideas of Di Francesco and Zinn-Justin 
\cite{difrancesco:05_3} outlined above, the previous observations were examined in \cite{hagendorf:15} for the inhomogeneous model and its transfer matrix. Applying the formalism of the quantum inverse scattering method (QISM) \cite{korepin:93}, a simple eigenvalue was found and the corresponding special eigenvector explicitly constructed using the algebraic Bethe ansatz. 
Scalar products involving this vector were explicitly computed and found to reproduce partition functions of the six-vertex model, whose homogeneous limits are related to ASM enumeration \cite{kuperberg:02}, thus confirming some observations made in \cite{hagendorf:13}. That partition functions of the six-vertex model appear in the nineteen-vertex model comes as a surprise and lacks a more profound understanding.

The present paper continues the study of the combinatorial structures of the integrable spin-one XXZ chain and its corresponding nineteen-vertex model. Our objectives are twofold.
On the one hand, we find new sum rules and components of the special eigenvector for the diagonal twisted boundary condition studied in \cite{hagendorf:15}, and thereby prove some conjectures of \cite{hagendorf:13} that were still open.
On the other hand, we identify a new anti-diagonal twisted boundary condition for which the transfer matrix also possesses a simple eigenvalue, and find that the corresponding eigenvector has a particularly rich combinatorial 
structure. We treat both cases using QISM techniques, specifically the algebraic Bethe ansatz in the first case
and the quantum separation of variables technique \cite{niccoli:13,niccoli:15} in the second.
 For both twists, we explicitly compute certain
 components and scalar products of the special eigenvectors in the inhomogeneous case in terms of a variety of six-vertex model partition functions. Their homogeneous limit yields generating functions for the enumeration of symmetry classes of ASMs \cite{kuperberg:02}. 
The combinatorial quantities hidden in the special spin-chain eigenvector are thus readily revealed.
As side results, we find new determinant and pfaffian expressions for the generating functions of these symmetry classes of ASMs, extending a calculation by Behrend, Di Francesco and Zinn-Justin \cite{behrend:12}.
 
The layout of this paper is as follows. In \cref{sec:spinoneXXZ}, we discuss the integrable spin-one XXZ chain with particular diagonal and anti-diagonal twisted boundary conditions. We present our main results for 
the special eigenvectors of the spin-chain Hamiltonian in both cases and their relation to the weighted enumeration of ASMs. In \cref{sec:transfermatrices}, we review the construction of the inhomogeneous nineteen-vertex model via the fusion procedure and its analysis using the QISM formalism. We prove the existence of a simple eigenvalue of the transfer matrices of the vertex models and construct the corresponding eigenvectors. In \cref{sec:sumrules,sec:specialcomps}, we compute scalar products and special components in terms of partition functions of a six-vertex model. We also obtain explicit formulas for the homogeneous limit of the various partition functions. We present concluding remarks in \cref{sec:conclusion} along with an overview of open problems. In \cref{app:BYBE,app:ZA}, we respectively discuss some technicalities regarding solutions to the boundary Yang-Baxter equation and present the derivation of an auxiliary partition function.

%%%%%%%%%%%%%%%%%%%%%
%
\section{The integrable spin-one XXZ chain}
\label{sec:spinoneXXZ}
%
%%%%%%%%%%%%%%%%%%%%%

In this section, we discuss our motivation, namely the study of the integrable spin-one XXZ chain with certain twisted boundary conditions, and present our main results.
The spin-chain Hamiltonian has a special eigenvector which we conjecture to be the ground-state for a wide range of the anisotropy parameter.
As we will see, a number of its sum rules and components display intriguing relations with the combinatorics of weighted enumeration of ASMs.

In \cref{sec:hamXXZ}, we introduce the spin-chain Hamiltonian and discuss some characteristics of its spectrum, in particular the existence of a special eigenvalue. \cref{sec:qsumrules} addresses square norms and scalar products involving the corresponding special eigenvector, as well as their connection to the weighted
enumeration of ASMs with symmetries. We state results for certain components of the eigenvectors in \cref{sec:specialcomponents}. Throughout the whole section, we deliberately omit the proofs as the results follow from more general findings in \cref{sec:transfermatrices,sec:sumrules,sec:specialcomps}.

%%%%%%%%%
\subsection{Hamiltonians and simple eigenvalues} 
\label{sec:hamXXZ}
%%%%%%%%%

\paragraph{Definition.} We consider a periodic chain with $N$ sites, each one carrying a quantum spin. The Hilbert space of the system is $V = V_1 \otimes V_2\otimes \cdots \otimes V_N$ where $V_j \simeq \mathbb C^3$. We denote by
\begin{equation}
  |{\Uparrow}\rangle= 
  \begin{pmatrix}
    1\\ 0 \\ 0
  \end{pmatrix}
  ,\quad  
  |{0}\rangle= 
  \begin{pmatrix}
    0\\ 1 \\ 0
  \end{pmatrix}
 ,  \quad 
  |{\Downarrow}\rangle = 
  \begin{pmatrix}
    0\\ 0 \\ 1
  \end{pmatrix}
\end{equation}
the canonical basis vectors of the Hilbert space for a single spin. The canonical basis of $V$ is given by the states $|\sigma_1\sigma_2\cdots\sigma_N\rangle = |\sigma_1\rangle \otimes |\sigma_2\rangle \otimes \cdots\otimes
|\sigma_N\rangle$ where 
$\sigma_j = \,\Uparrow,0$ or $\Downarrow$ for $j=1,\dots,N$. Furthermore, the spin operators for a single site are those of the spin-one representation of $\mathfrak{su}(2)$:
\begin{equation}
   s^{1} = \frac{1}{\sqrt{2}}
   \left(
   \begin{array}{ccc}
   0 & 1 & 0\\
   1 & 0 & 1\\
   0 & 1 & 0
   \end{array}
   \right), \quad
   s^{2} = \frac{1}{\sqrt{2}}
   \left(
   \begin{array}{ccc}
   0 & -\i & 0\\
   \i & 0 & -\i\\
   0 & \i & 0
   \end{array}
   \right),\quad
   s^{3} = 
   \left(
   \begin{array}{ccc}
   1 & 0 & 0\\
   0 & 0 & 0\\
   0 & 0 & -1
   \end{array}
   \right).
\end{equation}
We write $s_j^a$ for the operator $s^a$ acting on the site $j$.

The Hamiltonian which we will study is an integrable spin-one generalisation of the familiar XXZ chain. It is given by \cite{zamolodchikov:81,fateev:81}
\begin{equation}
  H = \sum_{j=1}^N \left(\sum_{a=1}^3 J_a (s_j^a s_{j+1}^a+ 2(s_j^a)^2) - 
\sum_{a,b=1}^3 A_{ab}s_j^as_j^b s_{j+1}^a s_{j+1}^b\right).
\label{eqn:ham}
\end{equation}
The coupling constants $J_a$ and $A_{ab}$ are subject to the relations $A_{ab}=A_{ba}$ and $A_{aa}=J_a$. The remaining constants are expressed in terms of a single real parameter $x$, which measures the anisotropy of the spin chain:
\begin{equation}
  J_1 = J_2 =1,\quad  J_3 = \frac{1}{2}(x^2-2),\quad  A_{12}=1,\quad  A_{13}=A_{23}=x-1.
\end{equation}
For example, the choice $x=2$ corresponds to the well-known Babujan-Takhtajan spin chain \cite{babujian:82,babujian:83,takhtajan:82}.

We also need to specify the boundary conditions, which amounts to relating $s_{N+1}^a$, for $a = 1,2,3$, to spin operators acting on the first site.
The simple identification $s^a_{N+1}=s_{1}^a$ corresponds
to the common choice of periodic boundary conditions. In this paper however, we consider twisted boundary conditions which differ from the periodic case. We focus on two particular cases, which we call the \textit{diagonal twist} and the
\textit{anti-diagonal twist}. They are defined by 
\begin{subequations}\label{eqn:twists}
\begin{alignat}{4}
   & s_{N+1}^1 = -s_1^1, & & s_{N+1}^2 = -s_1^2,& \quad &s_{N+1}^3 = s_1^3, & \quad & \text{(diagonal twist)},\\
   & s_{N+1}^1 = s_1^1, &\quad & s_{N+1}^2 = -s_1^2,& \quad & s_{N+1}^3 = -s_1^3, & \quad & \text{(anti-diagonal twist)}.
\end{alignat}
\end{subequations}
The terminology \textit{(anti-)diagonal twist} will become clear in \cref{sec:transfermatrices}. We note that for both twists and for real $x$, the Hamiltonian is hermitian and therefore diagonalisable with real eigenvalues.

\paragraph{Spectrum.} 
The model is integrable for both twists as we shall see in \cref{sec:transfermatrices}.
In order to understand their special nature beyond integrability, we discuss the spectra of the Hamiltonians. To this end, let us examine some of their symmetries.
First, for both twists, $H$ is spin-reversal invariant:
\begin{equation}
  [H,F] = 0, \quad F = 
  \begin{pmatrix}
    0 & 0 & 1\\
    0 & 1 & 0\\
    1 & 0 & 0
  \end{pmatrix} 
  ^{\otimes N}.
\end{equation}
Second, for the diagonal twist, the Hamiltonian commutes with the total magnetisation 
\begin{equation}\label{eq:mag}
M = \sum_{j=1}^N s_j^3.
\end{equation}
In the case of the anti-diagonal twist, the magnetisation is however only conserved mod $2$:
\begin{equation}
  [H,(-1)^M] = 0.
\end{equation}
These two $\mathbb Z_2$ symmetry operators allow us to divide the Hilbert space $V$ into subsectors that are invariant under the action of $H$. In particular, the special eigenvectors studied in later sections belong to specific eigenspaces of these $\mathbb Z_2$ operators.

We start our discussion of the spectrum with the diagonal twist. The spin-reversal invariance implies that the Hamiltonian does not couple the eigenspaces of the operator $F$. Hence, we may diagonalise $H$ separately within them. We shall prove the following
proposition.
\begin{proposition}\label{eq:spectra.diag}
   For any $N\ge 2$ and $x \in \mathbb R$,
    the non-zero part of the spectrum of $H$ with diagonal twist restricted to the subsectors where $F\equiv 1$ and $F\equiv -1$ coincide (including degeneracies).
\end{proposition}
\begin{figure}[h] 
\centering
\includegraphics[width=.9\textwidth]{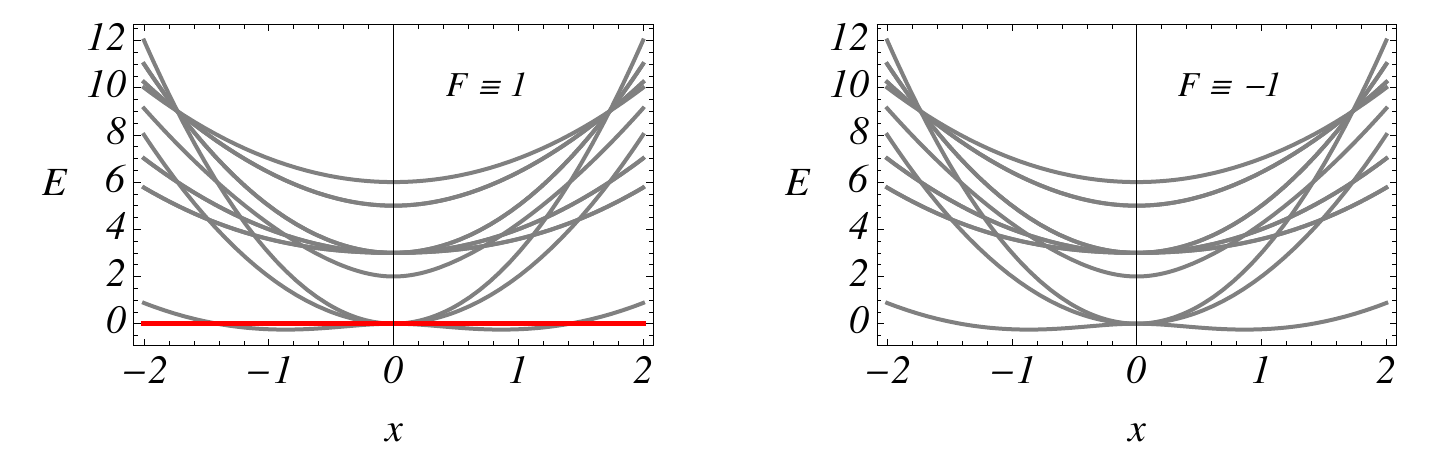} 
\caption{The spectrum of the Hamiltonian with the diagonal twist restricted to the two eigenspaces of the spin-flip
operator $F$ as a function of $x$, for $N=3$ sites. With the exception of an extra eigenvalue $E=0$ 
in the subsector where $F\equiv 1$, the spectra coincide exactly.}
\label{fig:specdiagN3}
\end{figure}
The spectrum for $N=3$ as a function of $x$ is illustrated in \cref{fig:specdiagN3}. 
The simple eigenvalue $E=0$ is the only one not present in both subsectors.
We call the corresponding eigenvectors \textit{zero-energy states}.
\begin{theorem}  \label{thm:DTwistSpecialEV}
  For any $N\ge 2$ and $x \in \mathbb R$,
  the Hamiltonian with diagonal twist possesses the simple eigenvalue $E=0$ in the subsector where $M\equiv 0$ and $F\equiv 1$. 
\end{theorem}
The existence of the simple eigenvalue was proved in \cite{hagendorf:15} through the explicit construction of a zero-energy state with zero magnetisation, which we denote by $|\phi_{\text{\rm \tiny D}}\rangle$. Its invariance under 
spin reversal will be proved in the present article.

For the anti-diagonal twist, the situation is similar. \cref{fig:specantidiagN3} illustrates the spectrum of the Hamiltonian restricted to the two eigenspaces of the operator $(-1)^M$ for $N=3$ sites as a function of $x$. We see once again that the spectra coincide exactly with the exception of the eigenvalue $E=0$. Indeed, we shall prove the following statement.

\begin{proposition}\label{eq:spectra.antidiag} 
   For any $N\ge 2$ and $x \in \mathbb R$,
   the non-zero part of the spectrum of $H$ with anti-diagonal twist restricted to the subsectors where $(-1)^M \equiv 1$ and $(-1)^M \equiv -1$ coincide (including degeneracies).
\end{proposition}
\begin{figure}[h]
\centering
\includegraphics[width=.9\textwidth]{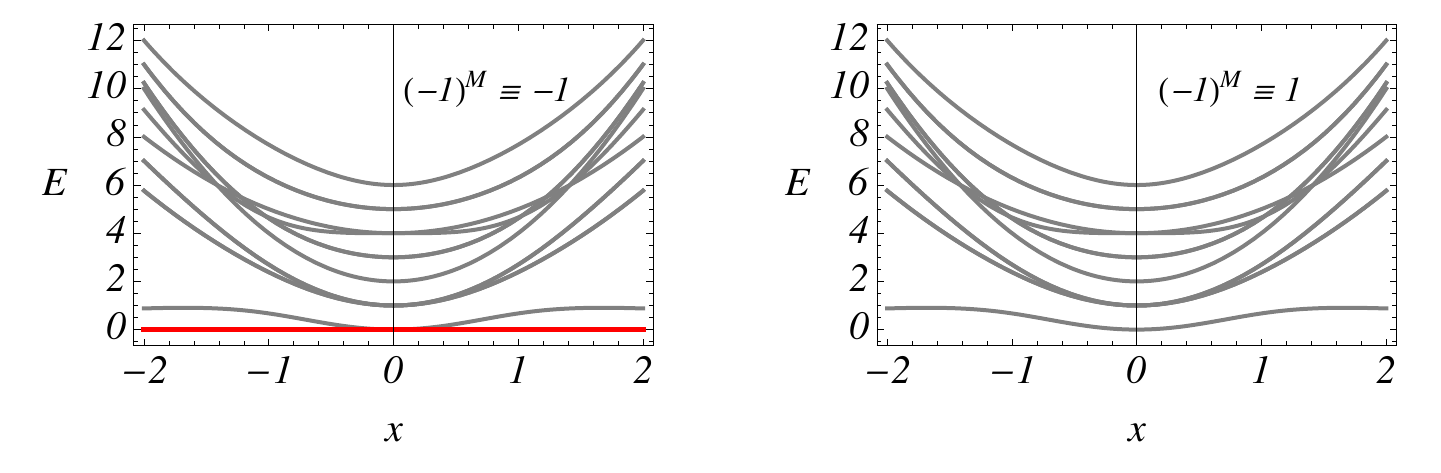}
\caption{The spectrum of the Hamiltonian with the anti-diagonal twist restricted to the two eigenspaces of the operator $(-1)^M$ as a function of $x$, for $N=3$ sites. The two spectra coincide exactly with the exception of the simple eigenvalue $E=0$.}
\label{fig:specantidiagN3}
\end{figure} 
As for the diagonal twist, we focus on the special eigenvalue $E=0$. \begin{theorem}
   \label{thm:ADTwistSpecialEV}
  For any $N\ge 2$ and $x \in \mathbb R$, the Hamiltonian with anti-diagonal twist possesses the simple eigenvalue $E=0$ in the subsector where $(-1)^M \equiv (-1)^N$ and $F\equiv (-1)^{N}$.
\end{theorem} 

We shall prove this theorem through the explicit construction of a zero-energy state $|\phi_{\text{\rm \tiny AD}}\rangle$ in the subsector stated in this theorem.

\paragraph{Zero-energy states and ground states.} For $N=3$ sites, the spectra displayed in \cref{fig:specdiagN3,fig:specantidiagN3} suggest that $E=0$ is the ground-state eigenvalue for certain ranges of the anisotropy parameter $x$. 
Using \textsc{Mathematica}, we have computed the exact spectra of the two Hamiltonians for system sizes up to $N=7$ and many values for $x$ in the interval $[-10,10]$.

The results from this investigation are compatible with the following three conjectures.
\begin{conjecture} For chains of even length 
$N$ and all values of the anisotropy parameter $x$, the spectrum of the Hamiltonian with the diagonal twist is non-negative. The ground-state eigenvalue is $E=0$. It is non-degenerate for $x\neq 0$ and has degeneracy $N+1$ for $x = 0$.
\end{conjecture}
\begin{conjecture}
For chains of odd length $N$ and $0< |x|< \sqrt{2}$, the Hamiltonian with the diagonal twist has a single negative doubly-degenerate ground-state eigenvalue.\footnote{This was pointed out to us by R.~Weston and J.~Yang.} All other eigenvalues are non-negative, and the first excited state has the non-degenerate eigenvalue $E=0$. 
For $x = 0$ and $|x|\ge\sqrt{2}$, the spectrum is non-negative and $E=0$ is the ground-state eigenvalue. It has degeneracy $2N+1$ for $x = 0$, three for $|x|=\sqrt{2}$ and one for $|x|>\sqrt{2}$.
\end{conjecture}
\begin{conjecture}
  For chains of arbitrary length and all values of the anisotropy parameter $x$, the spectrum of the Hamiltonian with anti-diagonal twist is non-negative. $E=0$ is the ground-state eigenvalue. It is non-degenerate for $x\neq 0$ and has degeneracy three for $x = 0$.
\end{conjecture}

%%%%%%%%%
\subsection{Scalar products}
\label{sec:qsumrules}
%%%%%%%%%

In this section and the following, we present our results on various scalar products and components of the zero-energy states $|\phi_{\text{\rm \tiny D}}\rangle$ and $|\phi_{\text{\rm \tiny AD}}\rangle$. These quantities involve various polynomials related to the weighted enumeration of alternating sign matrices belonging to specific symmetry classes.

For any vector $|\phi\rangle$ in $V$, we define its components to be the coefficients of its expansion along the canonical basis vectors,
\begin{equation}
  |\phi\rangle = \sum_{\sigma_1,\dots,\sigma_N}\phi_{\sigma_1\cdots\sigma_N}|\sigma_1\cdots\sigma_N\rangle.
\end{equation}
Likewise for any covector $\langle\phi'|$ in the dual space $V^\ast$, the expansion along the dual canonical basis vectors,
\begin{equation}
  \langle \phi'| = \sum_{\sigma_1,\dots,\sigma_N}\phi'_{\sigma_1\cdots\sigma_N}\langle \sigma_1\cdots\sigma_N|,
\end{equation}
defines its components. The dual canonical basis vectors are obtained from the canonical basis vectors by transposition: $\langle \sigma_1 \cdots \sigma_N| = |\sigma_1\cdots\sigma_N\rangle^t$. The dual pairing is thus
\begin{equation}
  \langle \sigma'_1 \cdots \sigma'_N|\sigma_1\cdots\sigma_N\rangle = \prod_{j=1}^N \delta_{\sigma_j\sigma'_j}.
\end{equation}
The pairing extends by linearity to arbitrary vectors and covectors, and yields the (real) scalar product $\langle \phi'|\phi\rangle = \sum_{\sigma_1,\dots,\sigma_N}\phi'_{\sigma_1\cdots\sigma_N}\phi_{\sigma_1\cdots\sigma_N}$. To any vector $|\phi\rangle$ corresponds naturally a co-vector $\langle \phi|=|\phi\rangle^t$. In this case, the dual pairing yields the square norm
\begin{equation}
  ||\phi||^2 =\langle \phi |\phi\rangle= \sum_{\sigma_1,\dots,\sigma_N}\phi_{\sigma_1\cdots\sigma_N}^2
\end{equation}
of the vector $|\phi\rangle$.

\paragraph{Diagonal twist.} 
For the diagonal twist, we will use a normalisation for $|\phi_{\text{\rm \tiny D}}\rangle$ that appears naturally from the construction of an eigenstate $|\psi_{\text{\rm \tiny D}}\rangle$ of the transfer matrix of the inhomogeneous nineteen-vertex model. This is discussed in \cref{sec:homog.definitions}, with the normalisation of $|\phi_{\text{\rm \tiny D}}\rangle$ defined by \eqref{eqn:defpsiD} and \eqref{eqn:defPhiD}. With this convention, it was shown in \cite{hagendorf:15} that all components of $|\phi_{\text{\rm \tiny D}}\rangle$ are polynomials in $x$ with integer coefficients. The following result was also proven in \cite{hagendorf:15}.
\begin{theorem}[\hspace{-0.01cm}\cite{hagendorf:15}]
\label{thm:phiDsquarenorm}
  The square norm is given by
\begin{equation}
  ||\phi_{\text{\rm \tiny D}}||^2 = A(N;x^2),
\end{equation}
where the polynomial $A(N;t)$ is the generating function for the $t$-enumeration of ASMs described below.
\end{theorem}

ASMs are square matrices whose entries are
$-1,0$ or $1$, and are
such that \textit{(i)} 
each row and each column sums to one and
\textit{(ii)} the non-zero entries along each row and column alternate in sign.
For instance, the seven
$3\times 3$ alternating sign matrices are given by 
\begin{equation}
\begin{array}{c}
\begin{pmatrix}+&0&0\\0&+&0\\0&0&+\end{pmatrix},
\quad
\begin{pmatrix}+&0&0\\0&0&+\\0&+&0\end{pmatrix},
\quad
\begin{pmatrix}0&+&0\\+&0&0\\0&0&+\end{pmatrix},
\quad
\begin{pmatrix}0&+&0\\0&0&+\\+&0&0\end{pmatrix},\\[0.7cm]
\begin{pmatrix}0&0&+\\+&0&0\\0&+&0\end{pmatrix},
\quad
\begin{pmatrix}0&0&+\\0&+&0\\+&0&0\end{pmatrix},
\quad
\begin{pmatrix}0&+&0\\+&\!-\!&+\\0&+&0\end{pmatrix},
\end{array}
\end{equation}
where we write $\pm$ for $\pm 1$. In order to construct the polynomial $A(N;t)$, we assign the weight $t^k$ to each ASM with $k$ negative entries. The function $A(N;t)$ is defined as the sum of the weights of the ASMs of size $N\times N$. We refer to this as {\it $t$-enumeration}. From the list of ASMs of size $3\times 3$, we see that $A(3;t)=6+t$.

\paragraph{Anti-diagonal twist.} Polynomials related to weighted ASM enumeration also occur in the case of the anti-diagonal twist. As before, we need to fix the normalisation of the zero-energy state $|\phi_{\text{\tiny AD}}\rangle$: We choose 
\begin{equation}\label{eq:normAD}
  (\phi_{\text{\tiny AD}})_{\Uparrow\cdots \Uparrow}=1
\end{equation}
for all $N$. We will argue in \cref{eq:exiandpoly} that with this normalisation, all other components are polynomials in $x$.
Instead of considering the square norm, we consider the more general scalar product
\begin{equation}
  Z_{\text{\tiny AD}}(y)=\langle\phi_{\text{\rm \tiny AD}}|y^M|\phi_{\text{\rm \tiny AD}}\rangle.
\end{equation} 
It is a Laurent polynomial in $y$ for which the coefficient of $y^m$ is the square norm of the vector's projection on the subsector where $M \equiv m$. 
 We have the following closed form:
\begin{theorem}
The zero-energy state for the anti-diagonal twist satisfies the sum rule 
\begin{equation}
Z_{\text{\rm \tiny AD}}(y)=\Det{i,j=0}{N-1}\left(y^{-1}\delta_{ij}+y \sum_{k=0}^{N-1}\binom{i}{k} \binom{j}{k} x^{2(j-k)}\right).
\label{eqn:sumrulead}
\end{equation}
\end{theorem}
Although it is not obvious, this expression is in fact invariant under $y \to y^{-1}$ due to the spin-reversal invariance of the zero-energy state,  $F|\phi_{\text{\rm \tiny AD}}\rangle = (-1)^{N}|\phi_{\text{\rm \tiny AD}}\rangle$.
There are special values for $y$ where \eqref{eqn:sumrulead} becomes a generating function for the enumeration of ASMs. The simplest one is $y = q$, where $x$ and $q$ are related by $x=q+q^{-1}$. For this value, we obtain 
\begin{equation}
   \label{eqn:ZADASMGF}
   Z_{\text{\tiny AD}}(q) = x^N A(N;x^2), \quad x=q+q^{-1},
\end{equation}
with $A(N;t)$ the same polynomial as above. The more intuitive choice $y=1$ yields the square norm of $|\phi_{\text{\tiny AD}}\rangle$.

To describe this quantity as well as others appearing later, we need to consider certain families of ASMs which are invariant under the action of specific symmetries, such as rotation or reflection symmetries, and their weighted enumerations. For a given invariant matrix, the images of a negative entry under the repeated action of the symmetry group form a so-called orbit of the group action. We assign the
weight $t^k$ to a matrix with $k$ such orbits of negative entries. In practice, $k$ can be determined by counting the negative entries in a fundamental domain of the group action.
As before, the generating function for the $t$-enumeration is then just given as the sum of the weights of
all matrices of interest. 

In the present case, we consider all $2N\times 2N$ ASMs with half-turn symmetry. A fundamental domain is given by all entries to the left of the median (see the left panel of \cref{fig:asmhtqt} for an example).
\begin{figure}[h]  
  \centering
  \begin{tikzpicture}
    \draw (0,0) node{$\left(
\begin{array}{cccccc}
0 & 0 & + & 0 & 0 & 0 \\
0 & + & - & 0 & + & 0 \\
0 & 0 & + & 0 & - & + \\
+ & - & 0 & + & 0 & 0 \\
0 & + & 0 & - & + & 0 \\
0 & 0 & 0 & + & 0 & 0
\end{array}
\right)$};
  \draw[dotted] (0,-1.25) -- (0,1.25);

  \draw (6,0) node{$
   \left(
\begin{array}{cccccccc}
 0 & 0 & 0 & 0 & + & 0 & 0 & 0 \\
 0 & + & 0 & 0 & - & 0 & + & 0 \\
 0 & 0 & 0 & 0 & + & 0 & 0 & 0 \\
 + & - & + & 0 & 0 & 0 & 0 & 0 \\
 0 & 0 & 0 & 0 & 0 & + & - & + \\
 0 & 0 & 0 & + & 0 & 0 & 0 & 0 \\
 0 & + & 0 & - & 0 & 0 & + & 0 \\
 0 & 0 & 0 & + & 0 & 0 & 0 & 0
\end{array}
\right)$};
 \draw[dotted] (6,-1.7) -- (6,1.7);
\draw[dotted] (3.6,0) -- (8.4,0); 
  \end{tikzpicture}
  \caption{\textit{Left:} an ASM of size $6 \times 6$ with half-turn symmetry. It possesses $k=2$ orbits
  of negative entries. \textit{Right:} a QTASM of size $8 \times 8$ with $k=1$ orbit of negative entries and $m=4$ non-zero entries in the upper-left quadrant.}
  \label{fig:asmhtqt}
  \end{figure}
We denote by $A_{\text{\rm \tiny HT}}^+(2N;t)$ the corresponding generating function. We show in \cref{sec:sumrules} that the square norm is
\begin{equation}
   \label{eqn:ZADHTASMGF}
  Z_{\text{\tiny AD}}(1)=||\phi_{\text{\rm \tiny AD}}||^2 = \frac{A_{\text{\rm \tiny HT}}^+
  (2N;x^2)}{A(N;x^2)}.
\end{equation}

Finally, we consider $y=\i$. At this value, $Z_{\textrm{\tiny AD}}$
 trivially vanishes for systems of odd length. Indeed, it follows from $Z_{\text{\tiny AD}}(y=\i) = Z_{\text{\tiny AD}}(y^{-1}=-\i)$ and the equation
\begin{equation}
  Z_{\text{\rm \tiny AD}}(\i) = \langle \phi_{\text{\rm \tiny AD}}| e^{\i \pi M/2}|\phi_{\text{\rm \tiny AD}}\rangle = (-1)^N \langle \phi_{\text{\rm \tiny AD}}| e^{\i \pi M/2}(-1)^M|\phi_{\text{\rm \tiny AD}}\rangle = (-1)^N  Z_{\text{\tiny AD}}(-\i)
\end{equation}
that $Z_{\text{\tiny AD}}(\i)=0$ for $N=2n+1$. For $N=2n$ however, we obtain a generating function for a further refinement of the weighted enumeration of half-turn symmetric ASMs. It is obtained by assigning the weight $(-1)^m t^k$ to any half-turn symmetric ASM with $k$ orbits of negative entries and $m$ non-zero entries in the upper-left quadrant. We denote by $A^-_{\text{\rm \tiny HT}}(2N;t)$ the corresponding generating function. We show that for $y=\i$ the expression \eqref{eqn:sumrulead} becomes
\begin{equation}
 Z_{\text{\tiny AD}}(\i)=\frac{(-1)^nA^-_{\text{\rm \tiny HT}}(2N,x^2)}{A(N,x^2)}, \quad N=2n.
 \label{eqn:ZADHTASMGF2}
\end{equation}

Kuperberg \cite{kuperberg:02} established certain factorisation properties of the polynomials $A^\pm_{\text{\rm \tiny HT}}(2N,x^2)$ which imply that the expressions in \eqref{eqn:ZADHTASMGF} and \eqref{eqn:ZADHTASMGF2} are indeed polynomials in $x^2$.

\paragraph{Scalar product.}

Because the spin-chain Hamiltonians with diagonal and anti-diagonal twist have no obvious relation, the scalar product of the zero-energy states $|\phi_{\text{\tiny D}}\rangle$ and $|\phi_{\text{\tiny AD}}\rangle$ appears to be an unnatural quantity to consider. It trivially vanishes for odd $N$ because the two states belong to different eigenspaces of the $\mathbb Z_2$-symmetry operators. For even $N$ however, the scalar product is non-vanishing and given in terms of an interesting combinatorial quantity, which $t$-enumerates $2N\times 2N$ quarter-turn symmetric ASMs. We give an example of a quarter-turn symmetric ASM in the right panel of \cref{fig:asmhtqt}.
Each such matrix is assigned the weight $t^k$ where $k$ is the number of orbits of negative entries, or equivalently the number of negative entries in one quadrant.
Their generating function is denoted by $A_{\text{\rm \tiny QT}}(2N; t)$ and was found by Kuperberg \cite{kuperberg:02} to factorise as the product of two polynomials in $t$, $A_{\text{\rm \tiny QT}}^{(1)}(2N;t)$ and $A_{\text{\rm \tiny QT}}^{(2)}(2N;t)$.
\begin{theorem}
  \label{thm:scalarprod}
  For $N=2n$ sites, the scalar product of the zero-energy states for the diagonal and anti-diagonal twist is given by
  \begin{equation}
    \langle\phi_{\text{\rm \tiny D}}|\phi_{\text{\rm \tiny AD}}\rangle = A_{\text{\rm \tiny QT}}(2N; x^2)  = A_{\text{\rm \tiny QT}}^{(1)}(2N;x^2)A_{\text{\rm \tiny QT}}^{(2)}(2N;x^2).
    \label{eqn:scalarprodphi}
  \end{equation}
  Moreover,
  $A_{\text{\rm \tiny QT}}^{(1)}(2N;x^2)$ can be written as the following pfaffian:
  \begin{equation}
     A_{\text{\rm \tiny QT}}^{(1)}(2N;x^2) = \Pf{i,j=0}{N-1}
     \left((-1)^j \binom{j}{i} x^{j-i-1}-(-1)^i \binom{i}{j} x^{i-j-1}\right).
    \label{eqn:newPfaffianAQT}
  \end{equation}
\end{theorem}

To the best of our knowledge, the pfaffian formula \eqref{eqn:newPfaffianAQT} has not
appeared in the literature before. It allows us to derive an alternative formula for the scalar product $Z_{\textrm{\tiny AD}}(y=\i)$ considered above. We use the known factorisation \cite{kuperberg:02}
\begin{equation}
  A^-_{\text{\rm \tiny HT}}(2N;t) = (-t)^{N/2} A(N;t) 
  \big(A_{\text{\rm \tiny QT}}^{(1)}(2N;x^2)\big)^2, 
  \qquad N = 2n,
\end{equation}
in \eqref{eqn:ZADHTASMGF2} and obtain for $N=2n$ the expression
\begin{align}
  Z_{\textrm{\tiny AD}}(y=\i) &= (-x^2)^{N/2}
  \big(A_{\text{\rm \tiny QT}}^{(1)}(2N;x^2)\big)^2 \nonumber   \\
  &=(-x^2)^{N/2}
  \Det{i,j= 0}{N-1}
  \left((-1)^j \binom{j}{i} x^{j-i-1}-(-1)^i \binom{i}{j} x^{i-j-1}\right).
\end{align}

%%%%%%%%%
\subsection{Special components}
\label{sec:specialcomponents}
%%%%%%%%%

In this subsection, we present results for certain special components of the vectors $|\phi_{\text{\rm \tiny D}}\rangle$ and $|\phi_{\text{\rm \tiny AD}}\rangle$. 

\paragraph{Diagonal twist.} Some components of the zero-energy state for the diagonal twist are given by generating functions appearing in problems of weighted ASM enumeration. These results are stated in the next two theorems.
\begin{theorem}  \label{thm:phiDsimplecom}
  If the normalisation of the zero-energy state $|\phi_{\text{\rm \tiny D}}\rangle$ is adjusted as in \eqref{eqn:defpsiD} and \eqref{eqn:defPhiD}, then
  \begin{subequations}
  \begin{alignat}{3}
     & (\phi_{\text{\rm \tiny D}})_{\Uparrow\cdots\Uparrow\Downarrow\cdots\Downarrow} = A(n;x^2) &\qquad & \text{for}\quad N=2n,\\
     & (\phi_{\text{\rm \tiny D}})_{\Uparrow\cdots\Uparrow0\Downarrow\cdots\Downarrow}= A(n;x^2)&\qquad & \text{for}\quad N=2n+1.
  \end{alignat}
  \end{subequations}
\end{theorem}
This statement was proved in \cite{hagendorf:15} under certain assumptions about the spin-reversal properties of the vector. 
In \cref{prop:SRdiagonal},
we provide a proof of the spin-reversal invariance and thus of the theorem.

In \cite{hagendorf:13}, the form of two further components was conjectured to be given by the generating function $A_{\text{\rm \tiny V}}(2n+1;t)$ which $t$-enumerates the vertically symmetric alternating sign matrices (VSASMs) of size $(2n+1)\times (2n+1)$. There are three $5\times 5$ VSASMs:
\begin{equation}\label{eq:VSASMs5}
\left( \begin{array}{ccccc}
0 & 0 & + & 0 & 0 \\
+ & 0 & - & 0 & + \\
0 & 0 & + & 0 & 0 \\
0 & + & - & + & 0 \\
0 & 0 & + & 0 & 0
\end{array} \right),
\quad
\left( \begin{array}{ccccc}
0 & 0 & + & 0 & 0 \\
0 & + & - & + & 0 \\
+ & - & + & - & + \\
0 & + & - & + & 0 \\
0& 0 &+& 0& 0
\end{array} \right),
\quad
\left( \begin{array}{ccccc}
0 & 0 & + & 0 & 0 \\
0 & + & - & + & 0 \\
0 & 0 & + & 0 & 0 \\
+ & 0 & - & 0 & + \\
0 & 0 & + & 0 & 0
\end{array} \right).
\end{equation}
In any VSASM, the entries along the symmetry axis are fixed to the alternating sequence $+1,-1,\dots,+1,-1,+1$. They may therefore be disregarded for the $t$-enumeration. We assign the weight $t^k$ to a VSASM with $k$ negative entries in the $(2n+1) \times n$ 
submatrix to the left of the symmetry axis. $A_{\text{\rm \tiny V}}(2n+1;t)$ is defined as the sum of the weights of the  
VSASMs of size $(2n+1)\times (2n+1)$. For example, from \eqref{eq:VSASMs5} we find that $A_{\text{\rm \tiny V}}(5;t)=2+t$.
\begin{theorem}
  \label{thm:PhiDSpecCom}
  For the diagonal twist, the zero-energy state possesses the special components
  \begin{subequations}
  \begin{alignat}{3} 
    & (\phi_{\text{\rm \tiny D}})_{\Uparrow\Downarrow\cdots \Uparrow\Downarrow} = A_{\text{\rm \tiny V}}(2n+1;x^2) \qquad &&\text{for} \quad N=2n, \label{eqn:PhiDSpecCompEven}\\
    & (\phi_{\text{\rm \tiny D}})_{0\cdots 0} = x^n A_{\text{\rm \tiny V}}(2n+1;x^2) \qquad &&\text{for} \quad N=2n+1\label{eqn:PhiDSpecCompOdd},
  \end{alignat}
  \end{subequations}
  where
    \begin{equation}
      A_{\text{\rm \tiny V}}(2n+1;t) = 
      \Det{i,j=0}{n-1}\left(\sum_{k=0}^{n-1}\binom{i+j+1}{i+k+1}
      \binom{i+k+1}{2k+1}t^k\right).
      \label{eqn:AVexplicit}
    \end{equation}
\end{theorem}
Notice that \eqref{eqn:AVexplicit} provides an alternative expression to Robbins' \cite{robbins:00} conjectured determinant formula for $A_{\text{\rm \tiny V}}(2n+1;t)$, and seems not to have appeared in the literature before.

\paragraph{Anti-diagonal twist.} 
Certain components of the zero-energy state for the anti-diagonal twist can be expressed in terms of generating functions that enumerate $2n\times 2n$ ASMs with double U-turn boundary conditions (UUASMs). 
In these matrices, each pair of successive rows, $2j-1$ and $2j$ with $j=1,\dots, n$, read from left to right for the odd row and after a U-turn from right to left for the even row, obeys the same rules as a row of an ordinary ASM. The same holds for pairs of successive columns, with left/right replaced by bottom/top. Here, the rows and columns are labeled from $1$ to $2n$ starting from the top-left corner.
An example is shown in the left panel of \cref{fig:uuasmvhpasm}.
The weight for a UUASM is $t^k y^m z^{m'}$ where $k$ is the number of negative entries, whereas $m$ and $m'$ respectively count the number odd rows and odd columns with an odd number of non-zero entries. 
Following Kuperberg's notation, we denote by $A_{\text{\rm \tiny UU}}(4n; t,y,z)$ the corresponding generating function for $2n \times 2n$ UUASMs. He proved that \cite{kuperberg:02}
\begin{equation}
   A_{\text{\rm \tiny UU}}^{(2)}(4n; t,y,z)=\frac{A_{\text{\rm \tiny UU}}(4n; t,y,z)}{A_{\text{\rm \tiny V}}(2n+1; t)}
\end{equation}
is a polynomial in $t,y$ and $z$. In \cref{prop:prop510}, we provide an expression for $A_{\text{\rm \tiny UU}}^{(2)}(4n; t,y,z)$ in terms of the determinant of a matrix with polynomial entries in $t, y$ and $z$.
As stated in the next two theorems, this polynomial appears
 in the zero-energy state for the anti-diagonal twist for particular specifications of $y$ and $z$.
\begin{theorem}\label{thm:componentsAD1}
  The zero-energy state for the anti-diagonal twist has the special components
  \begin{subequations}
  \begin{alignat}{3}
  & (\phi_{\text{\rm \tiny AD}})_{0\cdots 0} = (-x)^n \tilde A^{(2)}_{\text{\rm \tiny UU}}(4n;x^2) && \qquad \text{for}\quad N=2n,\\
  &(\phi_{\text{\rm \tiny AD}})_{\Uparrow\Downarrow\cdots \Uparrow\Downarrow\Uparrow} = 
  (-1)^n
  \tilde A^{(2)}_{\text{\rm \tiny UU}}(4(n+1);x^2) && \qquad \text{for}\quad N=2n+1,\label{eq:ududucomp}
  \end{alignat}
  \end{subequations}
  where the polynomial\footnote{Kuperberg\cite{kuperberg:02} defines
  \begin{equation*}
    \tilde A^{(2)}_{\text{\rm \tiny UU}}(4n;t)=\frac{A_{\textrm{\tiny HT}}^+(4n;t)}{A(2n;t)A_{\text{\rm \tiny UU}}^{(2)}(4n;t,1,1)}
  \end{equation*}
  which is equivalent to \eqref{eqn:defA2tilde}.
    }
  \begin{equation}
  \tilde A^{(2)}_{\text{\rm \tiny UU}}(4n;t) = t^{-n} A^{(2)}_{\text{\rm \tiny UU}}(4n;t,y=-1,z=-1)
  \label{eqn:defA2tilde}
  \end{equation}
  is given by the determinant formula
  \begin{equation}
    \tilde A^{(2)}_{\text{\rm \tiny UU}}(4n;t) = \Det{i,j=0}{n-1}\left(\sum_{k=0}^{n-1} \binom{i+j}{i+k}\binom{i+k}{2k} t^k\right).
    \label{eq:AUU}
  \end{equation}
\end{theorem}

\begin{figure}[h]
\centering
\begin{tikzpicture} 
  \draw (0,0) node {$
  \left(
  \begin{array}{cccc}
    + & - & 0 & +\\
    0 & + & 0 & -\\
    0 & 0 & 0 & 0\\
    0 & 0 & 0 & +
  \end{array}
  \right.
  $};
  \draw (1,.9) arc [start angle = 0, end angle = 180, radius = .275cm];
  \draw (-.15,.9) arc [start angle = 0, end angle = 180, radius = .275cm];
  
  \draw (1.3,.2) arc [start angle = -90, end angle = 90, radius = .22cm];
  \draw (1.3,-.6) arc [start angle = -90, end angle = 90, radius = .22cm];
 
  \draw (6,0) node {$
     \left(
  \begin{array}{ccccccc}
    0 & 0 & 0 & + & 0 & 0 & 0\\
    0 & + & 0 & - & 0 & + & 0\\
    + & - & + & \ast & + & - & +\\
    0 & + & 0 & - & 0 & + & 0\\
     0 & 0 & 0 & + & 0 & 0 & 0
  \end{array}
  \right)
  $};
\end{tikzpicture}
\caption{\textit{Left:} a $4\times 4$ UUASM with $k=2$ negative entries, and $m=1$ odd rows and $m'=1$ odd columns with an odd number of non-zero entries. 
The half circles indicate the U-turns. \textit{Right:} the simplest VHPASM,
which has $k=0$ negative entries in its $2\times 3$ fondamental domain.}
\label{fig:uuasmvhpasm} 
\end{figure}

The polynomial 
\begin{equation}
  A_{\text{\rm \tiny VHP}}(4n+2;t) = \lim_{y\to \infty} y^{-n} A_{\text{\rm \tiny UU}}(4n; t,y,y^{-1})
\end{equation}
$t$-enumerates the class of $(4n+1) \times (4n+3)$ vertically and horizontally perverse alternating sign matrices
(VHPASMs) which were introduced by Kuperberg \cite{kuperberg:02}. 
These rectangular matrices are invariant under reflections about both the horizontal and vertical medians. Moreover, their rows and columns satisfy the same rules as those of ordinary ASMs.
They also have the special property that the central entry $\ast$ is $-1$ when read horizontally but $+1$ when read vertically. As a result of the reflection symmetries, the entries of the horizontal and vertical medians are respectively fixed to the sequences $+1, -1, \dots, +1, -1, +1$ of lengths $4n+3$ and $4n+1$.
These are ignored for the $t$-enumeration. The weight of a given matrix is $t^k$ if it has $k$ negative entries in the $2n\times (2n+1)$ upper-left quadrant.  The simplest VHPASM is shown in the right panel of \cref{fig:uuasmvhpasm}. 
\begin{theorem}  \label{thm:componentsAD2}
  For $N=2n$ sites, the zero-energy state of the Hamiltonian with anti-diagonal twist possesses the special component
  \begin{equation}
    (\phi_{\text{\rm \tiny AD}})_{\Uparrow\Downarrow\cdots \Uparrow\Downarrow} = \frac{A_{\text{\rm \tiny VHP}}(4n+2;x^2)}{A_{\text{\rm \tiny V}}(2n+1;x^2)}= A_{\text{\rm \tiny VHP}}^{(2)}(4n+2;x^2)
  \end{equation}
  where $A_{\text{\rm \tiny VHP}}^{(2)}(4n+2;t)$ is given by
\begin{equation}
  A_{\text{\rm \tiny VHP}}^{(2)}(4n+2;t)=\lim_{y\to \infty} y^{-n} A_{\text{\rm \tiny UU}}^{(2)}(4n; t,y,y^{-1}) = \Det{i,j=0}{n-1}\left(\,\sum_{k=0}^{n-1} \binom{i + j}{i+k} \binom{i +k+1}{2k}t^k \right).
\label{eq:AVHP}
  \end{equation}
\end{theorem}

To the best of our knowledge, the formulas \eqref{eq:AUU} and \eqref{eq:AVHP} have not appeared in the literature before.

%%%%%%%%%%%%%%%%%%%%%
%
\section{The nineteen-vertex model} 
\label{sec:transfermatrices}
%
%%%%%%%%%%%%%%%%%%%%%

In this section, we initiate the derivation of the results presented in \cref{sec:spinoneXXZ}. To this end, we employ the same strategy that was used in \cite{hagendorf:15} for the diagonal twist: We generalise the problem and investigate the inhomogeneous transfer matrix of the corresponding nineteen-vertex model. 

We briefly review the construction of the nineteen-vertex model through the fusion procedure and its transfer matrix in \cref{sec:Rmat+fusion,sec:tm}. In \cref{sec:SoV}, we discuss certain elements of the quantum separation of variables technique which are relevant to our problem. Using this technique, we prove the existence of a simple non-degenerate eigenvalue of the transfer matrix of the nineteen-vertex model in \cref{sec:simpleeig}. (A result regarding the non-degeneracy of the special eigenstate for the diagonal twist is also given.) We also construct the corresponding eigenvectors. Finally in \cref{sec:propev}, we discuss certain properties of these eigenvectors, in particular their homogeneous limit.
 
%%%%%%%%%
\subsection[$R$-matrices and fusion]{\texorpdfstring{$\bm{R}$}{R}-matrices and fusion} \label{sec:Rmat+fusion}
%%%%%%%%%

The $R$-matrix of the nineteen-vertex model is constructed from the $2 \times 2$ fusion of the elementary $R$-matrix of the six-vertex model. We use the usual notation
\begin{equation}
|{\uparrow}\rangle = \begin{pmatrix}1 \\ 0\end{pmatrix}, \qquad |{\downarrow}\rangle = \begin{pmatrix}0 \\ 1\end{pmatrix}
\end{equation}
for the canonical basis of $\mathbb C^2$. In the basis
 $\{|{\uparrow\uparrow}\rangle,|{\uparrow\downarrow}\rangle,|{\downarrow\uparrow}\rangle,|{\downarrow\downarrow}\rangle\}$ of $\mathbb C^2 \otimes \mathbb C^2$, the elementary $R$-matrix reads
\begin{equation}
R^{(1,1)}(z) = \begin{pmatrix} [qz] & 0 & 0 & 0 \\ 0 & [z] & [q] & 0 \\ 0 & [q] & [z] & 0 \\ 0 & 0 & 0 & [qz] \end{pmatrix}
\label{eqn:6vRmatrix}
\end{equation} 
where we use the short-hand notation
\begin{equation}
  [z]=z-z^{-1}.
\end{equation}
Using the fusion procedure for $R$-matrices \cite{kulish:81,kulish:82,kirillov:87}, one can construct a two-parameter family of $R$-matrices, $R^{(m,n)}(z) \in \textrm{End}(\mathbb C^{m+1} \otimes \mathbb C^{n+1})$ with $m,n \in \mathbb Z_+$, that are solutions to the Yang-Baxter equation on the tensor product
$V_1 \otimes V_2 \otimes V_3= \mathbb C^{m+1} \otimes \mathbb C^{n+1} \otimes \mathbb C^{p+1}$:
\begin{equation}\label{eq:YB}
R_{12}^{(m,n)}(z/w)R_{13}^{(m,p)}(z)R_{23}^{(n,p)}(w) = R_{23}^{(n,p)}(w)R_{13}^{(m,p)}(z) R_{12}^{(m,n)}(z/w).
\end{equation}
Here, the labels $i,j$ of $R^{(m,n)}_{ij}(z)$ indicate that it acts non-trivially on $V_i$ and $V_j$ and as the identity on the other factor of the tensor product.

The elementary $R$-matrix satisfies the inversion relation
\begin{equation}\label{eq:invrel}
R^{(1,1)}(z)R^{(1,1)}(1/z) = [q z] [q/z] \, \boldsymbol 1.
\end{equation}
For $z = q^{\pm 1}$, $R^{(1,1)}(z)$ is not invertible and is given by
\begin{equation}
R^{(1,1)}(q) = B P^+, \qquad R^{(1,1)}(1/q) = -2 [q] P^-,
\end{equation}
where $B = \textrm{diag}([q^2], 2 [q], 2[q], [q^2])$ and $P^+$ and $P^-$ are the projectors on the symmetric and anti-symmetric subspaces of $\mathbb C^2 \otimes \mathbb C^2$:
\begin{equation}
P^+ = \begin{pmatrix} 1&0&0&0 \\ 0&1/2&1/2&0\\ 0&1/2&1/2&0\\ 0&0&0&1\end{pmatrix}, \qquad P^- = \begin{pmatrix} 0&0&0&0 \\ 0&1/2&-1/2&0\\ 0&-1/2&1/2&0\\ 0&0&0&0\end{pmatrix}.
\end{equation}

To construct the $R$-matrix of the mixed $(m,n)=(1,2)$ model, we note that $R^{(1,1)}_{12}(z)R^{(1,1)}_{13}(qz)P^+_{23}$ is annihilated by $P^-_{23}$ from both the right and the left. We define
\begin{equation}\label{eq:R12123}
[q z] R^{(1,2)}_{1(23)}(z) = U_{23}R^{(1,1)}_{12}(z)R^{(1,1)}_{13}(qz) P^+_{23} U_{23}^{-1}
\end{equation}
where
\begin{equation}
U = \begin{pmatrix} 1 & 0 & 0 & 0 \\ 0 & \alpha & \alpha & 0 \\ 0 & 0 & 0 & 1 \\ 0 & \alpha & -\alpha & 0 \end{pmatrix}, \qquad \alpha = \frac12 \sqrt{\frac{[q^2]}{[q]}}.
\end{equation}
The matrix $U$ operates a change of basis of $\mathbb C^2 \otimes \mathbb C^2$ into its symmetric and antisymmetric subspaces.  From the above remark, $R^{(1,2)}_{1(23)}(z)$ is zero on the antisymmetric subspace. The $R$-matrix of the mixed model is obtained by projecting onto the symmetric subspace:
\begin{equation}
R^{(1,2)}(z) = Q_{23} R^{(1,2)}_{1(23)}(z) Q_{23}^t, \qquad 
Q = \begin{pmatrix} 1 & 0 & 0 & 0\\ 0 & 1 & 0 & 0\\ 0 & 0 & 1& 0\end{pmatrix}.
\end{equation}
We denote by
\begin{equation}
|{\Uparrow}\rangle = |{\uparrow\uparrow}\rangle, \qquad |{0}\rangle =\frac1{\sqrt2}(
|{\uparrow\downarrow}\rangle+|{\downarrow\uparrow}\rangle), \qquad |{\Downarrow}\rangle = |{\downarrow\downarrow}\rangle
\end{equation}
the canonical basis states of the symmetric subspace. In the basis $\{|{\uparrow \Uparrow}\rangle, |{\uparrow\!0}\rangle, |{\uparrow \Downarrow}\rangle, |{\downarrow \Uparrow}\rangle, |{\downarrow\!0}\rangle, |{\downarrow \Downarrow}\rangle\}$
of $\mathbb C^2 \otimes \mathbb C^3$, we have
\begin{equation}
R^{(1,2)}(z) = \begin{pmatrix}
[q^2 z] & 0 & 0 & 0 & 0 & 0 \\
0 & [q z] & 0 & \sqrt{[q][q^2]} & 0 & 0 \\
0 & 0 & [z] & 0 & \sqrt{[q][q^2]} & 0 \\
0 & \sqrt{[q][q^2]} & 0 & [z] & 0 & 0 \\
0 & 0 & \sqrt{[q][q^2]} & 0 & [q z] & 0 \\
0 & 0 & 0 & 0 & 0 & [q^2 z] \\
\end{pmatrix}.
\end{equation}
The weights of the corresponding ten-vertex model are read from $R^{(1,2)}(z/q)$ and are given in \cref{fig:10v}.

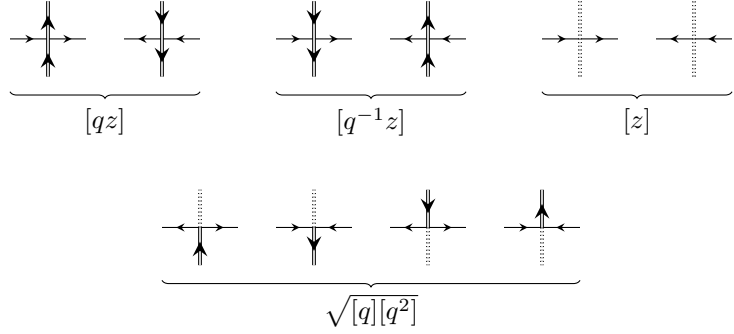
\begin{figure}[h] 
  \centering
 \begin{tikzpicture}
     \draw[postaction={on each segment={mid arrow}}] (0,0) -- (.5,0) -- (1,0);
     \draw[postaction={on each segment={mid arrow}},double] (.5,-.5) -- (.5,0) --(.
5,.5);
     \begin{scope}[xshift=1.5cm]
    \draw[postaction={on each segment={mid arrow}}] (1,0) -- (.5,0) -- (0,0);
     \draw[postaction={on each segment={mid arrow}},double] (.5,.5) -- (.5,0) -- (.
5,-.5) ;
    \end{scope}
    
    \begin{scope}[xshift=3.5cm]
       \draw[postaction={on each segment={mid arrow}}] (0,0) -- (.5,0) -- (1,0);
     \draw[postaction={on each segment={mid arrow}},double]  (.5,.5) -- (.5,0) --(.
5,-.5);
    \end{scope}
    \begin{scope}[xshift=5cm]
    \draw[postaction={on each segment={mid arrow}}] (1,0) -- (.5,0) -- (0,0);
     \draw[postaction={on each segment={mid arrow}},double] (.5,-.5)  -- (.5,0) --  
(.5,.5);

    \end{scope}
    
     \begin{scope}[xshift=7cm]
     \draw[postaction={on each segment={mid arrow}}] (0,0) -- (.5,0) -- (1,0);
     \draw[densely dotted ,double]  (.5,.5)  --(.5,-.5);
    \end{scope}
    \begin{scope}[xshift=8.5cm]
    \draw[postaction={on each segment={mid arrow}}] (1,0) -- (.5,0) -- (0,0);
     \draw[densely dotted,double] (.5,-.5)  -- (.5,0) --  (.5,.5);

    \end{scope}
    
    \begin{scope}[yshift=-2.5cm,xshift=2.cm]
     \draw[postaction={on each segment={mid arrow}}] (.5,0) -- (0,0);
     \draw[postaction={on each segment={mid arrow}}] (.5,0) -- (1,0);  
     \draw[densely dotted,double] (.5,0)  -- (.5,.5);
     \draw[postaction={on each segment={mid arrow}},double] (.5,-.5) -- (.5,0);  
     
     \begin{scope}[xshift=1.5cm]
     \draw[postaction={on each segment={mid arrow}}] (0,0) -- (.5,0);
     \draw[postaction={on each segment={mid arrow}}] (1,0) -- (.5,0);  
     \draw[densely dotted,double] (.5,0)  -- (.5,.5);
     \draw[postaction={on each segment={mid arrow}},double] (.5,0) -- (.5,-.5);  
     \end{scope}
     
     \begin{scope}[xshift=3cm]
     \draw[postaction={on each segment={mid arrow}}]   (.5,0)-- (0,0) ;
     \draw[postaction={on each segment={mid arrow}}]   (.5,0)-- (1,0);  
     \draw[densely dotted,double] (.5,0)  -- (.5,-.5);
     \draw[postaction={on each segment={mid arrow}},double] (.5,.5) -- (.5,0);  
     \end{scope}
     
     \begin{scope}[xshift=4.5cm]
    \draw[postaction={on each segment={mid arrow}}] (0,0) -- (.5,0);
     \draw[postaction={on each segment={mid arrow}}] (1,0) -- (.5,0); 
     \draw[densely dotted,double] (.5,0)  -- (.5,-.5);
     \draw[postaction={on each segment={mid arrow}},double] (.5,0) -- (.5,.5);  
     \end{scope}

    \end{scope}
    \draw [decoration={brace,mirror,raise=0.7cm},decorate] (0,0) -- (2.5,0); 
    \draw [decoration={brace,mirror,raise=0.7cm},decorate,xshift=3.5cm] (0,0) -- 
(2.5,0); 
    \draw [decoration={brace,mirror,raise=0.7cm},decorate,xshift=7cm] (0,0) -- 
(2.5,0); 
    
    \draw (1.25,-1.1) node {$[qz]$};
    \draw [xshift=3.5cm] (1.25,-1.1) node {$[q^{-1}z]$};
    \draw [xshift=7cm] (1.25,-1.1) node {$[z]$};
    
    \draw [decoration={brace,mirror,raise=0.7cm},decorate] (2,-2.5) -- (7.5,-2.5); 
    \draw (4.75,-3.6) node {$\sqrt{[q][q^2]}$};
  \end{tikzpicture} 
 \caption{The vertex configurations of the ten-vertex model and their statistical weights.}
  \label{fig:10v}
\end{figure}

To obtain the $R$-matrix of the $(m,n)=(2,2)$ model, we apply the fusion procedure a second time by defining
\begin{equation}
[z][qz]R^{(2,2)}_{(12)(34)}(z) = U_{12} R^{(1,2)}_{2(34)}(z/q) R^{(1,2)}_{1(34)}(z)P^+_{12}U^{-1}_{12}.
\end{equation}
We again project onto the symmetric subspaces:
\begin{equation}
R^{(2,2)}(z) = Q_{12} Q_{34} R^{(2,2)}_{(12)(34)}(z) Q_{12}^t Q_{34}^t.
\end{equation}
In the basis $\{|{\Uparrow \Uparrow}\rangle, |{\Uparrow\!0}\rangle, |{\Uparrow \Downarrow}\rangle, |{0\!\Uparrow}\rangle, |{0 0}\rangle, |{0\!\Downarrow}\rangle, |{\Downarrow \Uparrow}\rangle, |{\Downarrow\!0}\rangle, |{\Downarrow \Downarrow}\rangle\}$, it reads
\begin{equation}
R^{(2,2)}(z) = \begin{pmatrix}
p_1 & 0 & 0 & 0 & 0 & 0 & 0 & 0 & 0 \\
0 & p_2 & 0 & p_3 & 0 & 0 & 0 & 0 & 0 \\
0 & 0 & p_4 & 0 & p_5 & 0 & p_6 & 0 & 0 \\
0 & p_3 & 0 & p_2 & 0 & 0 & 0 & 0 & 0 \\
0 & 0 & p_5 & 0 & p_7 & 0 & p_5 & 0 & 0 \\
0 & 0 & 0 & 0 & 0 & p_2 & 0 & p_3 & 0 \\
0 & 0 & p_6 & 0 & p_5 & 0 & p_4 & 0 & 0 \\
0 & 0 & 0 & 0 & 0 & p_3 & 0 & p_2 & 0 \\
0 & 0 & 0 & 0 & 0 & 0 & 0 & 0 & p_1
\end{pmatrix}, \qquad
\begin{array}{l}
p_1 = [qz][q^2z] \\[0.1cm]
p_2 = [z][qz] \\[0.1cm]
p_3 = [q^2][qz] \\[0.1cm]
p_4 = [z/q][z] \\[0.1cm]
p_5 = [q^2][z] \\[0.1cm]
p_6 = [q][q^2] \\[0.1cm]
p_7 = [z][qz]+[q][q^2].
\end{array}
\end{equation}
The weights of the nineteen possible vertices can be
read from the non-zero matrix entries of $R^{(2,2)}(z)$. One can check that  \eqref{eq:YB} is satisfied for all admissible combinations of $R^{(1,1)}(z)$, $R^{(1,2)}(z)$ and $R^{(2,2)}(z)$.

We note that the projection onto the symmetric subspaces by the application of $Q$ is in many cases optional, 
 meaning that calculations made in $(\mathbb C^3)^{\otimes N}$ can often be performed in $(\mathbb C^2 \otimes \mathbb C^2)^{\otimes N}$ instead.
This will in particular enable us to resort to known results for the six-vertex model in \cref{sec:sumrules,sec:specialcomps}.

%%%%%%%%%
\subsection{Transfer matrices}
\label{sec:tm}
%%%%%%%%%

In inhomogeneous vertex models, inhomogeneity parameters $w_j$ are associated to each factor $V_j$ of the tensor product $V = \bigotimes_{j=1}^N V_j$. The monodromy matrix of the ten-vertex model is defined as
\begin{equation}
\mathcal T_a = R^{(1,2)}_{a,N}(q^{-1}z/{w_N})R^{(1,2)}_{a,N-1}(q^{-1}z/{w_{N-1}})\cdots R^{(1,2)}_{a,1}(q^{-1}z/{w_1}) = \begin{pmatrix}\mathcal A(z) & \mathcal B(z) \\ \mathcal C(z) & \mathcal D(z)\end{pmatrix}_{\!a} ,
\end{equation}
where $a$ labels the auxiliary space $\mathbb C^2$. The monodromy matrix is therefore an endomorphism of $\mathbb C^2 \otimes (\mathbb C^3)^{\otimes N}$.
The transfer matrix of the ten-vertex model is defined as
\begin{equation}
T^{(1)}(z) = \textrm{tr}_a\Big( \Omega^{(1)}_a \mathcal T_a\Big) =\left\{
\begin{array}{ll} 
\mathcal \i\big(\mathcal A(z) - \mathcal D(z)\big) & \textrm{(diagonal twist),} \\[0.2cm]
\mathcal \i \big(\mathcal B(z) + \mathcal C(z)\big) & \textrm{(anti-diagonal twist),} \\
\end{array}\right.
\end{equation}
where $\textrm{tr}_a$ denotes the trace over the auxiliary space. The operator $\Omega^{(1)}$ is a twist operator on $\mathbb C^2$:
\begin{equation}
\Omega^{(1)} = \begin{pmatrix} \i & 0 \\ 0 & -\i \end{pmatrix} \quad \text{(diagonal twist)},
\qquad 
\Omega^{(1)} = \begin{pmatrix} 0 & \i \\ \i & 0 \end{pmatrix} \quad \text{(anti-diagonal twist)}.
\end{equation}

The transfer matrix for the nineteen-vertex model is similarly constructed as
\begin{equation}
T^{(2)}(z) = \textrm{tr}_a\Big( \Omega^{(2)}_a R^{(2,2)}_{a,N}(z/{w_N})R^{(2,2)}_{a,N-1}(z/{w_{N-1}})\cdots R^{(2,2)}_{a,1}(z/{w_1})\Big) \in \textrm{End}\big((\mathbb C^3)^{\otimes N}\big),
\end{equation}
where the auxiliary space is now $\mathbb C^3$. The twist operator $\Omega^{(2)}$ is constructed as
\begin{equation}
\Omega^{(2)} = Q\, U(\Omega^{(1)} \otimes \Omega^{(1)}) U^{-1} Q^t
\end{equation}
and explicitly reads:
\begin{equation}
\Omega^{(2)} = \left(\begin{matrix} -1 & 0 & 0 \\ 0 & 1 & 0  \\ 0 & 0 & -1 \end{matrix}\right) \quad \text{(diagonal twist)},
\quad 
\Omega^{(2)} = -\left(\begin{matrix} 0 & 0 & 1 \\ 0 & 1 & 0 \\ 1 & 0 & 0 \end{matrix}\right) \quad \text{(anti-diagonal twist)}.
\label{eqn:twistmatrices}
\end{equation}
The terminology \textit{diagonal} and \textit{anti-diagonal} for the twist is inspired from the structure of these twist matrices.
For both twists, the transfer matrix of the nineteen-vertex model can be written in terms of $T^{(1)}$:
\begin{equation} 
  T^{(2)}(z)=T^{(1)}(z)T^{(1)}(qz)-a(qz)d(z)\, \boldsymbol {1}, 
  \qquad a(z) = \prod_{i=1}^N [q z/w_i], \quad d(z) = \prod_{i=1}^N [q^{-1} z/w_i].
  \label{eqn:fusion}
\end{equation}

Let us write the dependence of the inhomogeneity parameters in the transfer matrix as $T^{(i)}(z) = T^{(i)}(z|w_1,\dots w_N)$.
As a consequence of the Yang-Baxter equation \eqref{eq:YB}, the matrices $\mathcal A(z)$, $\mathcal B(z)$, $\mathcal C(z)$ and $\mathcal D(z)$ satisfy a set of quadratic relations. Crucially, these relations imply the commutation of transfer matrices with identical inhomogeneity parameters:
\begin{equation}
\label{eq:CommutationOfTs}
[T^{(i)}(z|w_1, \dots, w_N),T^{(j)}(z'|w_1, \dots, w_N)]=0, \qquad i,j \in \{1,2\},
\end{equation} 
for all $z,z'$.

If the $w_j$ are all set to the same value, the model is said to be {\it homogeneous}. Because of the obvious property $T^{(i)}(\lambda z|\lambda w_1, \dots, \lambda w_N)=T^{(i)}(z| w_1, \dots, w_N)$ for $\lambda \neq 0$, for the homogeneous limit, we will specialise the $w_j$ to $1$ without loss of generality. 
For both twists, the Hamiltonian \eqref{eqn:ham} can be expressed in terms of the homogeneous $T^{(2)}(z)$ matrix as
\begin{equation}
H = N + \frac{[q^2]}2 \frac{\diff}{\diff z} \log T^{(2)}(z|w_1=1,w_2=1,\dots,w_N=1) \Big|_{z=1} 
\label{eqn:hamfromT}
\end{equation}
where we recall that
\begin{equation}
x = q+q^{-1}.
\end{equation}
In terms of the twist matrices, the boundary conditions \eqref{eqn:twists} for the Hamiltonian can be written as
\begin{equation}
  s^a_{N+1} = \Omega^{(2)}_1 s^a_{1} \Omega^{(2)}_1, \quad a=1,2,3.
\end{equation}

The zero-energy state of the Hamiltonian turns out to be the homogeneous limit of a special eigenstate of the transfer matrices $T^{(1)}(z)$ and $T^{(2)}(z)$ whose properties will be investigated in the next sections.

So far we have considered $q$ to be a generic complex number. However, the specialisations $q=q_c$ where $(q_c)^4 = 1$ require closer scrutiny. At these values, the $R$-matrix of the ten-vertex model is diagonal, resulting in $T^{(1)}(z)$ being trivial for both twists. Nevertheless, the Hamiltonian \eqref{eqn:ham} is non-trivial at these values and is obtained from \eqref{eqn:hamfromT} by taking the limit $q\rightarrow q_c$. In \cref{sec:homog.definitions}, the special eigenstate of $H$ at $q = q_c$ will similarly be constructed by taking the proper limit of the generic construction.

\paragraph{Diagrammatic calculus.}
Some of the computations presented in later sections are performed using diagrams. In particular, the graphical representations of $\mathcal B(z)$ and $\mathcal C(z)$ are
\begin{equation}
\mathcal B(z) = \ \
\begin{tikzpicture}[baseline=-0.1cm] 
     \foreach \x in {0,.5,2}
     {  
       \draw[double] (\x,-0.35) -- (\x,0.35);
     }
       \draw[yshift=0 cm] (1.25,-.25) node {$\cdots$};
       \draw[yshift=0.5 cm] (1.25,-.25) node {$\cdots$};       
            
     \draw[postaction={on each segment={mid arrow}},yshift=0 cm] (0,0)--(-0.5,0);
     \draw[postaction={on each segment={mid arrow}},yshift=0 cm] (2.0,0)--(2.5,0);
     \draw[yshift=0 cm]  (0,0)--(2.0,0);
     
     \draw (-.5,0) node [left] {$z$};
     \draw (0,-.65) node {$w_1$};
     \draw (.6,-.65) node  {$w_2$};
      \draw (2,-.65) node {$w_{N}$};
    \end{tikzpicture}\ ,
    \qquad
    \mathcal C(z) = \ \
\begin{tikzpicture}[baseline=-0.1cm] 
     \foreach \x in {0,.5,2}
     {  
       \draw[double] (\x,-0.35) -- (\x,0.35);
     }
       \draw[yshift=0 cm] (1.25,-.25) node {$\cdots$};
       \draw[yshift=0.5 cm] (1.25,-.25) node {$\cdots$};       
            
     \draw[postaction={on each segment={mid arrow}},yshift=0 cm] (-0.5,0)--(0,0);
     \draw[postaction={on each segment={mid arrow}},yshift=0 cm] (2.5,0)--(2.0,0);
     \draw[yshift=0 cm]  (0,0)--(2.0,0);
     
     \draw (-.5,0) node [left] {$z$};
     \draw (0,-.65) node {$w_1$};
     \draw (.6,-.65) node  {$w_2$};
      \draw (2,-.65) node {$w_{N}$};
    \end{tikzpicture}\, .
\end{equation}
The partitions functions of the ten-vertex models on a square lattice can therefore be evaluated either by hand, that is by explicitly writing out 
the possible configurations and summing the resulting weights, or algebraically by computing a corresponding scalar product. By convention, we choose the left-to-right order of composition to correspond to a diagram drawn from top to bottom. Here is an example of a partition function on a $3\times 4$ inhomogeneous lattice with specific boundary conditions:
\begin{equation}
\langle{0 \Uparrow\, \Downarrow 0} |\mathcal B(z_3) \mathcal C(z_2) \mathcal B(z_1)|{\Uparrow 0 \Uparrow\, \Downarrow}\rangle= \ \
\begin{tikzpicture}[baseline=0.4cm] 

       \draw[double] (0,-0.5) -- (0,1);
       \draw[double] (0.5,-0.0) -- (0.5,1.5);
       \draw[double] (1.0,-0.5) -- (1.0,1.5);
       \draw[double] (1.5,-0.5) -- (1.5,1.0);
\draw[double,postaction={on each segment={mid arrow}}] (0.0,-0.45) -- (0.0,0);
        \draw[double,densely dotted] (0.5,0) -- (0.5,-0.5);
\draw[double,postaction={on each segment={mid arrow}}] (1.0,-0.45) -- (1.0,0);
\draw[double,postaction={on each segment={mid arrow}}] (1.5,-0.25) -- (1.5,-0.45);

     \draw[double,densely dotted] (0,1) -- (0,1.5);
     \draw[double,postaction={on each segment={mid arrow}}] (0.5,1.1) -- (0.5,1.5);
     \draw[double,postaction={on each segment={mid arrow}}] (1.0,1.4) -- (1.0,1.0);
     \draw[double,densely dotted] (1.5,1) -- (1.5,1.5);
            
     \foreach \y in {0,0.5,1.0}
     {
     \draw[yshift=\y cm]  (0,0)--(1.8,0);
     }

     \draw[postaction={on each segment={mid arrow}},yshift=0 cm] (0,0)--(-0.5,0);
     \draw[postaction={on each segment={mid arrow}},yshift=0 cm] (1.5,0)--(2.0,0);

     \draw[postaction={on each segment={mid arrow}},yshift=0 cm] (-0.5,0.5)--(0,0.5);
     \draw[postaction={on each segment={mid arrow}},yshift=0 cm] (2.0,0.5)--(1.5,0.5);
     
     \draw[postaction={on each segment={mid arrow}},yshift=0 cm] (0,1)--(-0.5,1);
     \draw[postaction={on each segment={mid arrow}},yshift=0 cm] (1.5,1)--(2.0,1);
     
     \draw (-.5,0) node [left] {$z_1$};
     \draw (-.5,0.5) node [left] {$z_2$};
     \draw (-.5,1) node [left] {$z_3$};
     \draw (0,-.75) node {$w_1$};
     \draw (.5,-.75) node  {$w_2$};
     \draw (1.0,-.75) node {$w_3$};
     \draw (1.5,-.75) node {$w_4$};
    \end{tikzpicture} \ \ .
\end{equation}

\paragraph{Spectral degeneracies.} We are now in a position to discuss the spectral degeneracies in the spectrum of the Hamiltonian stated in \cref{eq:spectra.diag,eq:spectra.antidiag}. We focus on the anti-diagonal twist and first examine the transfer matrices. Notice the simple properties 
\begin{equation}\label{eq:comms}
\{T^{(1)}(z),(-1)^M\} = 0, \qquad [T^{(2)}(z),(-1)^M] = 0.
\end{equation}
Suppose that $|\psi\rangle$ is an eigenvector of $T^{(1)}(z)$: $T^{(1)}(z)|\psi\rangle=\theta^{(1)}(z)|\psi\rangle$. The anticommutation relation implies that $T^{(1)}(z)\big((-1)^M|\psi\rangle\big)=-\theta^{(1)}(z)\big((-1)^M|\psi\rangle\big)$. 
 If $|\psi\rangle$ is also an eigenvector of $(-1)^M$, then $\theta^{(1)}(z)=0$. 
In this case, \eqref{eqn:fusion} implies that
$|\psi\rangle$ solves $T^{(2)}(z)|\psi\rangle=\theta^{(2)}(z)|\psi\rangle$ with the special eigenvalue 
\begin{equation} 
  \theta^{(2)}(z) =
   -a(q z)d(z) = 
   -\prod_{i=1}^N [q^{-1}z/w_i][q^2z/w_i].
  \label{eqn:specialev}
\end{equation}

Conversely, if $\theta^{(1)}(z)$ is non-zero, then the vectors $|\psi\rangle$ and $ (-1)^M|\psi\rangle$ are linearly independent and, from \eqref{eq:comms}, are both eigenvectors of the transfer matrix $T^{(2)}(z)$ with the same eigenvalue $\theta^{(2)}(z)$.
The same holds for the linear combinations $|\psi_\pm\rangle = (1\pm(-1)^M)|\psi\rangle/2$, which are the projections of $|\psi\rangle$ onto the eigenspaces of $(-1)^M$ with eigenvalue $\pm 1$. $T^{(1)}(z)$ maps between the two vectors: $T^{(1)}(z)|\psi_\pm\rangle = \theta^{(1)}(z)|\psi_\mp\rangle$. When taking the homogeneous limit, this doublet construction implies the degeneracies observed in the non-zero part of the Hamiltonian spectrum. Finally, the discussion for the model with diagonal twist is completely analogous: 
it suffices to consider the $\mathbb Z_2$ operator $F$ instead (as it anticommutes with $T^{(1)}(z)$ in this case). This ends the proofs of \cref{eq:spectra.diag,eq:spectra.antidiag}.

For the anti-diagonal twist, we will show in \cref{sec:simpleeig} that $T^{(2)}(z)$ indeed has 
the eigenvalue \eqref{eqn:specialev} in the generic inhomogeneous case
and will construct the corresponding eigenvector. In the homogeneous limit, it becomes a zero-energy state of the spin-chain Hamiltonian with anti-diagonal twist.

%%%%%%%%%
\subsection{Quantum separation of variables basis}
\label{sec:SoV}
%%%%%%%%%

The canonical basis of the Hilbert space $(\mathbb C^3)^{\otimes N}$ is not convenient in order to construct eigenstates of the transfer matrices $T^{(1)}(z)$ and $T^{(2)}(z)$ for the anti-diagonal twist. 
Below, we present another basis in which the operator $\mathcal D(z)$ acts diagonally. This basis was previously used to diagonalise transfer matrices of fused vertex models with anti-periodic boundary conditions by the quantum separation of variables technique. This was done for $q = 1$  in \cite{niccoli:13} and for arbitrary $q$ in \cite{niccoli:15}. We follow closely their definitions and arguments.

The new basis states are denoted by $||\bm h\rrangle \equiv ||h_1,\dots,h_N\rrangle$ where the \emph{heights} $h_j$, $j=1,\dots,N$, are in $\{0,1,2\}$.
These states are defined by\footnote{Here and in the following, it is understood that products of the form $\prod_{h=0}^{\ell-1}f(h)$ evaluate to~$1$ for $\ell=0$.}
\begin{equation}
  ||\bm h\rrangle = \prod_{j=1}^N \prod_{h=0}^{h_j-1}\left(\frac{\mathcal B(q^{1-h}w_j)}{a(q^{1-h}w_j)}\right)|\wedge\rangle
  , \qquad |\wedge\rangle = |{\Uparrow \cdots \Uparrow}\rangle . 
  \label{eqn:defhr}
\end{equation}
The operator $\mathcal B(z)$ lowers the magnetisation by one, and thus we find:
\begin{equation}
   M||\bm h\rrangle = \sum_{j=1}^N(1- h_j)||\bm h\rrangle.
   \label{eq:Mh}
\end{equation} 
Expressions for the action of the operators $\mathcal B(z)$, $\mathcal C(z)$ and $\mathcal D(z)$ in the height basis are known. These are found by using the quadratic relations satisfied by $\mathcal B(z), \mathcal C(z)$ and $\mathcal D(z)$ which are standard \cite{korepin:93}. One obtains \cite{niccoli:13,niccoli:15}
\begin{subequations}
\label{eqn:bcd}
\begin{align}
  \mathcal D(z)||\bm h\rrangle &= \left(\prod_{j=1}^N [q^{-1+h_j}z/w_j]\right)||\bm h\rrangle,\\
  \mathcal B(z)||\bm h\rrangle & =\sum_{j=1}^N a(q^{1-h_j}w_j)\left(\prod_{k\neq j}^N\frac{[q^{-1+h_k}z/w_k]}{[q^{h_k-h_j}w_j/w_k]}\right)\tau_j^+||\bm h\rrangle,\\
  \mathcal C(z)||\bm h\rrangle & = -\sum_{j=1}^N d(q^{1-h_j}w_j)\left(\prod_{k\neq j}^N\frac{[q^{-1+h_k}z/w_k]}{[q^{h_k-h_j}w_j/w_k]}\right)\tau_j^-||\bm h\rrangle,
\end{align}
\end{subequations}
where $\tau^\pm_j||\bm h\rrangle = ||h_1,\dots,h_{j-1},h_j\pm 1,h_{j+1},\dots, h_N\rrangle$ 
if the new height at $j$, $h_j \pm 1$, takes values in $\{0,1,2\}$, and $\tau^\pm_j||\bm h\rrangle =0$ otherwise. The action of the operator $\mathcal A(z)$ can be obtained from the relations \eqref{eqn:bcd} and the so-called quantum determinant relation, but it won't be needed in what follows.

The dual height states are defined by
\begin{equation}
  \llangle \bm h|| = \langle \wedge | \prod_{j=1}^N \prod_{h=0}^{h_j-1}\left(\frac{\mathcal C(q^{1-h}w_j)}{d(q^{-h}w_j)}\right).
  \label{eqn:defhl}
\end{equation}
The left action of the operators $\mathcal B(z),\mathcal C(z)$, and $\mathcal D(z)$ on the dual states is given by formulas similar to those of the right actions, and can be derived along the lines of \cite{niccoli:13}:
\begin{subequations}
\label{eqn:bcdtransposed}
\begin{align}
  \llangle \bm h||\mathcal D(z) &=  \llangle \bm h||\left(\prod_{j=1}^N [q^{-1+h_j}z/w_j]\right),\\
  \llangle \bm h||\mathcal B(z) & =-\sum_{j=1}^N a(q^{2-h_j}w_j)\left(\prod_{k\neq j}^N\frac{[q^{-1+h_k}z/w_k]}{[q^{h_k-h_j}w_j/w_k]}\right)\llangle \bm h||\tau_j^+,\\
  \llangle \bm h||\mathcal C(z) & = \sum_{j=1}^N d(q^{-h_j}w_j)\left(\prod_{k\neq j}^N\frac{[q^{-1+h_k}z/w_k]}{[q^{h_k-h_j}w_j/w_k]}\right)\llangle \bm h||\tau_j^-,
\end{align}
\end{subequations}
where $\llangle \bm h||\tau^\pm_j= \llangle h_1,\dots,h_{j-1},h_j\mp 1,h_{j+1},\dots, h_N||$ if 
the new height at $j$, $h_j \mp 1$, takes values in $\{0,1,2\}$, and $\llangle \bm h||\tau^\pm_j=0$ otherwise.

The height states and their duals respectively form left- and right-eigenbases of the operator $\mathcal D(z)$. Its eigenvalues are all distinct for generic values of the parameters. This implies the completeness relation
\begin{equation}
  \sum_{\bm h} \frac{||\bm h\rrangle \llangle \bm h||}{\llangle \bm h|| \bm h\rrangle}= \bm{1}.
  \label{eqn:completeness}
\end{equation}
The normalisation factor can be obtained from the scalar product between a left- and right-vector:
\begin{equation}
  \llangle \bm h||\bm h'\rrangle = (-1)^{h_1+\dots+h_N}\prod_{j=1}^N \delta_{h_j,h_j'} \prod_{1\le j<k \le N} \frac{[w_j/w_k]}{[q^{h_k-h_j}w_j/w_k]}.
  \label{eqn:scalarprod}
\end{equation}

%%%%%%%%%
\subsection{Simple eigenvalues}
\label{sec:simpleeig}
%%%%%%%%%

\paragraph{Diagonal twist.} For the diagonal twist, the diagonalisation of $T^{(1)}(z)$ and $H$ can be approached using the algebraic Bethe ansatz \cite{korepin:93}. In the magnetisation sector $M = N-n$, the eigenvalues of $T^{(1)}(z)$ and their corresponding Bethe eigenstates are given by
\begin{subequations}\label{eq:BAsol}
\begin{equation}\label{eq:BAsoleig}
\theta^{(1)}(z|w_1, \dots, w_N) = \i \bigg(a(z) \prod_{j=1}^n \frac{[q z_j/z]}{[z_j/z]} - d(z) \prod_{j=1}^n \frac{[q z/z_j]}{[z/z_j]}\bigg),
\end{equation}
\begin{equation}
\langle \psi(z_1, \dots, z_n) | = \langle \wedge|\prod_{j=1}^n \mathcal{C}(z_j), \qquad |\psi(z_1, \dots, z_n)\rangle = \prod_{j=1}^n \mathcal{B}(z_j)|\wedge\rangle.
\end{equation}
\end{subequations}
These are written in terms of the Bethe roots $z_j$ which satisfy the Bethe ansatz equations:
\begin{equation}\label{eq:BAeq}
\prod_{i=1}^N \frac{[q z_j/w_i]}{[q^{-1} z_j/w_i]} = - \prod_{i \neq j}^n \frac{[q z_j/z_i]}{[q^{-1}z_j/z_i]}, \qquad j = 1, \dots, n.
\end{equation}
We note that if $z_1,\dots,z_n$ is a solution to the Bethe equations, then $\epsilon_1z_1,\dots,\epsilon_n z_n$ with arbitrary signs $\epsilon_1=\pm1,\dots,\epsilon_n=\pm 1$ is another solution. However, since $\mathcal B(-z)=(-1)^{N-1}\mathcal B(z)$ and $\mathcal C(-z)=(-1)^{N-1}\mathcal C(z)$, any choice of the signs produces, up to a factor, the same Bethe state. Hence, any two solutions to the Bethe equations which differ only by possible signs (and a permutation of the indices) can be identified.

As shown in \cite{hagendorf:15},
the transfer matrix $T^{(1)}(z)$ has the simple eigenvalue $\theta^{(1)}(z)=0$, 
which is obtained by this simple solution to the Bethe ansatz equations:
\begin{equation}\label{eq:BAsolution}
n=N, \qquad z_j = w_j, \quad j = 1, \dots, n.
\end{equation}
The corresponding left- and right-eigenvector are thus given by\footnote{For generic inhomogeneity parameters, the left eigenvector is not simply the transpose of the right eigenvector, as we shall see in \cref{sec:propev}. We nonetheless use the notation $\langle \psi_{\textrm{\tiny D}}|$ for simplicity. In this case, $\langle \psi_{\textrm{\tiny D}}|\psi_{\textrm{\tiny D}}\rangle$ is not a true norm. The same applies to the anti-diagonal twist.}
\begin{equation}  \langle \psi_{\textrm{\tiny D}}| =\langle \wedge|\prod_{j=1}^N \mathcal{C}(w_j), \qquad |\psi_{\textrm{\tiny D}}\rangle = \prod_{j=1}^N \mathcal{B}(w_j)|\wedge\rangle.
  \label{eqn:defpsiD}
\end{equation}
These are also eigenstates of $T^{(2)}$ with the eigenvalue $\theta^{(2)}$ given in \eqref{eqn:specialev}.
Here we present a new result about the non-degeneracy of this eigenvalue. It presupposes the completeness of the Bethe ansatz, which we expect holds (at least) for generic values of the $w_j$ 
and of $q$.
\begin{proposition}
  \label{prop:psidunique}
  For generic values of the $w_j$ and of $q$, if the Bethe ansatz is complete, then $\langle \psi_{\text{\rm \tiny D}}|$ and $|\psi_{\text{\rm\tiny D}}\rangle$ are respectively the unique left- and right-eigenstates of $T^{(1)}(z)$ with eigenvalue $\theta^{(1)}(z)=0$.
  \begin{proof}
    The assumption about completeness implies that every eigenstate of $T^{(1)}(z)$ is of the form \eqref{eq:BAsol}, with each $z_j \in \mathbb C^\times$. We must therefore show using \eqref{eq:BAsoleig} that the only solution to $\theta^{(1)}(z)=0$ is \eqref{eq:BAsolution} up to permutations of the roots.
    Let us suppose that for $\theta^{(1)}(z)=0$ for $n$ Bethe roots $z_1,\dots, z_n$. It is convenient to introduce the centred Laurent polynomial
    \begin{equation}
      Q(z)=\prod_{j=1}^n [z/z_j].
    \end{equation}
    From \eqref{eq:BAsoleig} with $\theta^{(1)}(z)=0$, it follows that
    \begin{equation}
      a(qz)Q(z) = d(q z) Q(q^2z)
      \label{eqn:TQThetaZero}
    \end{equation}
    for all $z\in \mathbb C^\times$. 
   Because $d(\pm q w_j)=0$, the choice $z=\pm w_j$ for $j=1,\dots,N$ yields
	\begin{equation}
      		a(\pm qw_j)Q(\pm w_j) = d(\pm qw_j) Q(\pm q^2w_j) = 0.
    	\end{equation}
The factor $a(\pm qw_j)$ is non-zero for each $j=1,\dots,N$ and for generic values of the parameters. We conclude that $Q(z)$ has the $2N$ distinct zeroes $z=\pm w_1,\dots,\pm w_N$. 
The degree width of $Q(z)$ is therefore greater than or equal to $2N$, implying that $n \ge N$. (The degree width of a Laurent polynomial is the difference in degree of the leading and trailing terms.) We may thus write
    \begin{equation}
      Q(z) = \left(\prod_{j=1}^N[z/w_j]\right)g(z),
      \label{eqn:almostthere}
    \end{equation}
    where $g(z)$ is a centred Laurent polynomial of degree width $2(N-n)$. The substitution of \eqref{eqn:almostthere} into \eqref{eqn:TQThetaZero} leads to $g(z)= g(q^2 z)$, for all $z \in \mathbb C^\times$. The only solution to this equation is a constant and thus $n=N$.
Comparing the leading and trailing term of $Q(z)$, we conclude that this constant must be a sign, $g(z) = \epsilon$ with $\epsilon^2=1$. Hence
    \begin{equation}
      Q(z) = \prod_{j=1}^N[z/z_j]=\epsilon\prod_{j=1}^N[z/w_j].
    \end{equation}
    It follows that up to a permutation of the roots, $z_j = \epsilon_j w_j$ where $\epsilon_j^2=1$ and thus $\epsilon= \prod_{j=1}^N\epsilon_j$. Since any two solutions to the Bethe equations which differ only by signs may be identified, we may set $\epsilon_j=1$ for all $j=1,\dots,N$ and recover thus \eqref{eq:BAsolution}.
  \end{proof}
  \end{proposition}

\paragraph{Anti-diagonal twist.} In the rest of this section, we show that for the anti-diagonal twist,
 $T^{(1)}(z)$ also has the simple eigenvalue zero, for any values of the inhomogeneity parameters $w_j$.
We use the quantum separation of variables technique to construct explicitly the corresponding \emph{unique} 
 left- and right-eigenvector $\langle \psi_{\text{\rm \tiny{AD}}}|$ and $|\psi_{\text{\rm \tiny{AD}}}\rangle$.

\begin{proposition}\label{prop:psiAD.def}
  For generic values of the inhomogeneity parameters, the transfer matrix $T^{(1)}(z)$ with the anti-diagonal twist possesses the non-degenerate eigenvalue $\theta^{(1)}(z)=0$. The projections of the corresponding left- and right-eigenvector onto the separation of variables basis are given by
  \begin{equation}
    \langle \psi_{\text{\rm \tiny{AD}}}||\bm h\rrangle = \prod_{i,j=1}^N\frac{[q^{1-h_i}w_i/w_j]}{[q w_i/w_j]}, \quad \llangle \bm h||\psi_{\text{\rm \tiny AD}}\rangle =\prod_{i,j=1}^N\frac{[q^{2(h_i-1)}w_i/w_j]}{[q^{-2}w_i/w_j]}.
    \label{eqn:projections}
  \end{equation}
\end{proposition}
\begin{proof} Let us consider the left-eigenvalue problem. We need to show that there is (up to normalisation) a single solution to the equation
\begin{equation}\label{eq:AboveEquation}
  \langle \psi_{\text{\rm \tiny{AD}}}|T^{(1)}(z) = \i \langle \psi_{\text{\rm \tiny{AD}}}|\left(\mathcal B(z) + \mathcal C(z)\right)=0.
\end{equation}
We take the scalar product of this equation with an arbitrary basis state $||\bm h\rrangle$, which gives $\langle \psi_{\text{\rm \tiny{AD}}}|\mathcal B(z)||\bm h\rrangle+\langle \psi_{\text{\rm \tiny{AD}}}|\mathcal C(z)||\bm h\rrangle=0$, and apply the equations \eqref{eqn:bcd}. Because of the integrability of the model, see \eqref{eq:CommutationOfTs}, there exists an eigenbasis of $T^{(1)}(z)$, of which $\langle \psi_{\text{\rm \tiny{AD}}}|$ is an element, which is independent of $z$. We are therefore free to specialise $z$ in \eqref{eq:AboveEquation} to any value.
We notice that there is a great simplification for $z=q^{1-h_i}w_i$: 
We have
\begin{subequations}
\begin{align} 
  \mathcal B(z=q^{1-h_i}w_i)||\dots,h_i,\dots\rrangle & = a(q^{1-h_i}w_i)||\dots,h_i+1,\dots\rrangle,\\
  \mathcal C(z=q^{1-h_i}w_i)||\dots,h_i,\dots\rrangle & = -d(q^{1-h_i}w_i)||\dots,h_i-1,\dots\rrangle.
\end{align}
\end{subequations}
Hence we find for each $i=1,\dots,N$ the equations
\begin{equation}
  a(q^{1-h_i}w_i) \langle \psi_{\text{\tiny{AD}}}||\dots,h_i+1,\dots\rrangle = d(q^{1-h_i}w_i) \langle \psi_{\text{\tiny{AD}}}||\dots,h_i-1,\dots\rrangle \label{eq:rechh2}
\end{equation}
where it is understood that the expressions on either side of this equality are zero unless the corresponding height at $i$, $h_i + 1$ or $h_i-1$, is in $\{0,1,2\}$.
 This equation is trivial to solve for each~$i$: We find that $\langle \psi_{\text{\rm \tiny{AD}}}||\dots,
h_i=1
,\dots\rrangle =0$
 and that $\langle \psi||\dots,h_i=2,\dots\rrangle$ is easily expressed in terms of $\langle \psi||\dots,h_i=0,\dots\rrangle$. 
Hence, given $||\bm h\rrangle$, \eqref{eq:rechh2} expresses $\langle \psi_{\text{\rm \tiny{AD}}}||\bm h\rrangle$
as $\langle \psi_{\text{\rm \tiny{AD}}}||\bm {h'}\rrangle$ times a prefactor, where $h'_j \le h_j$ for $j= 1, \dots, N$ and $h'_j < h_j$ for one $j$. 
 This process can be repeated until all heights are zero, in which case we set $\langle \psi_{\text{\rm \tiny{AD}}}||\bm 0\rrangle =1$ by convention. 
 
 Up to this conventional normalisation, the solution is unique for generic values of the $w_j$. There are indeed specifications for the $w_j$
 for which the previous construction fails because some height states are either linearly dependent or simply undefined.
 
 For the right eigenvector, the reasoning is completely analogous with the exception that one needs to use \eqref{eqn:bcdtransposed}
 instead of \eqref{eqn:bcd}.
\end{proof}

The eigenvectors $\langle\psi_{\text{\tiny{AD}}}|$ and $|\psi_{\text{\tiny{AD}}}\rangle$ can respectively be reconstructed from the projections on the basis states $\llangle \bm h||$ and $||\bm h\rrangle$ by using the completeness relation \eqref{eqn:completeness}:
\begin{equation}\label{eq:explicitEV}
\langle\psi_{\text{\tiny{AD}}}| = \sum_{\bm h}  \frac{\langle\psi_{\text{\tiny{AD}}}||\bm h\rrangle }{\llangle \bm h|| \bm h\rrangle} \llangle \bm h||,
\qquad |\psi_{\text{\tiny{AD}}}\rangle = \sum_{\bm h} ||\bm h\rrangle\frac{ \llangle \bm h||\psi_{\text{\tiny{AD}}}\rangle}{\llangle \bm h|| \bm h\rrangle}.
\end{equation}
 It follows from \eqref{eqn:fusion} that they are eigenvectors of the transfer matrix $T^{(2)}(z)$ of the nineteen-vertex model. Just as for the diagonal twist \cite{hagendorf:15}, the corresponding eigenvalue is given in \eqref{eqn:specialev}

%%%%%%%%%
\subsection{Properties of the eigenvectors}
\label{sec:propev}
%%%%%%%%%

In this subsection, we collect a number of properties of the vectors $|\psi_{\text{\rm \tiny D}}\rangle$ and $|\psi_{\text{\rm \tiny AD}}\rangle$ which will be useful in the forthcoming sections. If necessary, we write out the explicit dependence of the inhomogeneity parameters, for instance $|\psi_{\text{\rm \tiny D}}\rangle=|\psi_{\text{\rm \tiny D}}(w_1,\dots,w_N)\rangle$ and $\mathcal B(z)=\mathcal B(z|w_1,\dots,w_N)$. 

\subsubsection{Operations involving the inhomogeneity parameters}

We start with two propositions about the behaviour of the special eigenvectors under reversal of the inhomogeneity parameters $w_1,\dots, w_N$ and/or the parameter $q$.

\begin{proposition} \label{prop:transposition}
The special left- and right-eigenvectors are related by
\begin{align}
 \langle\psi_{\text{\rm \tiny D}}(w_1,\dots,w_N)|= |\psi_{\text{\rm \tiny D}}(w_1^{-1},\dots,w_N^{-1})\rangle^t,\quad 
  \langle\psi_{\text{\rm \tiny AD}}(w_1,\dots,w_N)|= |\psi_{\text{\rm \tiny AD}}(w_1^{-1},\dots,w_N^{-1})\rangle^t.
\end{align}
\end{proposition}
\begin{proof}
For the diagonal twist, the proof was given in \cite{hagendorf:15}. It relies on the so-called crossing symmetry of $R^{(1,2)}(z)$,
\begin{equation}
  R^{(1,2)}(z)^{t_2}= - 
  (\sigma^2\otimes \bm 1)R^{(1,2)}(q^{-2}z^{-1})
(\sigma^2\otimes \bm 1),\quad \sigma^2 =\left(
  \begin{array}{cc}
    0 & - \i\\
    \i & 0
  \end{array}
  \right),
\end{equation}
where $t_2$ refers to transposition with respect to the factor $\mathbb C^3$ in $\mathbb C^2 \otimes \mathbb C^3$. 
From this property, we deduce that for both twists the transfer matrix and its transpose are related according to
\begin{equation}
  T^{(1)}(z|w_1,\dots,w_N) = (-1)^{N-1} T^{(1)}(z^{-1}|w_1^{-1},\dots,w_N^{-1})^t. 
\end{equation}
From now on, focus is on the anti-diagonal twist. We write out $|\psi_{\text{\tiny{AD}}}\rangle=|\psi_{\text{\tiny{AD}}}(w_1,\dots,w_N)\rangle$ and introduce the co-vector $\langle \varphi|= |\psi_{\text{\tiny{AD}}}(w_1^{-1},\dots,w_N^{-1})\rangle^t$. Applying it from the left to the equation leads to $\langle \varphi|T^{(1)}(z|w_1,\dots,w_N)=0$. Hence, $\langle \varphi|$ is a left null vector. The non-degeneracy of the eigenvalue zero implies that
\begin{equation}
  \langle \varphi|  = \lambda \langle \psi_{\text{\tiny AD}}(w_1,\dots,w_N)|
\end{equation}
for some $\lambda$. Taking the scalar product with $|\wedge\rangle$ gives
\begin{equation}
  \langle \varphi|\wedge\rangle = \langle \wedge|\psi_{\text{\tiny AD}}(w_1^{-1},\dots,w_N^{-1})\rangle = \lambda \langle \psi_{\text{\rm \tiny AD}}(w_1,\dots, w_N)|\wedge\rangle.
  \end{equation}
From \eqref{eqn:projections}, we infer that $\lambda=1$.
\end{proof}

\begin{proposition}\label{prop:qtoqinv}
The special eigenvectors have the property
\begin{subequations}
\begin{align}
  & |\psi_{\text{\rm \tiny D}}(w_1,\dots,w_N)\rangle\big|_{q\to q^{-1}}=|\psi_{\text{\rm \tiny D}}(w_1^{-1},\dots,w_N^{-1})\rangle,\\ 
  & |\psi_{\text{\rm \tiny AD}}(w_1,\dots,w_N)\rangle\big|_{q\to q^{-1}}=|\psi_{\text{\rm \tiny AD}}(w_1^{-1},\dots,w_N^{-1})\rangle.
\end{align}
\end{subequations}
\end{proposition}
\begin{proof}
  The proof relies on the simple property
  \begin{equation}
    \left.R^{(1,2)}(z)\right|_{q\to q^{-1}} = -(\sigma^3\otimes \bm 1)R^{(1,2)}(z^{-1})(\sigma^3\otimes \bm 1),
  \end{equation}
  which implies in particular that
  \begin{subequations}
  \begin{align}
    & \mathcal B(z|w_1,\dots,w_N)\big|_{q\to q^{-1}} = (-1)^{N-1}\mathcal B(z^{-1}|w_1^{-1},\dots,w_N^{-1}),\\
    & \mathcal C(z|w_1,\dots,w_N)\big|_{q\to q^{-1}} = (-1)^{N-1}\mathcal C(z^{-1}|w_1^{-1},\dots,w_N^{-1}).
  \end{align}
  \end{subequations}
  These two relations can directly be applied to \eqref{eqn:defpsiD} which proves the proposition for the diagonal twist. For the anti-diagonal twist, we take the sum of these relations and obtain thus
  \begin{equation}
    \left.T^{(1)}(z|w_1,\dots,w_N)\right|_{q\to q^{-1}} = (-1)^{N-1} T^{(1)}(z^{-1}|w_1^{-1},\dots,w_N^{-1}).
  \end{equation}
  We apply this equation to $|\psi_{\text{\rm \tiny AD}}(w_1,\dots,w_N)\rangle\big|_{q\to q^{-1}}$ and conclude that it is a null vector of $T^{(1)}(z^{-1}|w_1^{-1},\dots,w_N^{-1})$. As the null space is one-dimensional, we find $\left.|\psi_{\text{\rm \tiny AD}}(w_1,\dots,w_N)\rangle\right|_{q\to q^{-1}}=\lambda |\psi_{\text{\rm \tiny AD}}(w_1^{-1},\dots,w_N^{-1})\rangle$ for some $\lambda$. Upon taking the scalar product with $\langle\wedge|$,
  we conclude that $\lambda=1$.
\end{proof}

The next proposition establishes an exchange relation 
for special eigenvectors related by a permutation of their inhomogeneity parameters $w_1,\dots,w_N$.
In order to formulate it, we introduce the matrix
\begin{equation}
  \check R(z) = P R^{(2,2)}(z),
  \label{eqn:Rcheck}
\end{equation}
where $P$ is the permutation operator on $\mathbb C^3\otimes \mathbb C^3$: $P(| v\rangle \otimes | w \rangle) = | w \rangle \otimes | v \rangle$.

\begin{proposition}
\label{prop:exchange}
For $j=1,\dots,N-1$,
we have the relation
\begin{equation}
  \check R_{j,j+1}\left(\frac{w_j}{w_{j+1}}\right)|\psi(\dots,w_j,w_{j+1},\dots)\rangle= \left[\frac{q w_j}{w_{j+1}}\right]\left[\frac{q^2 w_j}{w_{j+1}}\right]|\psi(\dots,w_{j+1},w_j,\dots)\rangle,
\end{equation}
for both $|\psi\rangle = |\psi_{\text{\rm \tiny D}}\rangle$ and $|\psi\rangle = |\psi_{\text{\rm \tiny AD}}\rangle$.
\end{proposition}
\begin{proof}
  For the diagonal twist, this relation was established in \cite{hagendorf:15}. We focus on the anti-diagonal twist. Notice that the Yang-Baxter equation implies
  \begin{equation}
    \check R_{j,j+1}\left(\frac{w_j}{w_{j+1}}\right)T^{(1)}(z|\dots,w_j,w_{j+1},\dots) =  T^{(1)}(z|\dots,w_{j+1},w_{j},\dots)\check R_{j,j+1}\left(\frac{w_j}{w_{j+1}}\right).
  \end{equation}
  Applying this to $|\psi\rangle = |\psi_{\textrm{\tiny AD}}\rangle$, we find that 
  \begin{equation}
     T^{(1)}(z|\dots,w_{j+1},w_{j},\dots)\left(\check R_{j,j+1}\left(\frac{w_j}{w_{j+1}}\right)|\psi_{\textrm{\tiny AD}}(\dots,w_j,w_{j+1},\dots)\rangle\right) =0.
  \end{equation}
  As the null space of the transfer matrix with anti-diagonal twist is one-dimensional, we conclude that 
  \begin{equation}
    \check R_{j,j+1}\left(\frac{w_j}{w_{j+1}}\right)|\psi_{\textrm{\tiny AD}}(\dots,w_j,w_{j+1},\dots)\rangle = \lambda |\psi_{\textrm{\tiny AD}}(\dots,w_{j+1},w_j,\dots)\rangle.
  \end{equation}
  for some $\lambda$. Its value is found by projecting each side of the equation onto $\langle \wedge|$ and using $\langle \wedge|\check R(z) = [qz][q^2 z]\langle \wedge| $.
\end{proof}

Next, we consider the behaviour of the eigenvectors under cyclic shifts. The translation operator $S$ is the linear operator acting on the canonical basis states of 
$(\mathbb C^3)^{\otimes N}$ as
\begin{equation}
   S|\sigma_1\cdots \sigma_{N-1}\sigma_N\rangle = |\sigma_N\sigma_1 \cdots \sigma_{N-1}\rangle.
\end{equation}
The 
\textit{twisted translation operator} is defined by $S' = S \Omega_N^{(2)}$, where $\Omega^{(2)}$ is chosen from \eqref{eqn:twistmatrices} according to the choice of the twist. For both choices, $S'$ commutes with the corresponding Hamiltonian.
\begin{proposition} \label{prop:trslcov}
The vectors $|\psi\rangle = |\psi_{\text{\rm \tiny D}}\rangle$ and $|\psi\rangle = |\psi_{\text{\rm \tiny AD}}\rangle$ have the translation covariance property
\begin{equation}\label{eq:trslcov}
  S'|\psi(w_1,\dots,w_N)\rangle = \bigg(\prod_{j=1}^{N-1}\frac{[q^{-1}w_N/w_j]}{[qw_N/w_j]}\bigg)|\psi(w_N,w_1,\dots,w_{N-1})\rangle.
\end{equation}
\end{proposition}
\begin{proof}
  For both the diagonal and anti-diagonal twist, the translation covariance is a consequence of the exchange relation and the form of the simple eigenvalue \eqref{eqn:specialev}. First, notice that the transfer matrix of the nineteen-vertex model has the property
  \begin{equation} 
    T^{(2)}(z=w_1) = [q][q^2]S'\check R_{N-1,N}\left(\frac{w_1}{w_N}\right)\check R_{N-2,N-1}\left(\frac{w_1}{w_{N-1}}\right)\cdots \check R_{1,2}\left(\frac{w_1}{w_2}\right).
  \end{equation}
  We apply this to the eigenvector and find
  \begin{equation}
    [q][q^2]\prod_{j=2}^N\left[\frac{q w_1}{w_j}\right]\left[\frac{q^2 w_1}{w_j}\right]S'|\psi(w_2,\dots,w_N,w_1)\rangle = \theta^{(2)}(w_1)|\psi(w_1,\dots,w_N)\rangle.
  \end{equation}
  The statement of the proposition follows after evaluation of the eigenvalue $\theta^{(2)}(w_1)$ and a relabelling of the indices.
\end{proof}

%%%%%%%%%%
\subsubsection{$\mathbb Z_2$-symmetries}
\label{sec:Z2symm}
%%%%%%%%%%

In this subsection, we consider properties of the special eigenvectors in relation with the magnetisation and spin-reversal operators. We start with the magnetisation. 
The next proposition is straightforward from the construction of the special eigenstates.
\begin{proposition}
  $
    M|\psi_{\text{\rm \tiny D}}\rangle = 0$ \text{and} $(-1)^M|\psi_{\text{\rm \tiny AD}}\rangle = (-1)^N |\psi_{\text{\rm \tiny AD}}\rangle$.
  \label{prop:magnetisation}
\end{proposition}
For the anti-diagonal twist the behaviour under spin-reversal can be deduced from the covariance under cyclic shifts.
\begin{proposition}
  $F|\psi_{\text{\rm \tiny AD}}\rangle = (-1)^N |\psi_{\text{\rm \tiny AD}}\rangle.$
  \label{prop:SRantidiagonal}
\end{proposition}
\begin{proof}
  From the translation covariance \eqref{eq:trslcov},
   it follows that $(S')^N |\psi_{\text{\rm \tiny AD}}\rangle = |\psi_{\text{\rm \tiny AD}}\rangle$. On the other hand, because $S' |\sigma_1 \cdots \sigma_{N-1}\sigma_N\rangle = (-1) |(-\sigma_N) \sigma_1\cdots \sigma_{N-1}\rangle$, 
  we conclude that $(S')^N = (-1)^N F$  and the statement readily follows.\end{proof}

Analysing the behaviour of $|\psi_{\text{\rm \tiny D}}\rangle$ under spin reversal turns out to be more challenging. 
In the proof of the next proposition, we will encounter an intriguing feature
which will play a fundamental role in the forthcoming sections: 
the emergence of partition functions for the six-vertex model within our analysis of the nineteen-vertex model. 
\begin{proposition} 
  \label{prop:SRdiagonal}
  $F|\psi_{\text{\rm \tiny D}}\rangle = |\psi_{\text{\rm \tiny D}}\rangle.$
\end{proposition}
\begin{proof}
  For convenience, we prove the invariance under the action of $F$ on
  the dual vector, $\langle \psi_{\text{\rm \tiny D}}|F= \langle \psi_{\text{\rm \tiny D}}|$, which is equivalent to the statement 
  of this proposition by virtue of \cref{prop:transposition}. The first step of the proof is to show that for arbitrary $z_1,\dots,z_N$, we have
  \begin{equation}
    \langle \psi_{\text{\rm \tiny D}}|F\prod_{j=1}^N\mathcal B(z_j)|\wedge\rangle = \langle \psi_{\text{\rm \tiny D}}|\prod_{j=1}^N\mathcal B(z_j)|\wedge\rangle.
    \label{eqn:SRfundequ}
  \end{equation}
    The right-hand side is
  computed in \cite{hagendorf:15}: 
  \begin{equation}
    \langle \psi_{\text{\rm \tiny D}}|\prod_{j=1}^N\mathcal B(z_j)|\wedge\rangle = \langle \wedge |\prod_{j=1}^N\mathcal C(w_j)\prod_{j=1}^N\mathcal B(z_j)|\wedge\rangle = \bigg(\prod_{i=1}^N a(w_i)\bigg)Z_{\text{\rm \tiny IK}}(z_1,\dots,z_N;w_1,\dots,w_N).
    \label{eqn:spikdet}
  \end{equation}
  Here $Z_{\text{\rm \tiny IK}}(z_1,\dots,z_N;w_1,\dots,w_N)$ is the partition function of the six-vertex model with specific statistical weights on an $N\times N$ lattice with domain-wall boundary conditions.
 This lattice is illustrated on the left panel of \cref{fig:dwbc} for $N=4$.
  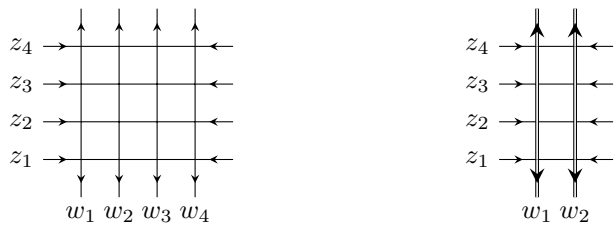
\begin{figure}
  \centering
  \begin{tikzpicture}
    \begin{scope}
    \draw (0,0) grid[step=.5cm] (1.5,1.5);
    \foreach \x in {0,0.5,1,1.5}
    {
      \draw[postaction={on each segment={mid arrow}}] (-.5,\x)--(0,\x);
      \draw[postaction={on each segment={mid arrow}}] (2,\x)--(1.5,\x);
      \draw[postaction={on each segment={mid arrow}}] (\x,0)--(\x,-0.5);
      \draw[postaction={on each segment={mid arrow}}] (\x,1.5)--(\x,2);
      
    }
    
    \draw (0,-.5) node [below] {$w_1$};
    \draw (.5,-.5) node [below] {$w_2$};
    \draw (1.,-.5) node [below] {$w_3$};
    \draw (1.5,-.5) node [below] {$w_4$};
    
    \draw (-.5,0) node [left] {$z_1$};
    \draw (-.5,.5) node [left] {$z_2$};
    \draw (-.5,1) node [left] {$z_3$};
    \draw (-.5,1.5) node [left] {$z_4$};
    \end{scope}
    
    \begin{scope}[xshift=6cm]
    \foreach \x in {0,0.5,1,1.5}
    {
      \draw[postaction={on each segment={mid arrow}}] (-.5,\x)--(0,\x);
      \draw[postaction={on each segment={mid arrow}}] (1,\x)--(0.5,\x);
      \draw (0,\x)--(0.5,\x);
    }
    
     \foreach \x in {0,0.5}
    {
      \draw[double, postaction={on each segment={mid arrow}}] (\x,0)--(\x,-.5);
      \draw[double,postaction={on each segment={mid arrow}}] (\x,1.5)--(\x,2);
      \draw[double] (\x,0)--(\x,1.5);
    }
    
    \draw (0,-.5) node [below] {$w_1$};
    \draw (.5,-.5) node [below] {$w_2$};
    
    \draw (-.5,0) node [left] {$z_1$};
    \draw (-.5,.5) node [left] {$z_2$};
    \draw (-.5,1) node [left] {$z_3$};
    \draw (-.5,1.5) node [left] {$z_4$};
    \end{scope}

  \end{tikzpicture}
  \caption{{\it Left:} Domain-wall boundary conditions for the six-vertex model on a $4 \times 4$ square. {\it Right:} Domain-wall boundary conditions for the ten-vertex model on a $4\times 2$ rectangle.} 
  \label{fig:dwbc}
\end{figure}
  The statistical weights $\mathfrak a(z),\mathfrak b(z),\mathfrak c(z)$ of the vertex configurations are given in \cref{fig:6v}. 
  The first (respectively second) group of arguments correspond to the spectral parameters assigned to rows (respectively columns), and the local 
  statistical weight of a vertex at position $(i,j)$ is evaluated with $z=z_i/w_j$. 
    \begin{figure}[h] 
  \centering
  \begin{tikzpicture}
    \drawvertex{1}{0}{0}
    \drawvertex{2}{1.25}{0}
    \drawvertex{3}{3.25}{0}
    \drawvertex{4}{4.5}{0}
    \drawvertex{5}{6.5}{0}
    \drawvertex{6}{7.75}{0}

    \draw (-.5,-1) node {\textit{(i)}};
    
    \draw (1.125,-1) node {$\mathfrak a(z)=[q/z]$};
    \draw (4.425,-1) node {$\mathfrak b(z)=[q z]$};
    \draw (7.625,-1) node {$\mathfrak c(z)=[q^2]$};
    
    \begin{scope}[yshift=-.75cm]
    \draw (-.5,-1) node {\textit{(ii)}};
    
    \draw (1.125,-1) node {$\bar {\mathfrak a}(z)=[qz]$};
    \draw (4.425,-1) node {$\bar{\mathfrak b}(z)=[z]$};
    \draw (7.625,-1) node {$\bar{\mathfrak c}(z)=[q]$};
    \end{scope}
    
    \end{tikzpicture}
 \caption{Vertex configurations of the six-vertex model
 with two sets of statistical weights: \textit{(i)} weights in Kuperberg's parametrisation, \textit{(ii)} 
weights obtained from the $R$-matrix \eqref{eqn:6vRmatrix}.
}
  \label{fig:6v}
\end{figure}
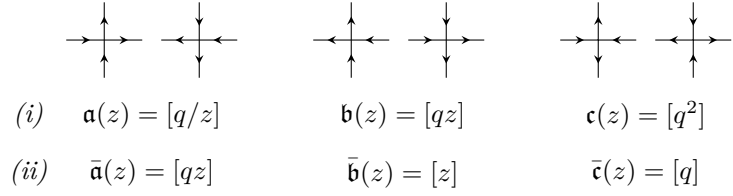

The Izergin-Korepin formula \cite{izergin:92} is an exact expression for this partition function:
  \begin{equation}
    Z_{\text{\rm \tiny IK}}(z_1,\dots,z_N;w_1,\dots,w_N) = \frac{\prod_{i,j=1}^N \mathfrak a(z_i/w_j)\mathfrak b(z_i/w_j)}{\prod_{1\le i<j \le N}[z_i/z_j][w_j/w_i]}\Det{i,j=1}{N}\left(\frac{\mathfrak c(z_i/w_j)}{\mathfrak a(z_i/w_j)\mathfrak b(z_i/w_j)}\right).
    \label{eqn:IKDet}
  \end{equation}
  
We now show that the left-hand side of \eqref{eqn:SRfundequ}, 
\begin{equation}\label{eq:371}
  \langle \psi_{\text{\rm \tiny D}}|F\prod_{j=1}^N \mathcal B(z_j)|\wedge\rangle
  = \langle \wedge |\prod_{j=1}^N \mathcal C(w_j)\prod_{j=1}^N \mathcal C(z_j)|\vee\rangle,
\end{equation}
produces the same result as the combined \eqref{eqn:spikdet} and \eqref{eqn:IKDet}.
 Here we used $F \mathcal B(z) = \mathcal C(z)F$ and introduced $|\vee\rangle = |{\Downarrow\cdots\Downarrow}\rangle$. Following the discussion about diagrammatic calculus in \cref{sec:tm}, the right-hand side of \eqref{eq:371} can be seen as the partition function of the ten-vertex model on a $2N\times N$ rectangle with domain-wall boundary conditions. It is illustrated in the right panel of \cref{fig:dwbc} for $N=2$. The parameters associated to the columns are $w_1,\dots,w_N$. Those associated to rows are $z_1,\dots, z_{2N}$ with $z_{j+N}=w_j,$ for $j=1,\dots,N$. From the fusion procedure, this partition function can be written in terms of a partition function $\bar Z_{\text{\rm \tiny IK}}$ of a six-vertex model with domain-wall boundary conditions on a $2N\times 2N$ square. Its vertex configurations have the statistical weights $\bar{\mathfrak a}(z),\bar{\mathfrak b}(z),\bar{\mathfrak c}(z)$ of the $R$-matrix \eqref{eqn:6vRmatrix}, also displayed in \cref{fig:6v}. The result is:
\begin{equation}
  \langle \psi_{\text{\rm \tiny D}}|F\prod_{j=1}^N \mathcal B(z_j)|\wedge\rangle = \lim_{z_{N+1}\to w_1,\dots,z_{2N}\to w_N}\frac{\bar Z_{\text{\rm \tiny IK}}(z_1,\dots,z_{2N};w_1,q w_1,\dots,w_N,q w_N)}{\prod_{i=1}^{2N}\prod_{j=1}^{N}[z_i/w_j]}.
\end{equation}
The partition function $\bar Z_{\text{\rm \tiny IK}}$ is obtained from the Izergin-Korepin formula \eqref{eqn:IKDet}, with the weights $\mathfrak a(z),\mathfrak b(z),\mathfrak c(z)$ replaced by $\bar{\mathfrak a}(z),\bar{\mathfrak b}(z),\bar{\mathfrak c}(z)$. Writing this out, we obtain
\begin{subequations}
\begin{align}
  \frac{\bar Z_{\text{\rm \tiny IK}}(z_1,\dots,z_{2N};w_1,q w_1,\dots,w_N,q w_N)}{\prod_{i=1}^{2N}\prod_{j=1}^{N}[z_i/w_j]} =\, &\frac{\prod_{i=1}^{2N}\prod_{j=1}^N[w_j/z_i][qz_i/w_j][q w_j/z_i]}{\prod_{1\le i<j\le 2N}[z_i/z_j]\prod_{1\le i<j\le N}[w_i/w_j][w_j/w_i]\prod_{i,j=1}^N[q w_i/w_j]}\nonumber \\  &\quad \times \Det{i,j=1}{2N}\left(M_{ij}\right)
  \label{eqn:horribleeqn}
\end{align}
with the determinant entries
\begin{equation}
  M_{i,2j-1} = \frac{[q]}{[z_i/w_j][q z_i/w_j]},\quad M_{i,2j} = -\frac{[q]}{[z_i/w_j][q w_j/z_i]}.
\end{equation}
\end{subequations}
Notice that $\det M = \det \tilde M$ where
\begin{equation}
  \tilde M_{i,2j-1}=M_{i,2j-1}, \quad \tilde M_{i,2j} = M_{i,2j-1}+M_{i,2j}=-\frac{\mathfrak c(z_i/w_j)}{\mathfrak a(z_i/w_j)\mathfrak b(z_i/w_j)}. 
\end{equation}
Aside from the minus sign, the matrix entries $\tilde M_{i,2j}$ are identical to those appearing in the determinant of \eqref{eqn:IKDet}.
In the limit $z_{N+1}\to w_1,\dots,z_{2N}\to w_N$, the entries $\tilde M_{i+N,2i-1}$ diverge. These divergences are compensated by some prefactors on the right-hand side of \eqref{eqn:horribleeqn}. 
Taking the limit thus corresponds to selecting
certain minors of $\tilde M$, which reduces $\det \tilde M$ to the determinant of an $N\times N$ matrix. The calculation is straightforward and the final result is indeed 
\begin{equation}
\langle \psi_{\text{\rm \tiny D}}|F\prod_{j=1}^N \mathcal B(z_j)|\wedge\rangle  = \left(\prod_{i=1}^N a(w_i)\right)Z_{\text{\rm \tiny IK}}(z_1,\dots,z_N;w_1,\dots,w_N),
\end{equation}
which ends the proof of \eqref{eqn:SRfundequ}.   
   
   The second step of the proof is to use \eqref{eqn:SRfundequ} to compute the scalar products $\langle \psi_{\text{\rm \tiny D}}|F||\bm h\rrangle$ and $\langle \psi_{\text{\rm \tiny D}}||\bm h\rrangle$ for each $||\bm h\rrangle$. Since $\langle \psi_{\text{\rm \tiny D}}|$ has zero magnetisation, these scalar products
   vanish unless $\sum_{i=1}^N h_i=N$. 
   By adjusting the parameters $z_1, z_2,\dots,z_N$ in \eqref{eqn:SRfundequ} appropriately, we find that
   $\langle \psi_{\text{\rm \tiny D}}|F||\bm h\rrangle = \langle \psi_{\text{\rm \tiny D}}||\bm h\rrangle$ for each $\bm h$. 
   Since the separation of variables basis is complete
   for generic values of $q$ and of the $w_j$, 
   we conclude that $\langle \psi_{\text{\rm \tiny D}}|F= \langle \psi_{\text{\rm \tiny D}}|$, which ends the proof.  
  
\end{proof}
%%%%%%%%%%
\subsubsection{Homogeneous limit}
\label{sec:homog.definitions}
%%%%%%%%%%

We are now in a position to provide the proofs of \cref{thm:DTwistSpecialEV,thm:ADTwistSpecialEV}. To this end, we consider the case where all inhomogeneity parameters take the same value $w_1=\cdots=w_N=1$. For the diagonal twist, we define
\begin{equation}
  |\phi_{\text{\rm \tiny D}}\rangle = \left([q][q^2]\right)^{-N/2}[q]^{-N(N-1)}|\psi_{\text{\rm \tiny D}}(w_1=1,\dots,w_N=1)\rangle.
  \label{eqn:defPhiD}
\end{equation}
From \eqref{eqn:defpsiD}, we see that this vector is well-defined, except for $q = 0, \pm 1, \pm \i$. We conclude from \cref{prop:trslcov,prop:magnetisation,prop:SRdiagonal} that 
\begin{equation}
  M |\phi_{\text{\rm \tiny D}}\rangle =0,\quad F |\phi_{\text{\rm \tiny D}}\rangle =|\phi_{\text{\rm \tiny D}}\rangle,\quad S'|\phi_{\text{\rm \tiny D}}\rangle=(-1)^{N+1}|\phi_{\text{\rm \tiny D}}\rangle.
\end{equation}
Furthermore, combining the special eigenvalue \eqref{eqn:specialev} with \eqref{eqn:hamfromT}, we see that $|\phi_{\text{\rm \tiny D}}\rangle$ is indeed a zero-energy state:
\begin{equation}
  H|\phi_{\text{\rm \tiny D}}\rangle=0.
\end{equation}
This concludes the proof of \cref{thm:DTwistSpecialEV}. Apart from the spin-reversal property, these features were all proved in \cite{hagendorf:15}.
Furthermore, \eqref{eqn:defPhiD} fixes the normalisation of $|\phi_{\text{\rm \tiny D}}\rangle$. The results involving $|\phi_{\text{\rm \tiny D}}\rangle$ presented in \cref{sec:spinoneXXZ} hold with this normalisation.

Following the discussion at the end of \cref{sec:tm}, extra caution must be exercised for the values $q=q_c$ satisfying $(q_c)^4=1$. The zero-energy eigenstates of $H$ are obtained through a limiting procedure, by first setting $w_1,\dots,w_N$ to $1$ and then taking the limit $q \rightarrow q_c$:
\begin{equation}\label{limitspecialeqD}
  |\phi_{\text{\rm \tiny D}}\rangle\big|_{q=q_c} = \lim_{q\rightarrow q_c}\left([q][q^2]\right)^{-N/2}[q]^{-N(N-1)}|\psi_{\text{\rm \tiny D}}(1,\dots,1)\rangle.
\end{equation}
It was shown in \cite{hagendorf:15} that the state $  |\phi_{\text{\rm \tiny D}}\rangle$ in \eqref{eqn:defPhiD} is polynomial in $x$, which implies that the limit in \eqref{limitspecialeqD} is well-defined.

For the anti-diagonal twist, setting all the inhomogeneity parameters to unity is more delicate. From \eqref{eqn:scalarprod} and \eqref{eq:explicitEV}, we see that the components of the vector $|\psi_{\text{\rm \tiny AD}}(w_1,\dots,w_N)\rangle$ in the separation of variables basis are singular in the homogeneous limit. Moreover, the height states no longer form a basis of $(\mathbb C^3)^{\otimes N}$ in this limit. These two difficulties turn out to compensate for one another: We argue in \cref{eq:exiandpoly}
that the vector
\begin{equation}\label{eq:phiADdef}
  |\phi_{\text{\rm \tiny AD}}\rangle = \lim_{w_1,\dots,w_N\to 1}|\psi_{\text{\rm \tiny AD}}(w_1,\dots,w_N)\rangle
\end{equation}
is indeed well defined for all $q \in \mathbb C^{\times}$. According to \cref{prop:trslcov,prop:magnetisation,prop:SRantidiagonal}, it has the properties
\begin{equation}
(-1)^M |\phi_{\text{\rm \tiny AD}}\rangle =(-1)^N|\phi_{\text{\rm \tiny AD}}\rangle,\quad F |\phi_{\text{\rm \tiny AD}}\rangle =(-1)^N|\phi_{\text{\rm \tiny AD}}\rangle, \quad S'|\phi_{\text{\rm \tiny AD}}\rangle=(-1)^{N+1}|\phi_{\text{\rm \tiny AD}}\rangle.
\end{equation}
One shows in complete analogy with the diagonal twist that this vector is a zero-energy state:
\begin{equation}
  H|\phi_{\text{\rm \tiny AD}}\rangle=0.
\end{equation}
This completes the proof of \cref{thm:ADTwistSpecialEV}. Furthermore, it is easy to see from \eqref{eqn:projections} that
\begin{equation*}
  (\phi_{\text{\rm \tiny AD}})_{\Uparrow\cdots \Uparrow} =\lim_{w_1,\dots,w_N\to 1} 
  \langle \wedge |\psi_{\text{\rm \tiny AD}}(w_1,\dots,w_N)\rangle = 1,
\end{equation*}
which corresponds to our choice for the normalisation in \cref{sec:qsumrules}.

As for the diagonal case, the eigenstate at the special values of $q$ where $(q_c)^4=1$ is obtained by a limiting procedure, but now by first taking the limit $w_1, \dots, w_N \to 1$:
\begin{equation}
|\phi_{\text{\rm \tiny AD}}\rangle\big|_{q=q_c} =  \lim_{q\rightarrow q_c}\lim_{w_1,\dots,w_N\to 1}|\psi_{\text{\rm \tiny AD}}(w_1,\dots,w_N)\rangle\label{eq:specialhomog}.
\end{equation}
In \cref{eq:exiandpoly}, we show that this limit is well-defined.

%%%%%%%%%%%%%%%%%%%%%
%
\section{Quadratic sum rules} 
\label{sec:sumrules}
%
%%%%%%%%%%%%%%%%%%%%%

In this section, we establish a number of quadratic sum rules through the explicit calculation of scalar products involving the left- and right-eigenvectors. In \cref{sec:squaresumrules}, we discuss a sum rule of the special eigenvector for the anti-diagonal twist. In \cref{sec:overlap}, we consider the scalar product of the special eigenvectors for the diagonal and anti-diagonal twists. Our results are expressed in terms of partition functions for a six-vertex model on certain lattice domains. In the homogeneous limit, these partition functions become generating functions for ASM enumeration. This allows us to recover the sum rules for the spin-chain zero-energy states given in \cref{sec:qsumrules}.

%%%%%%%%%
\subsection{Square norm sum rules}
\label{sec:squaresumrules}
%%%%%%%%%

%%%%%%%%%
\subsubsection{Generating functions for scalar products}
%%%%%%%%%

The vector $|\psi_{\text{\tiny AD}}\rangle$ has non-zero components only in magnetisation sectors where $(-1)^M = (-1)^N$.
Thus, instead of the simple scalar product
 $\langle \psi_{\text{\tiny AD}} |\psi_{\text{\tiny AD}}\rangle$, we consider the expression
\begin{equation}
  Z_{\text{\tiny AD}}(y;w_1,\dots,w_N)=\langle \psi_{\text{\tiny AD}}|y^M|\psi_{\text{\tiny AD}}\rangle = \sum_{\sigma}y^{m(\sigma)}\psi_{\text{\tiny AD}}(w_1^{-1},\dots,w_N^{-1})_{\sigma}\psi_{\text{\tiny AD}}(w_1,\dots,w_N)_{\sigma},
  \label{eqn:defZ}
\end{equation}
where $m(\sigma)$ is defined by $M|\sigma\rangle = m(\sigma)|\sigma\rangle$.
For the second equality of \eqref{eqn:defZ}, we used \cref{prop:transposition} to write the components of $\langle \psi_{\text{\tiny AD}} |$ as $\psi_{\text{\tiny AD}}(w_1^{-1},\dots,w_N^{-1})_{\sigma}$.
As a function of $y$, $Z_{\text{\tiny AD}}(y;w_1,\dots,w_N)$ is a centred Laurent polynomial of degree width $2N$. 
In a Laurent expansion in the variable $y$,
the coefficient of $y^m$ is the square norm of the vector's projection onto the subspace of magnetisation $m$. 
As the next proposition shows, it can be evaluated in terms of a determinant.
\begin{proposition}
We have
\begin{equation}
  Z_{\text{\rm \tiny{AD}}}(y;w_1,\dots,w_N)= \frac{\prod_{i,j=1}^N[q w_i/w_j]}{\prod_{1\le i<j\le N}[w_i/w_j][w_j/w_i]}\Det{i,j=1}{N}\left(\frac{y^{-1}}{[q w_i/w_j]}+\frac{y}{[q w_j/w_i]}\right).
  \label{eqn:ZAD}
\end{equation}
\begin{proof}
In order to compute $Z_{\text{\tiny AD}}(y;w_1,\dots,w_N)$, 
we use \eqref{eq:explicitEV}, \eqref{eqn:projections} and \eqref{eq:Mh} and
obtain
\begin{equation}\label{eq:Zady}
Z_{\text{\tiny AD}}(y;w_1,\dots,w_N)=\sum_{\bm h}\prod_{i=1}^N(-1)^{h_i}y^{1-h_i}\prod_{1\le i<j\le N}\frac{[q^{h_j-h_i}w_i/w_j]}{[w_i/w_j]}\prod_{i,j=1}^N \frac{[q^{2(h_i-1)}w_i/w_j][q^{1-h_i}w_i/w_j]}{[q^{-2}w_i/w_j][q w_i/w_j]}.
\end{equation}
All terms in this expression with at least one height taking the value one vanish. 
Furthermore, for $i\neq j$ and $h_i,h_j \in \{0,2\}$, one verifies the identity
\begin{equation}\label{eq:special.identity}
  [q^{h_j-h_i}w_i/w_j][q^{2(h_i-1)}w_i/w_j][q^{2(h_j-1)}w_j/w_i] =[q^{h_i-h_j}w_i/w_j][q^{2}w_i/w_j][q^{2}w_j/w_i].
\end{equation}
It allows us to simplify the expression \eqref{eq:Zady} to
\begin{equation}
Z_{\text{\tiny AD}}(y;w_1,\dots,w_N) =\sum_{\bm h}\prod_{i=1}^N(-1)^{h_i/2}y^{1-h_i}
\prod_{1\le i<j \le N}\frac{[q^{h_i-h_j}w_i/w_j]}{[w_i/w_j]}\prod_{i,j=1}^N \frac{[q^{1-h_i}w_i/w_j]}{[q w_i/w_j]}.
\end{equation}
Each summand can be rewritten with the help of a determinant identity of Cauchy's \cite{cauchy:41}:
\begin{equation}
  \Det{i,j=1}{N} \left(\frac{1}{[x_i/y_j]}\right) = \frac{\prod_{1\le i<j \le N} [x_i/x_j][y_j/y_i]}{\prod_{i,j=1}^N[x_i/y_j]}.
\end{equation}
Indeed, we choose $x_i = q^{h_i-1}w_i, \,y_j=w_j$ and find
\begin{align}
  Z_{\text{\tiny AD}}(y;w_1,\dots,w_N)
  &= \sum_{\bm h} \frac{\prod_{i,j=1}^N[q^{1-h_i}w_i/w_j][q^{h_i-1}w_i/w_j] }{\prod_{i,j=1}^N[q w_i/w_j] 
  \prod_{1\le i<j \le N}[w_i/w_j][w_j/w_i]}\Det{i,j=1}{N}\left(\frac{(-1)^{h_i/2}y^{1-h_i}}{[q^{h_i-1}w_i/w_j]}\right).
\end{align}
Finally, since only height profiles $\bm h$ with $h_i=0,2$ contribute to the sum, 
the double product in the numerator of each summand does not depend on the actual values of the heights and can be factorised. It remains to sum up the determinants which corresponds to performing simple row operations and leads to \eqref{eq:Zady}.
\end{proof}
\end{proposition}

\paragraph{Relation to six-vertex model partition functions.} 
There are several specialisations of $y$ for which $Z_{\text{\tiny AD}}(y; w_1, \dots, w_N)$ can be written in terms of known partition functions for the six-vertex model 
on various lattice domains. The vertex weights are given by the functions $\mathfrak a(z)$, $\mathfrak b(z)$ and $\mathfrak c(z)$ in \cref{fig:6v}. 

The first special point is $y=q$. Upon comparison of \eqref{eqn:IKDet} and \eqref{eqn:ZAD}, we find that $Z_{\text{\tiny AD}}(y=q;w_1,\dots,w_N)$ reduces to the Izergin-Korepin determinant up to a 
prefactor:
\begin{equation}
  Z_{\text{\tiny AD}}(y=q;w_1,\dots,w_N) = \frac{Z_{\text{\tiny IK}}(w_1,\dots,w_N;w_1,\dots,w_N)}{\prod_{i,j=1}^N [q w_i/w_j]}.
  \label{eqn:ZADfromIKDet}
\end{equation}

For $y=1$ and $y=\i$, the expression \eqref{eqn:ZAD} can be written in terms of $Z_{\text{\tiny IK}}(w_1,\dots,w_N;w_1,\dots,w_N)$ and the partition functions $Z_{\text{\tiny HT}}^{\pm}(z_1,\dots,z_N;w_1,\dots,w_N)$ for the six-vertex model on a $2N \times N$ rectangle with half-turn boundary conditions. As illustrated in \cref{fig:6vpartfunc}, the difference between $Z_{\text{\tiny HT}}^{-}$ and $Z_{\text{\tiny HT}}^{+}$ is that on the inhomogeneous lattice, the inhomogeneity parameters attached to the top $N$ rows are set to $-z_j$ instead of $z_j$.
These partition functions were computed by Kuperberg \cite{kuperberg:02}: 
\begin{subequations}
\begin{equation}
\frac{Z_{\text{\tiny HT}}^{\pm}(z_1,\dots,z_N;w_1,\dots,w_N)}{Z_{\text{\tiny IK}}(z_1,\dots,z_N;w_1,\dots,w_N)} =  \frac{\prod_{i,j=1}^N [q z_i/w_j][q w_j/z_i]}{\prod_{1\le i<j\le N}[z_j/z_i] [w_i/w_j]} \Det{i,j=1}{N} (M^\pm_{\text{\tiny HT}})_{ij}
\end{equation}
where
\begin{equation}
(M^\pm_{\text{\tiny HT}})_{ij} = \frac{1}{[q w_j/z_i]} \pm \frac{1}{[q z_i/w_j]}.
\end{equation}
\end{subequations}
 Upon comparing \eqref{eqn:ZAD} with his results, we find that
  \begin{subequations}
  \label{eqn:ZADfromHTDet}
\begin{align}
  Z_{\text{\tiny AD}}(y=1;w_1,\dots,w_N)=\frac{Z_{\text{\tiny HT}}^+(w_1,\dots,w_N;w_1,\dots,w_N)}{\prod_{i,j=1}^N [q w_i/w_j]Z_{\text{\tiny IK}}(w_1,\dots,w_N;w_1,\dots,w_N)},\\
   Z_{\text{\tiny AD}}(y=\i;w_1,\dots,w_N)=\frac{\i^N Z_{\text{\tiny HT}}^-(w_1,\dots,w_N;w_1,\dots,w_N)}{\prod_{i,j=1}^N [q w_i/w_j]Z_{\text{\tiny IK}}(w_1,\dots,w_N;w_1,\dots,w_N)}. 
\end{align}
  \end{subequations}

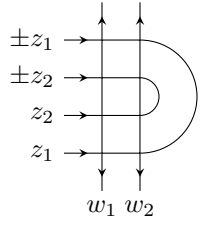
\begin{figure}
  \centering
  \begin{tikzpicture}
    \draw (0,0) grid[step=.5cm] (0.5,1.5);
    \foreach \x in {0,0.5}
    {
      \draw[postaction={on each segment={mid arrow}}] (\x,0)--(\x,-0.5);
      \draw[postaction={on each segment={mid arrow}}] (\x,1.5)--(\x,2);
      
      \draw[postaction={on each segment={mid arrow}}] (-.5,\x)--(0,\x);
      \draw[postaction={on each segment={mid arrow}},yshift=1cm] (-.5,\x)--(0,\x);
    }
    \draw (.5,0) arc [start angle=-90, end angle = 90, radius=.75cm];
    \draw (.5,.5) arc [start angle=-90, end angle = 90, radius=.25cm];
   
    \draw (0,-.5) node [below] {$w_1$};
    \draw (.5,-.5) node [below] {$w_2$};
    
       \draw (-.5,0) node [left] {$z_1$};
    \draw (-.5,.5) node [left] {$z_2$};
    \draw (-.5,1) node [left] {$\pm z_2$};
    \draw (-.5,1.5) node [left] {$\pm z_1$};
  \end{tikzpicture}  
  \caption{The inhomogeneous lattice with half-turn boundary conditions on a $4 \times 2$ rectangle. The inhomogeneity parameters of the top rows are set to either $z_j$ or $- z_j$, with the corresponding partition functions denoted by $Z^+_{\text{\tiny HT}}$ and $Z^-_{\text{\tiny HT}}$.}
  \label{fig:6vpartfunc}
\end{figure}

%%%%%%%%%
\subsubsection{Homogeneous limit}
%%%%%%%%%

In order to recover our results for the spin chain, we need to take the homogeneous limit of $Z_{\text{\rm \tiny AD}}(y;w_1,\dots,w_N)$:
\begin{equation}
  Z_{\text{\rm \tiny AD}}(y) = \lim_{w_1,\dots,w_N \to 1}Z_{\text{\rm \tiny AD}}(y;w_1,\dots,w_N).
\end{equation}
\begin{proposition}
\label{prop:ZADhom}
We have 
\begin{equation}
Z_{\text{\rm \tiny AD}}(y)
   = \Det{i,j=0}{N-1}\left(y^{-1}\delta_{ij}+y \sum_{k=0}^{N-1}\binom{i}{k}\binom{j}{k}x^{2(j-k)}\right),\quad x=q+q^{-1}.
     \label{eqn:ZADhom} 
\end{equation}
\begin{proof}
  We prove the statement by a direct application of the strategy presented by Behrend \textit{et al.}~\cite{behrend:12}. Here, we review their technique in some detail as we use it to compute other homogeneous limits in the following sections.

  The first part of the proof consists of rewriting the homogeneous limit in a suitable way. To this end, we consider the more general expression
 \begin{equation}
   Z_{\text{\rm \tiny{AD}}}^\ast(y;z_1,\dots,z_N; w_1,\dots,w_N)= \frac{\prod_{i,j=1}^N[q z_i/w_j]}{\prod_{1\le i<j \le N}[z_i/z_j][w_j/w_i]}\Det{i,j=1}{N}\left(\frac{y^{-1}}{[q z_i/w_j]}+\frac{y}{[q w_j/z_i]}\right),
  \label{eqn:ZADstar}
 \end{equation}
which satisfies $Z_{\text{\rm \tiny{AD}}}^\ast(y;w_1,\dots,w_N; w_1,\dots,w_N) = Z_{\text{\rm \tiny{AD}}}(y;w_1,\dots,w_N)$. In order to prove the proposition, we thus consider \eqref{eqn:ZADstar} in the limit $z_1,\dots,z_N,w_1,\dots,w_N\to 1$. It is convenient to introduce the variables $u_j=w_{j+1}^2,\,v_j = z_{j+1}^{-2}$ for $j=0, \dots,N-1$, in terms of which we have
 \begin{equation}
   Z_{\text{\rm \tiny{AD}}}^\ast(y;z_1,\dots,z_N; w_1,\dots,w_N) = 
   \frac{\prod_{i,j=0}^{N-1}(u_iv_j-q^2)}{q^{N(N-1)}}\left(\frac{\Det{i,j=0}{N-1}f(u_i,v_j)}{\prod_{0\le i<j\le N-1}(u_j-u_i)(v_j-v_i)}\right)
 \end{equation} 
 where we abbreviated
 \begin{equation}
   f(u,v) = \frac{y^{-1}}{uv-q^2}-\frac{yq^{-2}}{uv-q^{-2}}.
   \label{eqn:deffuv}
 \end{equation}
 The homogeneous limit is therefore
 \begin{align}
   Z_{\text{\rm \tiny AD}}(y) &= \frac{(1-q^2)^{N^2}}{q^{N(N-1)}}\lim_{\substack{u_0,\dots,u_{N-1}\to 1 \nonumber\\
 v_0,\dots,v_{N-1}\to 1} } \left(\frac{\Det{i,j=0}{N-1}f(u_i,v_j)}{\prod_{0\le i<j\le N-1}(u_j-u_i)(v_j-v_i)}\right)\\
 &= \frac{(1-q^2)^{N^2}}{q^{N(N-1)}}\lim_{\substack{u_0,\dots,u_{N-1}\to 1 \\
 v_0,\dots,v_{N-1}\to 1} } \Det{i,j=0}{N-1}
 f[u_0,\dots,u_i;v_0,\dots,v_j],
 \label{eqn:homlim}
 \end{align}
 where $f[u_0,\dots,u_i;v_0,\dots,v_j]$ is the so-called \textit{divided difference}
 \begin{equation}
   f[u_0,\dots,u_i;v_0,\dots,v_j] = \sum_{m=0}^i\sum_{n=0}^j \frac{f(u_m,v_n)}{\prod_{\substack{m'=0\\ m'\neq m}}^{i}(u_m-u_{m'})\prod_{\substack{n'=0\\ n'\neq n}}^{i}(u_n-u_{n'})}.
 \end{equation}
 
The second part of the proof is to compute this limit. To this end, it is useful to denote by $\{u^i v^j\}g(u,v)$ the coefficient in front of $u^i v^j$ of a power series $g(u,v)$. For any such power series, it was shown in \cite{behrend:12} that the confluent limit of the divided difference $g[u_0,\dots,u_{N-1};v_0,\dots,v_{N-1}]$ is 
\begin{equation}
 \lim_{\substack{u_0,\dots,u_{N-1}\to r \\
 v_0,\dots,v_{N-1}\to r} } g[u_0,\dots,u_{N-1};v_0,\dots,v_{N-1}]
 =\{u^i v^j\}g(u+r,v+r).
 \label{eqn:divideddiff}
 \end{equation}
 Here $r$ is 
 assumed to be such that both sides are well-defined.
 
We apply \eqref{eqn:divideddiff} to $g(u,v)= f(u,v)$ with $r=1$.
The series coefficients $\{u^i v^j\}f(u+1,v+1)$ can be computed with the help of the two-parameter family of matrices 
$L(\alpha,\beta)$ with entries
\begin{equation}
  L(\alpha,\beta)_{ij} = \binom{i}{j} \alpha^i \beta^j,\quad \alpha,\beta \in \mathbb C^\times, \quad i,j = 0, \dots, N-1.
   \label{eqn:defL}
 \end{equation}

 These matrices have the following properties. First, the determinant is given by
 \begin{equation}
   \det L(\alpha,\beta) = (\alpha\beta)^{N(N-1)/2},
   \label{eqn:propL0}
 \end{equation}
 and hence the matrices are invertible for $\alpha,\beta \neq 0$. Second, the product of a matrix and its transpose is
 \begin{equation}
   \left(L(\alpha,\beta)L(\alpha,\beta)^t\right)_{ij} = \{u^iv^j\}\left(\frac{1}{1-\alpha(u+v)
  -
   \alpha^2(\beta^2-1)uv}\right)=\alpha^{i+j}\sum_{k=0}^{N-1}\binom{i}{k}\binom{j}{k}\beta^{2k}
   \label{eqn:propL1}
 \end{equation}
 for $0\le i,j \le N-1$.
Finally, the matrices obey the multiplication law
 \begin{equation}
   L(\alpha_+,\beta_+)L(\alpha_0,\beta_0) = L(\alpha_-,\beta_-), \quad \alpha_0 = \frac{\alpha_{-}-\alpha_+}{\alpha_+\beta_+},\,\beta_0 = \frac{\alpha_-\beta_-}{\alpha_--\alpha_+}.
   \label{eqn:propL2}
 \end{equation}

Let us now proceed to the actual computation of the series coefficients. Comparing \eqref{eqn:deffuv} and \eqref{eqn:propL1}, we obtain
 \begin{equation}
   \{u^i v^j\}f(u+1,v+1) = \frac{1}{1-q^2}\left(y^{-1} L(\alpha_+,\beta_+)L(\alpha_+,\beta_+)^t+ y L(\alpha_-,\beta_-)L(\alpha_-,\beta_-)^t\right)_{ij}
 \end{equation}
 where
 \begin{equation}\label{eq:alphabeta}
   \alpha_\pm = \frac{1}{q^{\pm 2}-1},\quad \beta_\pm = q^{\pm 1}.
 \end{equation}
 Combining \eqref{eqn:homlim} and \eqref{eqn:divideddiff}, we find
 \begin{align}
   Z_{\text{\rm \tiny AD}}(y) &=\frac{(1-q^2)^{N(N-1)}}{q^{N(N-1)}}\det \left(y^{-1} L(\alpha_+,\beta_+)L(\alpha_+,\beta_+)^t+ y L(\alpha_-,\beta_-)L(\alpha_-,\beta_-)^t\right)\nonumber\\
   &= \frac{(1-q^2)^{N(N-1)}\det \left(L(\alpha_+,\beta_+)\right)^2}{q^{N(N-1)}}\nonumber \\
   & \qquad \times \det \left(y^{-1}\bm{1} + y L(\alpha_+,\beta_+)^{-1}L(\alpha_-,\beta_-)\left(L(\alpha_+,\beta_+)^{-1}L(\alpha_-,\beta_-)\right)^t\right).
 \end{align}
We simplify the resulting expression with the help of \eqref{eqn:propL0} and \eqref{eqn:propL2} with
\begin{equation}\alpha_0 = -(q+q^{-1})=-x,\quad\beta_0  = -(q+q^{-1})^{-1}=-x^{-1}.\end{equation} 
We obtain
\begin{align}
  Z_{\text{\rm \tiny AD}}(y) &= \det \left(y^{-1}\bm{1} + y L(\alpha_0,\beta_0)L(\alpha_0,\beta_0)^t\right)\nonumber\\
  & =  \Det{i,j=0}{N-1} \left(y^{-1}\delta_{ij} + (-1)^{i+j} y \sum_{k=0}^{N-1}\binom{i}{k}\binom{j}{k}x^{i-k}x^{j-k}\right)\nonumber\\
  & = \Det{i,j=0}{N-1} \left((-1)^{i+j}x^{i-j}\left( y^{-1}\delta_{ij} + y \sum_{k=0}^{N-1}\binom{i}{k}\binom{j}{k}x^{2(j-k)}\right)\right).
\end{align}
At the second equality, we used once again \eqref{eqn:propL1}. The factor $(-1)^{i+j}x^{i-j}$ in the third line is easily taken out of the determinant, and the statement of the proposition follows.
\end{proof}
\end{proposition}

As seen above, for $y=q,\, y=1$ and $y=\i$, the inhomogeneous quantity $Z_{\textrm {\tiny AD}}(y;w_1,\dots,w_N)$ is given in terms of partition functions for the six-vertex model on certain lattice domains. The vertex configurations on these lattice domains are in bijection with ASMs \cite{elkies:92,bressoudbook,kuperberg:02}. This bijection allows us to express the homogeneous limit of the partition functions in terms of generating functions for ASM enumeration.

The simplest case is the homogeneous limit of the Izergin-Korepin partition function. It is given by
\begin{equation}
  Z_{\textrm {\tiny IK}}(1,\dots,1;1,\dots,1) = [q]^{N(N-1)}[q^2]^{N} A(N; x^2)
\end{equation}
where $A(N;t)$ is the generating function for the $t$-enumeration of $N\times N$ ASMs that we described in \cref{sec:qsumrules}. Combining this with \eqref{eqn:ZADfromIKDet}, we obtain
\begin{equation}
  Z_{\textrm{\tiny AD}}(y=q) = x^N A(N;x^2).
  \label{eqn:ZADHom1}
\end{equation}

Kuperberg \cite{kuperberg:02} established the following relation between the generating functions $A_{\textrm {\tiny HT}}^\pm (2N;t)$ for the $t$-enumeration of HTASMs, see \cref{sec:qsumrules}, and the partition functions $Z_{\textrm {\tiny HT}}^\pm$:
\begin{equation}
  Z_{\textrm {\tiny HT}}^\pm(1,\dots,1;1,\dots,1) = [q]^{N(2N-1)}[q^2]^N A_{\textrm {\tiny HT}}^\pm (2N;x^2).
\end{equation}
Hence, \eqref{eqn:ZADfromHTDet} implies that
\begin{equation}
  Z_{\textrm {\tiny AD}}(y=1)= \frac{A_{\textrm {\tiny HT}}^+(2N;x^2)}{A(N;x^2)}, \quad Z_{\textrm {\tiny AD}}(y=\i)= \frac{\i^N A_{\textrm {\tiny HT}}^-(2N;x^2)}{A(N;x^2)}.
  \label{eqn:ZADHom2}
\end{equation}

Together with \eqref{eqn:ZADHom1} and \eqref{eqn:ZADHom2}, \cref{prop:ZADhom} leads to the results presented in \cref{sec:qsumrules}, provided that the homogeneous limit of the vector $|\psi_{\textrm{\tiny AD}}\rangle$ exists. This is addressed in the next section.

%%%%%%%%%
\subsubsection{Existence and polynomiality}
\label{eq:exiandpoly}
%%%%%%%%%

 As discussed in \cref{sec:homog.definitions}, the existence of $|\phi_{\text{\rm \tiny AD}}\rangle$ defined in \eqref{eq:phiADdef} is non-trivial considering the normalisation of $|\psi_{\text{\rm \tiny AD}}\rangle$ chosen in \cref{prop:psiAD.def}. 
 Our first aim is to prove that a specific homogeneous limit of $|\psi_{\text{\rm \tiny AD}}\rangle$ exists for $|q|=1$, that is for $-2\le x \le 2$.
\begin{proposition}
\label{prop:existencePhiAD}
  The limit
  \begin{equation}
    |\phi_{\text{\rm \tiny AD}}\rangle = \lim_{w_1,\dots,w_N \overset{\mathbb R}{\to} 1}|\psi_{\text{\rm \tiny AD}}(w_1,\dots,w_N)\rangle 
  \end{equation}
  exists for $|q|=1$. Here, $w_1,\dots,w_N\overset{ \mathbb R}{\to} 1$ indicates that the homogeneous limit is taken along real paths.
\begin{proof}
  Exceptionally, we include the parameter $q$ as an argument of the components, writing $\psi_{\text{\tiny AD}}(w_1,\dots,w_N;q)_\sigma$.
  For $w_1,\dots, w_N \in \mathbb R^\times$, we find
  \begin{align}
    Z(y=1;w_1,\dots,w_N) & = \sum_{\sigma}\psi_{\text{\tiny AD}}(w_1^{-1},\dots,w_N^{-1};q)_{\sigma}\psi_{\text{\tiny AD}}(w_1,\dots,w_N;q)_{\sigma}
   \label{eqn:truesn1}
    \\
    & = \sum_{\sigma}\psi_{\text{\tiny AD}}(w_1,\dots,w_N;q^{-1})_{\sigma}\psi_{\text{\tiny AD}}(w_1,\dots,w_N;q)_{\sigma}\nonumber 
    \\
    & = \sum_{\sigma}|\psi_{\text{\tiny AD}}(w_1,\dots,w_N;q)_{\sigma}|^2
   \nonumber
  \end{align}
where \cref{prop:qtoqinv} was used to obtain the second equality. For the third equality, we used $q^{-1}= q^\ast$ and the fact that by construction, the components of $|\psi_{\text{\rm \tiny AD}}\rangle$ contain no complex numbers other than $q$.
  It follows that 
  \begin{equation}
  \label{eq:psiZinequality}
    |\psi_{\text{\tiny AD}}(w_1,\dots,w_N;q)_{\sigma}|^2 \le Z(y=1;w_1,\dots,w_N)
  \end{equation}
   for any spin configuration $\sigma$. The right-hand side of this inequality is finite in the homogeneous limit according to \cref{prop:ZADhom}. We conclude that the rational function $\psi_{\text{\tiny AD}}(w_1,\dots,w_N;q)_{\sigma}$ is bounded in a sufficiently small real neighbourhood of the point $w_1=\cdots=w_N=1$. Hence, its homogeneous limit exists along any path where the $w_j$ are real.
   
  The proof also holds for the special values $q = q_c$ where $(q_c)^4=1$. Indeed, taking the limit $w_1, \dots, w_N \rightarrow 1$ and subsequently the limit $q \rightarrow q_c$ (as in \eqref{eq:specialhomog}) of the right-hand side of \eqref{eq:psiZinequality} produces a finite result, and the rest of the argument carries through.
\end{proof}
\end{proposition}
We note that the proof can easily be adapted to $q\in \mathbb R^\times$, that is for $x\in \mathbb R$ with $|x|\ge 2$. In that case however, the limit $w_1,\dots,w_N\to 1$ must be taken along paths with $|w_1|=\cdots= |w_N|=1$. 
This proves the existence of certain specific homogeneous limits of $|\psi_{\text{\tiny AD}} \rangle$,  for all $x \in \mathbb R$. Although we presently lack a complete proof, we believe that for all $x \in \mathbb C$, the limit in \eqref{eq:phiADdef} exists in full generality.

Our next goal is to investigate the polynomiality (in $x$) of each component of $|\phi_{\text{\tiny AD}}\rangle$. We proceed by taking a particular homogeneous limit, along the curve $w_j = t^j, j = 1, \dots, N$.
\begin{proposition}\label{prop:exist4allq}
The limit
\begin{equation}
\lim_{t \to 1} |\psi_{\text{\rm \tiny AD}}(t, t^2, \dots,t^N)\rangle\label{eq:limit2}
\end{equation}
exists for all $q \in \mathbb C^\times$.
\end{proposition}
\begin{proof}
Each component $\psi_{\text{\rm \tiny AD}}(w_1, w_2, \dots, w_N; q)_\sigma$ can be written the ratio of two polynomials in $N+1$ variables: $q$~and the $w_j$, $j = 1, \dots, N$. From the explicit construction of $|\psi_{\text{\tiny AD}}(w_1, \dots, w_N)\rangle$, we find that the polynomial in the denominator has the form
\begin{equation}\label{eq:DENOM}
q^k \prod_{i=1}^N(w_i)^{k_i} \prod_{1\le i<j\le N} [w_i/w_j]^{n_{ij}} \prod_{i,j=1}^N [q w_i/w_j]^{m_{ij}} [q^2 w_i/w_j]^{\ell_{ij}}
\end{equation}
where $k, k_i, n_{ij}, m_{ij}$ and $\ell_{ij}$ are positive integers that depend on $\sigma$. 
We now specialise the inhomogeneity parameters to $w_j = t^j$, $j = 1, \dots, N$, where $t$ is a free parameter. For every $\sigma$, the function 
\begin{equation}
f(q,t) = \psi_{\text{\rm \tiny AD}}(t, t^2, \dots, t^N; q)_\sigma
\end{equation}
is a rational function in both $q$ and $t$ and is therefore expressible as
\begin{equation}
f(q,t) = \frac{p_1(q,t)}{p_2(q,t)}
\end{equation}
where $p_1(q,t)$ and $p_2(q,t)$ are polynomials in the two variables $q$ and $t$.

The expression \eqref{eq:DENOM} for the denominator of $\psi_{\text{\rm \tiny AD}}(w_1, w_2, \dots, w_N; q)_\sigma$
allows us to write
\begin{equation}\label{eq:p2}
p_2(q,t) = q^{k} t^{k'} \prod_{1\le i<j \le N} [t^{i-j}]^{n_{ij}} \prod_{i,j=1}^N [q t^{i-j}]^{m_{ij}} [q^2t^{i-j}]^{\ell_{ij}}, \quad k' = \sum_{i=1}^N i\, k_i.
\end{equation}
Each $[t^{i-j}]$ contributes a single pole at $t=1$, so the total degree $n$ of that pole is given by $n = \sum_{1\le i<j\le N}n_{ij}$.

Because $p_1(q,t)$ is a polynomial in $t$, we can expand it around $t=1$ as
\begin{equation}\label{eq:p1}
p_1(q,t) = \sum_{k=0}^d (t-1)^k p_{1,k}(q),
\end{equation}
where each $p_{1,k}(q)$ is a polynomial in $q$. 

By \cref{prop:existencePhiAD}, the limit \eqref{eq:limit2} exists for $|q|=1$. This implies that the limit
\begin{equation}\label{eq:limitcircle}
\lim_{t \rightarrow 1} f(q = e^{\i \theta},t) = \kappa(e^{\i \theta})
\end{equation}
exists for all $\theta \in \mathbb R$.
Indeed, the case where the limit $t\rightarrow 1$
is taken along a real path is covered by \cref{prop:existencePhiAD}, so the limit exists along such a path. Because $f(e^{\i \theta},t)$ is a rational function of $t$, the limit \eqref{eq:limitcircle} exists even along complex paths.

Rewriting this in terms of \eqref{eq:p2} and \eqref{eq:p1}, we find
\begin{alignat}{2}
\kappa(e^{\i \theta}) 
&= \lim_{t \rightarrow 1} \frac{\sum_{k=0}^d (t-1)^k p_{1,k}(e^{\i \theta})}{(e^{\i\theta})^{k} t^{k'} \prod_{i<j} [t^{i-j}]^{n_{ij}} \prod_{i,j=1}^N [e^{\i \theta} t^{i-j}]^{m_{ij}} [(e^{\i \theta})^2t^{i-j}]^{\ell_{ij}}} \nonumber\\
&= \frac{1}{\alpha(e^{i \theta})} \lim_{t\rightarrow1} \sum_{k=0}^d (t-1)^{k-n} p_{1,k}(e^{\i \theta}),
\end{alignat}
with $\alpha(e^{i \theta})$ given by
\begin{equation}
\alpha(e^{\i \theta}) = e^{k\i \theta} \Big(\prod_{i<j} \big(2(i-j)\big)^{n_{ij}} \Big) \prod_{i,j=1}^N[e^{\i \theta}]^{m_{ij}} [e^{2\i \theta}]^{\ell_{ij}}.
\end{equation}
For $\theta \in \mathbb R$, $\alpha(e^{\i \theta})$ is zero only for $e^{\i \theta} = \pm 1, \pm \i$.

It follows that $p_{1,k}(q=e^{\i \theta}) = 0$ for $k = 0, \dots, n-1$, for a continuum of values of $q$ on the unit circle. Because each $p_{1,k}(q)$ is a polynomial in $q$, it can have only finitely many isolated zeroes on the unit circle, and therefore $p_{1,k}(q) = 0$ for $k = 0, \dots, n-1$, for all $q \in \mathbb C$. As a consequence, the limit
\begin{equation}
\label{eq:fqt}
\lim_{t\rightarrow1} f(q,t) = \frac{1}{\alpha(q)} \lim_{t \rightarrow 1} \sum_{k = n}^d (t-1)^{k-n} p_{1,k}(q) = \frac{p_{1,n}(q)}{\alpha(q)}
\end{equation} 
exists for all $q \in \mathbb C^\times$, and along any complex path for $t$. The seemingly problematic values $q = \pm \i, \pm 1$ were already covered by \cref{prop:existencePhiAD}, whereas the point $q = 0$ is excluded because $\alpha(q=0) = 0$. 
\end{proof}
The proof above uses the special choice $w_j = t^j$, but can in fact be repeated for any choice of rational functions $w_j = g_j(t)$ satisfying these conditions: \textit{(i)} $g_j(t) \in \mathbb R$ for $t \in \mathbb R$, \textit{(ii)} $g_j(1) = 1$ and \textit{(iii)} $g_i(t) \neq g_j(t)$ for $i \neq j$. The extension of the above proof to any curve $w_j = g_j(t)$ is probably feasible, but is beyond the scope of the current work.

That each component of $|\phi_{\textrm{\tiny AD}}\rangle$ is polynomial in $x$ is now straightforward. In the proof of \cref{prop:exist4allq}, we found in \eqref{eq:fqt} a rational expression for each $(\phi_{\text{\rm \tiny AD}})_\sigma$ which
has no singularities for $q \in \mathbb C^\times$.
Each component is thus a Laurent polynomial in $q$: 
\begin{equation}(\phi_{\text{\rm \tiny AD}})_\sigma = \sum_{k=-r}^{s} a_k q^k.\end{equation} From \eqref{prop:qtoqinv}, we find that 
\begin{equation}
(\phi_{\text{\rm \tiny AD}})_\sigma\big|_{q \rightarrow q^{-1}} = (\phi_{\text{\rm \tiny AD}})_\sigma,
\end{equation} 
which implies that $r = s$ and $a_k = a_{-k}$ for $k = 1, \dots, s$. One can then express $(\phi_{\text{\rm \tiny AD}})_\sigma$ in terms of $x$ by using $q^k + q^{-k} = 2T_k(\frac x 2)$, where $T_k(x)$ is the $k$-th Chebyshev polynomial of the first kind. 
This ends our discussion of the polynomiality of $|\phi_{\textrm{\tiny AD}}\rangle$.

%%%%%%%%%
\subsection{Mixed scalar products}
\label{sec:overlap}
%%%%%%%%%

\subsubsection{Inhomogeneous case}

Let us consider the mixed scalar product between the eigenvectors of the transfer matrices with diagonal and anti-diagonal twists:
\begin{equation}
  Z_{\text{\rm \tiny M}}(w_1,\dots, w_N)=\langle \psi_{\text{\rm \tiny{D}}}| \psi_{\text{\rm \tiny{AD}}}\rangle.
\end{equation}
Because the transfer matrices with the diagonal and anti-diagonal twist are seemingly unrelated, this is definitely not a natural scalar product to investigate. It nevertheless turns out to be interesting because of its connection with quarter-turn symmetric ASMs.

For odd $N$, this scalar product
vanishes according to \cref{prop:magnetisation} because the two states belong to different eigenspaces of the operator $(-1)^M$.
For even $N =2n$, we will show that $Z_{\text{\rm \tiny M}}(w_1,\dots, w_{2n})$ coincides with the partition function $Z_{\text{\rm \tiny{QT}}}(w_1,\dots,w_{2n})$ of the six-vertex model on an $2n\times 2n$ square with quarter-turn boundary conditions. The corresponding inhomogeneous lattice is illustrated in \cref{fig:6vqt} for $N=4$.
The vertex weights are given by the functions $\mathfrak a(z),\mathfrak b(z),\mathfrak c(z)$ in \cref{fig:6v}. Kuperberg 
\cite{kuperberg:02}
showed that this partition function can be written as
\begin{subequations}
\begin{equation}
  Z_{\text{\rm \tiny{QT}}}(w_1,\dots,w_{2n}) = [q^2]^n [q]^{3n} Z_{\text{\rm \tiny{QT}}}^{(1)}(w_1,\dots,w_{2n})Z_{\text{\rm \tiny{QT}}}^{(2)}(w_1,\dots,w_{2n})
\end{equation}
\end{subequations}
where
\begin{equation} 
Z_{\text{\rm \tiny{QT}}}^{(k)}(w_1,\dots,w_{2n}) = \Bigg(\prod_{1\le i<j \le 2n}\frac{[q w_i/w_j][qw_j/w_i]}{[w_j/w_i]}\Bigg)\Pf{i,j=1}{2n}\left(\frac{[(w_j/w_i)^k]}{[q w_i/w_j] [q w_j/w_i]
}\right).
\label{eqn:defZQTk}
\end{equation}
\begin{figure}[h]
\centering
\begin{tikzpicture}
 \draw (0,0) grid[step=.5cm] (1.5,1.5);
    \foreach \x in {0,0.5,1,1.5}
    {
      \draw[postaction={on each segment={mid arrow}}] (-.5,\x)--(0,\x);
      \draw (1.5,\x)--(1.75,\x);
      \draw[postaction={on each segment={mid arrow}}] (\x,0)--(\x,-0.5);
      \draw (\x,1.5)--(\x,1.75);
    }
    
    \draw[densely dotted] (1.75,1.75) -- (3.25,3.25);
    
    \draw (1.75,1.5) arc [start angle=-90, end angle = 180, radius=.25cm];
    \draw (1.75,1.) arc [start angle=-90, end angle = 180, radius=.75cm];
    \draw (1.75,.5) arc [start angle=-90, end angle = 180, radius=1.25cm];
    \draw (1.75,0) arc [start angle=-90, end angle = 180, radius=1.75cm];
   
    \draw (0,-.5) node [below] {$w_1$};
    \draw (.5,-.5) node [below] {$w_2$};
    \draw (1.,-.5) node [below] {$w_3$};
    \draw (1.5,-.5) node [below] {$w_4$};
    
    \draw (-.5,0) node [left] {$w_1$};
    \draw (-.5,.5) node [left] {$w_2$};
    \draw (-.5,1) node [left] {$w_3$};
    \draw (-.5,1.5) node [left] {$w_4$};

\end{tikzpicture}
\caption{A $4 \times 4$ square lattice domain with quarter-turn boundary conditions. The dotted line indicates orientation reversal upon crossing.}
\label{fig:6vqt}
\end{figure}
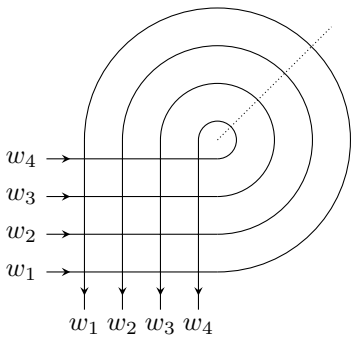

\begin{proposition}\label{prop:Z0ZQT}
  For $N=2n$, we have
  \begin{equation}
    Z_{\text{\rm \tiny{M}}}(w_1,\dots,w_{2n})
    =Z_{\text{\rm \tiny{QT}}}(w_1,\dots,w_{2n}).
  \end{equation}
  \begin{proof}
    The proof differs from most other results in this article involving partition functions of the six-vertex model, in the sense that it does not consist of a straightforward calculation. Instead, we follow Kuperberg \cite{kuperberg:02} who proved that $Z_{\text{\rm \tiny{QT}}}(w_1,\dots,w_{2n})$ is fixed up to a multiplicative constant by the following properties. First, it is an odd, symmetric function in its arguments. Second, with respect to each argument, it is a centred Laurent polynomial of degree width at most $2(4n-3)$. 
    Third, it obeys a special recurrence relation which relates the partition functions for $N=2n$ and $N' = 2n-2$.
    We establish these properties for $Z_{\text{\rm \tiny{M}}}(w_1,\dots,w_{2n})$ in \cref{cor:odd} and \cref{lem:ZOsym,lem:ZOlaurent,lem:ZOrecurrence}. The remaining constant is fixed by verifying that the statement holds for $n=1$.
  \end{proof}
\end{proposition}

\begin{lemma}
  \label{lem:ZOsym}
  $ Z_{\text{\rm \tiny M}}(w_1,\dots, w_N)$ is a symmetric function in the $w_j$.   \begin{proof}
    The proof is a straightforward application of the exchange relations found in \cref{prop:exchange}. The matrix $\check R(z)$ defined in \eqref{eqn:Rcheck} is symmetric and has the property
    \begin{equation}
      \check R(z^{-1})\check R(z) = \left([q/z][q^2/z][q z][q^2 z] \right)\bm{1}.
    \end{equation}
    Combining this with \cref{prop:transposition,prop:exchange}, we find
  \begin{align} 
    Z_{\text{\rm \tiny M}}(\dots, w_j,w_{j+1},\dots) & = \quad  \frac{\langle \psi_{\text{\rm \tiny D}}(\dots, w_j,w_{j+1},\dots)|\check R_{j,j+1}\left(\frac{w_{j+1}}{w_j}\right)}{[qw_{j+1}/w_j][q^2w_{j+1}/w_j]}\nonumber\\
    & \hspace{1.5cm}\qquad\times \frac{\check R_{j,j+1}\left(\frac{w_{j}}{w_{j+1}}\right)|\psi_{\text{\rm \tiny AD}}(\dots, w_j,w_{j+1},\dots)\rangle}{[qw_{j}/w_{j+1}][q^2w_{j}/w_{j+1}]}\\
    &=  \langle \psi_{\text{\rm \tiny D}}(\dots, w_{j+1},w_{j},\dots)|\psi_{\text{\rm \tiny AD}}(\dots,w_{j+1}, w_j,\dots)\rangle\nonumber\\
    & = Z_{\text{\rm \tiny M}}(\dots,w_{j+1}, w_j,\dots) \nonumber
  \end{align}
  for 
  $j=1,\dots,N-1$. 
  By repeatedly permuting pairs $(w_j,w_{j+1})$, we find that $ Z_{\text{\rm \tiny M}}(w_1,\dots, w_N)$ is invariant under the exchange of $w_i$ and $w_j$, for $i,j = 1, \dots, N$.
  \end{proof}
  
\end{lemma}

Next, we wish to establish the degree width as well as a recurrence relation for the scalar product $Z_{\text{\rm \tiny M}}(w_1,\dots, w_N)$. To this end, we derive an explicit summation formula.
  
\begin{lemma} The scalar product is given by
  \begin{align}
  Z_{\text{\rm \tiny M}}(w_1,\dots, w_{2n})  = \frac{\prod_{j=1}^{2n} a(w_j)}{\prod_{1\le i<j \le 2n }[w_i/w_j]} \sum_{\substack{I\sqcup J = \{1,\dots,2n\}\\ |I|=n}} & \epsilon(I,J) \left(\prod_{i \in I,\,j\in J}[q^2 w_i/w_j][q^{-1} w_i/w_j] \right)\nonumber\\
  & \times \det_{i \in I,\,j\in J}\left(\frac{[q^2]}{[q w_i/w_j][q w_j/w_i]}\right).
  \label{eqn:ZO}
\end{align}
The sum runs over the subsets $I$ of $\{1,\dots,2n\}$ of cardinality $n$,
and $I\sqcup J$ denotes the disjoint union of $I$ and $J$, meaning that $I \cap J = \emptyset$.
Furthermore,
 $\epsilon(I,J)$ is defined as
\begin{equation}
  \epsilon(I,J)=\mathrm{sgn}\left(
  \begin{array}{ccccc}
   1 & 2 & \cdots & 2n-1 & 2n\\
   i_1 & j_1 & \cdots & i_n & j_n
  \end{array}
  \right)
\end{equation}
where $i_1 < i_2< \cdots < i_n$ and $j_1 < j_2< \cdots < j_n$ are the ordered elements of $I$ and $J$.
\end{lemma}
\begin{proof}
Our aim is to compute
\begin{equation}\label{eq:scapro}
   \langle \psi_{\text{\rm\tiny{D}}}| \psi_{\text{\rm\tiny{AD}}}\rangle = \sum_{\bm h} S(\bm h) \psi_{\text{\tiny AD}}(\bm h)
\end{equation}
where
\begin{equation}\label{eq:ShandPsih}
  S(\bm h) = \langle \wedge|\prod_{j=1}^{2n} \mathcal C(w_j)||\bm h\rrangle,\qquad  \psi_{\text{\tiny AD}}(\bm h)  = \frac{\llangle \bm h||\psi_{\text{\rm \tiny AD}}\rangle}{\llangle \bm h||\bm h\rrangle}.
\end{equation}  
Because $\psi_{\text{\rm \tiny AD}}(\bm h) = 0$ if at least one of the heights is $1$, we focus on height profiles $\bm h = (h_1,\dots,h_{2n})$ with $h_j \in\{0,2\}$ for all $j=1,\dots,2n$. 
Under these restrictions, we use \eqref{eqn:scalarprod}, \eqref{eqn:projections} and \eqref{eq:special.identity} to write $\psi_{\textrm{\tiny AD}}(\bm h)$ as 
 \begin{equation}
 \label{eqn:simplePsiAD}
 \psi_{\textrm{\tiny AD}}(\bm h) = \prod_{i=1}^{2n} 
(-1)^{h_i/2}\prod_{1\le i<j\le 2n} \frac{[q^{h_i-h_j}w_i/w_j]}{[w_i/w_j]}, \qquad  h_1, \dots, h_{2n} \in \{0,2\}.
 \end{equation}
 Any $\bm h$ contributing to \eqref{eq:scapro}
is uniquely fixed by the index set $I(\bm h) =\{i: h_i=2\}$. Because $|\psi_{\text{\rm \tiny D}}\rangle$ has magnetisation zero, it is clear that $S(\bm h)$ is non-vanishing only if $\sum_{i=1}^{2n}(1-h_i)=0$. It is therefore sufficient to consider only the cases where $|I(\bm h)|=n$. In terms of the index set $I(\bm h)$,
 we have
\begin{equation}
   S(\bm h) = \frac{\langle \wedge|\left(\prod_{j=1}^{2n} \mathcal C(w_j)\right)\left(\prod_{i\in I(\bm h)} \mathcal B(w_i)\mathcal B(q w_i)\right)|\wedge\rangle}{\prod_{i \in I(\bm h)}a(w_i)a(qw_i)}.
   \label{eqn:Sh}
\end{equation}
The numerator can be evaluated with the help of \eqref{eqn:spikdet}, and thus be written in terms of the Izergin-Korepin determinant $Z_{\text{\tiny IK}}(z_1,\dots,z_N;w_1,\dots,w_N)$ 
given in \eqref{eqn:IKDet}. Here, the
set $\{z_i\}_{i=1}^n$ is given by $\{w_i\}_{i\in I(\bm h)}\cup \{q w_i\}_{i\in I(\bm h)}$, and we find
\begin{equation}
  S(\bm h) = \frac{\left(\prod_{j=1}^N a(w_j) \right) Z_{\text{\tiny IK}}(\{w_i\}_{i\in I(\bm h)},\{q w_i\}_{i\in I(\bm h)};w_1,\dots,w_N)}{\prod_{i \in I(\bm h)}a(w_i)a(qw_i)}.
\end{equation}
The result may be simplified further with the help of the known recursion formula for the Izergin-Korepin determinant:
\begin{equation}
  Z_{\text{\tiny IK}}(z_1=q w_1,z_2,\dots,z_N;w_1,\dots,w_N)=[q^2]\left(\prod_{i=2}^N [q z_i/w_1][q^2w_1/w_i]\right)Z_{\text{\tiny IK}}(z_2,\dots,z_N;w_2,\dots,w_N).
\end{equation}  
Applying this recursion formula $n$ times, we obtain after some algebra
\begin{equation}
  S(\bm h) = \frac{\left(\prod_{j=1}^N a(w_j) \right)Z_{\text{\tiny IK}}(\{w_i\}_{i\in I(\bm h)};\{w_j\}_{j\in J(\bm h)})}{\prod_{i \in I(\bm h),\,j\in J(\bm h)}[q w_i/w_j]}
\end{equation}
where $J(\bm h) = \{1,2,\dots,2n\}\backslash I(\bm h)$. Likewise, the coefficient $\psi_{\text{\tiny AD}}(\bm h)$ given in \eqref{eqn:simplePsiAD} can be written in terms of the index sets:
\begin{equation}
  \psi_{\text{\tiny AD}}(\bm h) = (-1)^n \epsilon(I(\bm h),J(\bm h))\frac{\prod_{i \in I(\bm h),\,j\in J(\bm h)}[q^2 w_i/w_j]}{\prod_{1\le i<j\le 2n}[w_i/w_j]}\prod_{\substack{i,i'\in I(\bm h)\\i<i'}}[w_i/w_{i'}]\prod_{\substack{j,j'\in J(\bm h)\\j<j'}}[w_{j'}/w_{j}].
\end{equation}
We insert these expressions into \eqref{eq:scapro} and replace the sum over $\bm h$ by a sum over all possible subsets $I$ of $\{1,\dots, 2n\}$ of cardinality $n$, with $J =\{1,\dots, 2n\} \setminus I$. Using the explicit form of the six-vertex-model partition function \eqref{eqn:IKDet}, the result of the proposition follows after some simplifications.
\end{proof}
The following corollary follows from a direct inspection of \eqref{eqn:ZO}.
\begin{corollary}\label{cor:odd}
$Z_{\text{\rm \tiny M}}(w_1,\dots,w_{2n})$ is an odd function in each of its arguments.
\end{corollary}
The next two lemmas establish the second and third properties of $Z_{\text{\rm \tiny M}}(w_1,\dots,w_{2n})$ used in the proof of \cref{prop:Z0ZQT}.

\begin{lemma}
  \label{lem:ZOlaurent}
  $Z_{\text{\rm \tiny M}}(w_1,\dots,w_{2n})$ is a Laurent polynomial of degree width at most $2(4n-3)$ in each $w_j,\,j=1,\dots,2n$.
\begin{proof}

  We consider $Z_{\text{\rm \tiny M}}(w_1,\dots,w_{2n})$ as a function of $w_1$, which is sufficient as it is symmetric in its arguments.
  
  First, we show that $Z_{\text{\rm \tiny M}}(w_1,\dots,w_{2n})$
  is a rational function with its only pole at $w_1=0$.
  Indeed, from the summation formula \eqref{eqn:ZO}, it is clear the only other potential poles are
   \textit{(i)} $w_1 = \pm q^{1}w_j, \pm q^{-1}w_j$ and \textit{(ii)} $w_1=\pm w_j$ for $j=2,\dots,2n$. 
   We now proceed to show that these poles are spurious, meaning that $Z_{\text{\rm \tiny M}}(w_1,\dots,w_{2n})$ is in fact regular at these points. 
   Because of \cref{cor:odd}, we investigate only the values where the multiplicative sign is positive.
   Case \textit{(i)} is ruled out because as $w_1 \to q^{\pm 1}w_j$, any divergence coming from the determinants in \eqref{eqn:ZO} is compensated by the vanishing prefactor $\prod_{i=1}^{2n}a(w_i)=\prod_{i,j=1}^{2n}[qw_i/w_j]$. For case \textit{(ii)}, it is sufficient to consider $w_1 \to w_2$ because of \cref{lem:ZOsym}.
  The prefactor of \eqref{eqn:ZO} diverges because it has a simple pole at $w_1= w_2$. We thus need to show that the sum vanishes. To this end, we note that the terms of the sum with $1,2\in I$ or $1,2\in J$ vanish individually linearly because the determinants have two equal lines at $w_1=w_2$. The other terms can be grouped into pairs $I=\{1,i_2,\dots,i_n\}, J=\{2,j_2,\dots,j_n\}$ and $I'=\{2,i_2,\dots,i_n\}, J'=\{1,j_2,\dots,j_n\}$. Clearly, $\epsilon(I,J) = - \epsilon(I',J')$. The summands in \eqref{eqn:ZO} corresponding to $I,J$ and $I',J'$ thus cancel pairwise in the limit $w_1\to w_2$, which compensates for
  the divergence of the prefactor.  
We conclude that $Z_{\text{\rm \tiny M}}(w_1,\dots,w_{2n})$, being a rational function of $w_1$ with its only pole at $w_1=0$, is in fact a Laurent polynomial in $w_1$.

Second, we find upper and lower bounds for the extremal powers of $w_{1}$ of this Laurent polynomial by analysing the behaviour of $Z_{\text{\rm \tiny M}}(w_1,\dots,w_{2n})$ around $w_1=0$ and $w_1=\infty$. In each case, we replace the factors in \eqref{eqn:ZO} by their most divergent part and count the overall power of $w_1$. The sum of these divergent parts (over the subsets $I$ of $\{1, \dots, 2n\}$ of cardinality $n$) may vanish, so the resulting power-counting only gives a bound on the degree. For the minimal power, this bound is found by considering the limit $w_1\to 0$:
the prefactor diverges like $w_1^{-2n+1}$ and each term in the sum like $w_1^{-2n+2}$ (the determinant vanishes quadratically as $w_1$ approaches zero). We conclude that for $w_1\to 0$, the function $Z_{\text{\rm \tiny M}}(w_1,\dots,w_{2n})$ has a pole of order at most $4n-3$. Similarly, by taking $w_1\to \infty$, we conclude that the leading term is proportional to $w_1^k$ for some $k \le 4n-3$. The Laurent polynomial therefore has degree width at most $2(4n-3)$.
\end{proof}
\end{lemma}

\begin{lemma}
\label{lem:ZOrecurrence}
$Z_{\text{\rm \tiny M}}(w_1,\dots,w_{2n})$ obeys the recurrence relation
 \begin{equation}
   \frac{Z_{\text{\rm \tiny M}}(w_1,\dots,w_{2n-1},w_{2n}=q w_{2n-1})}{Z_{\text{\rm \tiny M}}( w_1,\dots,w_{2n-2})}= \left([q][q^2]\prod_{i=1}^{2n-2} [q^{-1} w_{2n-1}/w_i][q^2 w_{2n-1}/w_i]\right)^2.
 \end{equation} 
\end{lemma}
\begin{proof}
  According to \cref{lem:ZOlaurent}, $Z_{\text{\rm \tiny M}}(w_1,\dots,w_{2n})$ is well defined for $w_{2n}=q w_{2n-1}$. For this specification, the prefactor in \eqref{eqn:ZO} vanishes, whereas the sum contains both regular and divergent terms. A limit must therefore be taken and the regular terms of the sum yield a vanishing contribution. Conversely, the divergent terms have $2n-1\in I$ and $2n \in J$ and produce a non-zero contribution. For these cases, we write $I = I' \,\sqcup \,\{2n-1\}, J = J'\, \sqcup \,\{2n\}$ and note that $\epsilon(I,J) = \epsilon(I',J')$. A direct calculation yields
  \begin{equation} \label{eq:ingredient1}
    \prod_{\substack{i \in I\\\,j\in J}}\left[\frac{q^2 w_i}{w_j}\right]\left[\frac{q^{-1} w_i}{w_j}\right]  =  -[q][q^2]\prod_{i=1}^{2n-2}\left[\frac{q^{-1} w_{2n-1}}{w_i}\right]\left[\frac{q^2 w_{2n-1}}{w_i}\right]\prod_{\substack{i \in I'\\\,j\in J'}}\left[\frac{q^2 w_i}{w_j}\right]\left[\frac{q^{-1} w_i}{w_j}\right].  \end{equation} 
    Likewise, combining the prefactor and the determinant, we obtain after some algebra
    \begin{align} 
      &\lim_{w_{2n}\to q w_{2n-1}}\frac{\prod_{i,j=1}^{2n}[q w_i/w_j]}{\prod_{1\le i<j\le 2n}[w_i/w_j]} \det_{i \in I,\,j\in J}\left(\frac{[q^2]}{[q w_i/w_j][q w_j/w_i]}\right)\label{eq:ingredient2}\\
      &=-\left([q][q^2]\prod_{i=1}^{2n-2}\left[\frac{q^{-1} w_{2n-1}}{w_i}\right]\left[\frac{q^2 w_{2n-1}}{w_i}\right]\right)\frac{\prod_{i,j=1}^{2n-2}[q w_i/w_j]}{\prod_{1\le i<j \le 2n-2}[w_i/w_j]} \det_{i \in I',\,j\in J'}\left(\frac{[q^2]}{[q w_i/w_j][q w_j/w_i]}\right).
     \nonumber 
    \end{align}
    Using \eqref{eq:ingredient1} and \eqref{eq:ingredient2}, we can now take the limit of \eqref{eqn:ZO}. As discussed above, the resulting sum ranges over sets $I'$ of $\{1, \dots, 2n-2\}$ of cardinality $n-1$. The desired recurrence relation readily follows.
\end{proof}

\subsubsection{Homogeneous limit}

The results of the previous section allow us to prove \cref{thm:scalarprod} by taking the homogeneous limit. Kuperberg \cite{kuperberg:02} showed that $Z_{\text{\rm \tiny QT}}(1, \dots, 1)= [q^2]^{n}[q]^{4n^2-n} A_{\text{\rm \tiny QT}}(4n;x^2)$, see \cref{sec:qsumrules} for the definition of $A_{\text{\rm \tiny QT}}(4n;x^2)$.
Using this relation and the definitions 
 of $|\phi_{\text{\rm \tiny D}}\rangle$ and $|\phi_{\text{\rm \tiny AD}}\rangle$, we obtain
\begin{equation} 
  \langle\phi_{\text{\rm \tiny D}}|\phi_{\text{\rm \tiny AD}}\rangle 
  = A_{\text{\rm \tiny QT}}(4n;x^2) =
  A^{(1)}_{\text{\rm \tiny QT}}(4n;x^2)A^{(2)}_{\text{\rm \tiny QT}}(4n;x^2),
\end{equation}
where the $A^{(k)}_{\text{\rm \tiny QT}}(4n;x^2)$ 
are given by
\begin{equation}
  A^{(k)}_{\text{\rm \tiny QT}}(4n;x^2)= [q]^{-2n(n-1)} 
  Z^{(k)}_{\text{\rm \tiny QT}}(1,\dots,1).
\end{equation}
We provide a simple pfaffian expression for $ A^{(1)}_{\text{\rm \tiny QT}}(4n;x^2)$.
\begin{proposition}
We have
\begin{subequations}
\begin{align}
 A^{(1)}_{\text{\rm \tiny QT}}(4n;x^2) = \Pf{i,j=0}{2n-1}\left((-1)^i \binom{i}{j}x^{i-j-1}-(-1)^j \binom{j}{i}x^{j-i-1}\right)
\end{align}
\end{subequations}
\begin{proof}
We use the explicit form of the partition function \eqref{eqn:defZQTk} together with a change of variables $u_i = w^2_{i+1}$ for $i=0,\dots,2n-1,$ in order to write
\begin{align}
 A^{(1)}_{\text{\rm \tiny QT}}(4n;x^2) & = 
[q]^{2n^2} \lim_{w_1,\dots,w_{2n}\to 1} \Bigg(\prod_{1\le i<j \le 2n}\frac{1}{[w_j/w_i]}\Bigg)\Pf{i,j=1}{2n}\left(\frac{[w_j/w_i]}{[q w_i/w_j][q w_j/w_i]}\right)\\
  & = [q]^{2n^2}
  \lim_{u_0,\dots,u_{2n-1}\to 1}  \Bigg(\prod_{0\le i<j \le 2n-1}\frac{1}{(u_j-u_i)}\Bigg)\Pf{i,j=0}{2n-1}
  f(u_i,u_j).
  \label{eqn:AQTIntermediate}
\end{align}
Here we abbreviated
\begin{equation}
    f(u,v) = \frac{u-v}{(q^2 u-v)(q^{-2}u-v)}.
\end{equation}
The prefactor in \eqref{eqn:AQTIntermediate} can be included into the pfaffian by using the identity $\text{pf}\,\left( MAM^t\right)= (\det M )\text{pf}\, A$ for $A_{ij}=f(u_i,u_j)$ and
\begin{equation}
  M_{ij} =
  \begin{cases}
   \prod_{\substack{k=0\\k\neq j}}^i(u_j-u_k)^{-1},& i \ge j,\\    0, & i < j.
  \end{cases}
\end{equation}
The result is the pfaffian of a divided difference,
\begin{equation}
  A_{\text{\rm \tiny QT}}^{(1)}(4n;x^2) = [q]^{2n^2}\lim_{u_0,\dots,u_{2n-1}\to 1} \Pf{i,j=0}{2n-1}
  f[u_0,\dots,u_i;u_0,\dots,u_j],
\end{equation}
see \eqref{eqn:homlim}.
This expression allows us to take the homogeneous limit and apply the techniques of Behrend \textit{et al.} \cite{behrend:12}, reviewed in \cref{prop:ZADhom}. From \eqref{eqn:divideddiff} with $r=1$, we obtain
\begin{equation}
  A_{\text{\rm \tiny QT}}^{(1)}(4n;x^2)=[q]^{2n^2}\Pf{i,j=0}{2n-1}\left(\{u^i v^j\}f(u+1,v+1)\right).
\end{equation}

 The series coefficients of $f(u+1,v+1)$ in a power series in $u$ and $v$ are given by
  \begin{equation}
   \{u^i v^j\}f(u+1,v+1) = \binom{i+j}{i}\frac{(-1)^i q^{i-j}-(-1)^j q^{j-i}}{(q+q^{-1})(q-q^{-1})^{i+j+1}}.
  \end{equation}
  We rewrite this expression in terms of the matrices $L(\alpha,\beta)$ defined in \cref{prop:ZADhom}. 
 Using the definition \eqref{eq:alphabeta} of $\alpha_\pm$, we have
  \begin{align}
    A^{(1)}_{\text{\rm \tiny QT}}(4n;x^2) & = [q]^{n(2n-1)}\mathop{\text{pf}}\left(x^{-1}(L(\alpha_-,1)L(\alpha_+,1)^t-L(\alpha_+,1)L(\alpha_-,1)^t)\right)\\
    & = [q]^{n(2n-1)} \det L(\alpha_+,1) \mathop{\text{pf}}\left(x^{-1}(L(\alpha_+,1)^{-1}L(\alpha_-,1)-(L(\alpha_+,1)^{-1}L(\alpha_-,1))^t\right).\nonumber
  \end{align}
  We simplify this expression by using the determinant \eqref{eqn:propL0}, the property $\det L(\alpha_+,1) = (q^{-1}/[q])^{n(2n-1)}$ and the product formula \eqref{eqn:propL2}. The latter yields
  $L(\alpha_+,1)^{-1}L(\alpha_-,1)=L(\alpha_0,\beta_0)$ with $\alpha_0 = -q x,\,\beta_0=q/x$, and thus
\begin{align}
    A^{(1)}_{\text{\rm \tiny QT}}(4n;x^2) & = q^{-n(2n-1)}\mathop{\text{pf}}\left(x^{-1}(L(\alpha_0,\beta_0)-L(\alpha_0,\beta_0)^t)\right).
  \end{align}
  The result for $A^{(1)}_{\text{\rm \tiny QT}}(4n;x^2)$ follows from using the explicit definition \eqref{eqn:defL} and the relation
  \begin{equation}
    \Pf{i,j=0}{2n-1} g_ig_jf_{ij} = 
 \left(\prod_{i=0}^{2n-1}g_i\right)
  \Pf{i,j=0}{2n-1}f_{ij}.
  \end{equation}
  \end{proof}
\end{proposition}
This concludes the proof of \cref{thm:scalarprod}. The polynomial $A_{\text{\rm \tiny QT}}^{(2)}(4n;x^2)$ can be evaluated along the same lines, using the function $\tilde f(u,v)=(u+v)f(u,v)$. We obtain
\begin{equation}
  A_{\text{\rm \tiny QT}}^{(2)}(4n;x^2)=\Pf{i,j=0}{2n-1}\left(a_{ij}-a_{ji}\right)
\end{equation}
where
\begin{equation}
  a_{ij}= (-1)^i q^{i-j}\left(\binom{i+j-1}{i}q +\binom{i+j-1}{j}q^{-1} \right).
\end{equation}
Although $A_{\text{\rm \tiny QT}}^{(2)}(4n;x^2)$ is known to be a polynomial in $x^2$, we have not managed to find an expression in terms of 
the pfaffian of a matrix with polynomial entries in $x$.

%%%%%%%%%%%%%%%%%%%%%
%
\section{Linear sum rules and special components}
\label{sec:specialcomps}
%
%%%%%%%%%%%%%%%%%%%%%

The goal of this section is to obtain the expressions for the special components of the spin-chain zero-energy states given in \cref{sec:specialcomponents}. In \cref{sec:ZA}, we introduce an auxiliary partition function for the ten-vertex model which will arise in certain scalar products involving $|\psi_{\text{\rm \tiny D}}\rangle$ and $|\psi_{\text{\rm \tiny AD}}\rangle$. \cref{sec:DTwist,sec:ADTwist} treat these scalar products and their homogeneous limits for the diagonal and anti-diagonal twist, respectively.

%%%%%%%%%
\subsection{A partition function for the ten-vertex model}
\label{sec:ZA}
%%%%%%%%%

In this section, we present a partition function for the ten-vertex model, which will allow us to compute certain components of the vectors $|\psi_{\text{\rm \tiny D}}\rangle$ and $|\psi_{\text{\rm \tiny AD}}\rangle$. We consider systems of even size $N = 2n$, and specialise the inhomogeneity parameters of the transfer matrices to the alternating sequence 
\begin{equation}
  w_{2i-1}=x_i, \quad w_{2i}=x_i^{-1}, \quad i=1,\dots,n.
  \label{eqn:alternatingparams}
\end{equation}
For two sites, we introduce the state
\begin{equation}\label{eq:chistate}
  |\chit(z)\rangle = [b z][bq z]|{\Uparrow\Downarrow}\rangle + [b^{-1}qz][bq z]|{00}\rangle + [b^{-1}z][b^{-1}q z]|{\Downarrow\Uparrow}\rangle,
\end{equation}
where $b$ is an arbitrary parameter. As discussed in \cref{app:BYBE}, $|\chit(z)\rangle$ is a solution to the boundary Yang-Baxter equation.
We write $|\chit(x_1,\dots,x_n)\rangle = \bigotimes_{j=1}^n|\chit(x_j)\rangle$ for $N=2n$ sites. Let us define the scalar product
\begin{equation}
  Z_{\text{\rm \tiny A}}(\{x_i\}_{i=1}^n;\{y_i\}_{i=1}^{2n}) = \langle \chit(x_1,\dots, x_n)|\prod_{j=1}^{2n}\mathcal B(y_j)|\wedge\rangle.
\end{equation}
This quantity is the partition function of the ten-vertex model on a $2n\times 2n$ square with U-turn boundary conditions at the top as shown on \cref{fig:spinonepartfunc}.
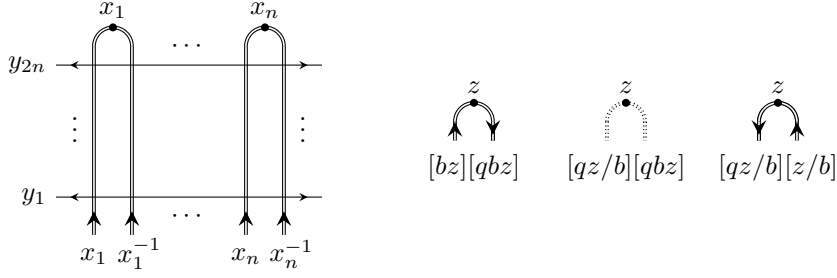
\begin{figure}  
  \centering
  \begin{tikzpicture}
     
     \foreach \x in {0,.5}
     {  
       \draw[double] (\x,0) -- (\x,2);
       \draw[double,xshift=2cm] (\x,0) -- (\x,2);
       
       \draw[double,postaction={on each segment={mid arrow}}] (\x,-.5) -- (\x,0);
       \draw[double,xshift=2cm,postaction={on each segment={mid arrow}}] (\x,-.5) -- (\x,0);
     }
     
     \foreach \x in {0,2}
     {  
        \draw [double,xshift=\x cm] (.5,2) arc [start angle=0, end angle = 180, radius=.25cm];
        \fill[xshift=\x cm] (0.25,2.25) circle (1.5pt) ;
     }
        
     \draw (0.25,2.25) node[above] {$x_1$} ;
     \draw (2.25,2.25) node[above] {$x_n$} ;
     
     \foreach \x in {0,3}
       \draw[xshift=\x cm] (-.25,1) node {$\vdots$};
     
     \foreach \y in {0,2.25}
       \draw[yshift=\y cm] (1.25,-.25) node {$\cdots$};

     \foreach \y in {0,1.75}
     {
     \draw[postaction={on each segment={mid arrow}},yshift=\y cm] (0,0)--(-0.5,0);
     \draw[postaction={on each segment={mid arrow}},yshift=\y cm] (2.5,0)--(3,0);
     \draw[yshift=\y cm]  (0,0)--(2.5,0);
     }
     
     \draw (-.5,0) node [left] {$y_1$};
     \draw (-.5,1.75) node [left] {$y_{2n}$};
     \draw (0,-.79) node {$x_1$};
     \draw (.6,-.75) node  {$x_1^{-1}$};
      \draw (2,-.79) node {$x_n$};
     \draw (2.6,-.75) node  {$x_n^{-1}$};
     
     \begin{scope}[xshift=5cm,yshift=1cm]
          \draw[double] (0.25,0) arc [start angle=0, end angle = 90, radius=.25cm];
          \draw[double] (0,.25) arc [start angle=90, end angle = 180, radius=.25cm];
          \fill (0,0.25) circle (1.5pt) ;
          \draw (0,.25) node [above] {$z$};
          \draw (0,-0.90) node [above] {$[bz][qbz]$};
          \draw[double] (0.25,-0.25) -- (0.25,0) (-0.25,-0.25) -- (-0.25,0);
          \draw[double,postaction={mid arrow}] (-0.25,-0.01) -- (-0.25,0);
          \draw[double,postaction={mid arrow}] (0.25,-0.05) -- (0.25,-0.25);          
    \end{scope}
     
     \begin{scope}[xshift=7cm,yshift=1cm]
          \draw[double,densely dotted] (-0.25,0) arc [start angle=180, end angle = -5, radius=.25cm];
          \fill (0,0.25) circle (1.5pt) ;
          \draw (0,.25) node [above] {$z$};
          \draw (0,-.90) node [above] {$[qz/b][qbz]$};
           \draw[double,densely dotted] (-0.25,0) -- (-0.25,-0.25);
           \draw[double,densely dotted] (0.25,0) -- (0.25,-0.25);      
    \end{scope}

    \begin{scope}[xshift=9cm,yshift=1cm]
          \draw[double] (-0.25,0) arc [start angle=180, end angle = 0, radius=.25cm];
          \draw[double] (-0.25,0) arc [ start angle=180, end angle = 140, radius=.25cm];
          \draw[double] (0,.25) arc [start angle=90, end angle = 0, radius=.25cm];
          \fill (0,0.25) circle (1.5pt) ;
          \draw (0,.25) node [above] {$z$};
          \draw (0,-.90) node [above] {$[qz/b][z/b]$};
          \draw[double] (0.25,-0.25) -- (0.25,0) (-0.25,-0.25) -- (-0.25,0);
          \draw[double,postaction={mid arrow}] (0.25,-0.01) -- (0.25,0);
          \draw[double,postaction={mid arrow}] (-0.25,-0.05) -- (-0.25,-0.25);      
    \end{scope}
    \end{tikzpicture}
  \caption{ A
  $2n\times 2n$ square for the ten-vertex model with U-turn boundary conditions and the statistical weights for the three admissible arrow orientations along a vertical U-turn.}
   \label{fig:spinonepartfunc}
\end{figure}
\begin{proposition}\label{prop:ZAspin1}
  The partition function $Z_{\text{\rm \tiny A}}(\{x_i\}_{i=1}^n;\{y_i\}_{i=1}^{2n})$ is given by
\begin{subequations}\label{eqn:ZA}
\begin{align}
  Z_{\text{\rm \tiny A}}(\{x_i\}_{i=1}^n;\{y_i\}_{i=1}^{2n})=
  & (-1)^n \frac{[q]^{2n}
  \prod_{i=1}^{2n}[b/y_i]\prod_{i=1}^n [q x_i^2][q^2 x_i^2]\prod_{i=1}^n\prod_{j=1}^{2n}\prod_{k=-1}^1[q^k x_i y_j][q^k x_i/y_j]}{\prod_{1\le i<j \le 2n}[\xi_i/\xi_j]
  [y_i/y_j]\prod_{i\le j}[q \xi_i \xi_j] 
  [y_i y_j]}\nonumber \\
  & \qquad \times \Det{i,j=1}{2n} M_{ij}
\end{align}
with  $\xi_{2i-1} = x_i$, $\xi_{2i} = q^{-1}x_i$ for $i = 1, \dots, n$,
and the matrix elements
\begin{equation}
  M_{ij} = \frac{1}{[q \xi_i/y_j][\xi_i/y_j]}-\frac{1}{[q \xi_iy_j][\xi_iy_j]}.
  \label{eqn:ZAmatrix}
\end{equation}
\end{subequations}
\end{proposition}
  \noindent The proof relies on the fusion construction of the ten-vertex model and on a known partition function of the six-vertex model,
 and is given in \cref{app:ZA}.

Let us list some properties of the partition function which will be of use in the following sections. A very elementary consequence of the explicit formula is the following transformation behaviour under $y_j\to y_j^{-1}$.
\begin{lemma}
  The partition function satisfies
  \begin{equation}
    Z_{\text{\rm \tiny A}}(\{x_i\}_{i=1}^n;\dots, y^{-1}_j,\dots)= \frac{[b y_j]}{[b/y_j]}Z_{\text{\rm \tiny A}}(\{x_i\}_{i=1}^n;\dots, y_j,\dots).
  \end{equation}
\end{lemma}
Another property, the so-called \emph{wheel condition}, can be proved directly from the graphical representation.
\begin{lemma}
\label{lemma:wheel}
  Let $n\ge 2$. For $j=1,\dots,n$,
  we have
  \begin{equation}
   Z_{\text{\rm \tiny A}}(\{x_i\}_{i=1}^n;\dots, q^{-1}x_j,\dots, x_j, \dots, q x_j,\dots)=0.
  \end{equation}
\end{lemma}
\begin{proof}
It is not difficult to see from \eqref{eqn:ZA} that the partition function
$Z_{\text{\rm \tiny A}}(\{x_i\}_{i=1}^n;\{y_i\}_{i=1}^{2n})$
is separately symmetric in
 the $x_i$ and the $y_j$. It is therefore
sufficient to prove the statement for $y_1=q x_1,\,y_2= x_1$ and $y_3=q^{-1}x_1$. 
Let us consider the graphical representation of the partition function in \cref{fig:spinonepartfunc} and inspect the lower left corner. If we set $y_1=q x_1,\,y_2= x_1$, then the vertices in the bottom left-corner are fixed to
  \begin{equation*} 
     \begin{tikzpicture}[baseline=0.4cm]
     \draw[postaction={on each segment={mid arrow}}] (.5,0) -- (0,0);
     \draw[postaction={on each segment={mid arrow}}] (.5,0) -- (1,0);  
     \draw[densely dotted,double] (.5,0)  -- (.5,.5);
     \draw[postaction={on each segment={mid arrow}},double] (.5,-.5) -- (.5,0); 
     \draw[postaction={on each segment={mid arrow}}] (.5,.5) -- (0,.5);
     \draw[postaction={on each segment={mid arrow}}] (.5,.5) -- (1,.5);
  
     \draw[postaction={on each segment={mid arrow}},double] (.5,1.5) -- (.5,1) -- (.5,.5); 
     \draw[postaction={on each segment={mid arrow}}] (1,1) -- (0.5,1)--(0,1);
     \draw (.55,-.5) node[below] {$x_1$};
     \draw (0,0) node[left] {$y_1=q x_1$};
     \draw (0,.5) node[left] {$y_2=x_1$};
     \draw (0,1) node[left] {$y_3$};
     
     \draw (0.5,2) node {$\vdots$};
     \foreach \y in {0,.5,1} \draw (1.5,\y) node {$\cdots$};
    \end{tikzpicture} \ \ .
  \end{equation*}
  Indeed, the two lowest vertices in this row are the only ones
  with non-vanishing weights (see \cref{fig:10v})
   compatible with the boundary condition.
   The third vertex from the bottom is forced through the arrow conservation rule of the model. Since these vertices are fixed, the partition function is proportional to the product of their weights, which is $[q][q^2][q y_3/x_1]$. This vanishes for $y_3= q^{-1} x_1$, which concludes the proof.
\end{proof}

%%%%%%%%%
\subsection{Diagonal twist}
\label{sec:DTwist}
%%%%%%%%%

\subsubsection{Inhomogeneous case}

We apply the results of the previous section in order to find a linear sum rule for the vector $|\psi_{\text{\rm \tiny D}}\rangle$ for $N=2n$ sites, with the inhomogeneity parameters fixed to the alternating sum \eqref{eqn:alternatingparams}. More specifically, we consider the scalar product
\begin{equation}
  \Xi_{\text{\rm \tiny D}}(x_1,\dots,x_n;b) = \langle \chit (x_1,\dots,x_n)|\psi_{\text{\rm \tiny D}}\rangle = Z_{\text{\rm \tiny A}}(\{x_i\}_{i=1}^n;\{x_i^{-1}\}_{i=1}^n,\{x_i\}_{i=1}^n).
  \label{eqn:defXiD}
\end{equation}
It can be evaluated in closed form in terms of 
a partition function $Z_{\text{\rm \tiny U}}(\{x_i\}_{i=1}^n;\{y_i\}_{i=1}^n;b)$ for the six-vertex model on a $2n\times n$ rectangle with U-turn boundary conditions. 
This lattice is illustrated in \cref{fig:uturns} for $n=2$.
Tsuchiya \cite{tsuchiya:98} obtained the following result: 
\begin{subequations}\label{eq:ZU5.8}
\begin{align}
  Z_{\text{\rm \tiny U}}(\{x_i\}_{i=1}^n;\{y_i\}_{i=1}^n;b) = &\frac{[q^2]^n
  \prod_{i=1}^n[b/y_i][q^2 x_i^2]\prod_{i,j=1}^n[q x_i/y_j][q y_j/x_i][q x_iy_j][q/(x_iy_j)]}{\prod_{1\le i<j \le n}[x_j/x_i][y_i/y_j]\prod_{1\le i\le j \le 2n}[1/(x_ix_j)][y_iy_j]} \nonumber\\ 
  & \times \Det{i,j=1}{n} 
  (M_{\text{\rm \tiny U}})_{ij}\label{eq:Tsu}
\end{align}
where
\begin{equation}
  (M_{\text{\rm \tiny U}})_{ij}=\frac{1}{[q x_i/y_j][q y_j/x_i]}-\frac{1}{[q x_i y_j][q/(y_jx_i)]}.
\end{equation}
\end{subequations}

\begin{figure}
  \centering
  \begin{tikzpicture} 
       
    \begin{scope}
    \draw (0,0) grid[step=.5cm] (0.5,1.5);
    \foreach \x in {0,0.5}
    {
      \draw[postaction={on each segment={mid arrow}}] (\x,0)--(\x,-0.5);
      \draw[postaction={on each segment={mid arrow}}] (\x,1.5)--(\x,2);
      
      \draw[postaction={on each segment={mid arrow}}] (-.5,\x)--(0,\x);
      \draw[postaction={on each segment={mid arrow}},yshift=1cm] (-.5,\x)--(0,\x);
  
      \draw (.5,\x)--(.75,\x);
      \draw[yshift=1cm] (.5,\x)--(.75,\x);
    }
    
    \draw (.75,1) arc [start angle=-90, end angle = 90, radius=.25cm];
    \draw (1,1.25) node[right] {$q x_2$};
    \fill (1,1.25) circle (1.5pt) ;
  
    \draw (.75,0) arc [start angle=-90, end angle = 90, radius=.25cm];
    \draw (1,0.25) node[right] {$q x_1$};
    \fill (1,0.25) circle (1.5pt) ;
   
    \draw (0,-.5) node [below] {$y_1$};
    \draw (.5,-.5) node [below] {$y_2$};
    
    \draw (-1.25,0) node [right]  {$x_1$};
    \draw (-1.25,.5) node  [right] {$x_1^{-1}$};
    \draw (-1.25,1) node [right]  {$ x_2$};
    \draw (-1.25,1.5) node [right]  {$x_2^{-1}$};
    
    \begin{scope}[xshift=3.5cm,yshift=-0.5cm]
      \draw (0.75,1) arc [start angle=-90, end angle = 0, radius=.25cm];
      \draw (1,1.25) arc [start angle=0, end angle = 90, radius=.25cm];
      \fill (1,1.25) circle (1.5pt) ;
      \draw (0.5,1.0) -- (0.74,1.0) (0.5,1.5) -- (0.74,1.5);
      \draw[postaction={mid arrow}] (0.74,1.5) -- (0.75,1.5);
      \draw[postaction={mid arrow}] (0.75,1.0) -- (0.55,1.0);
      \draw (1,1.25) node [right] {$z$};
      \draw (0.8,.5) node {$[b z]$};
    \end{scope}
    
    \begin{scope}[xshift=5.0cm, yshift=-0.5cm]
      \draw (0.75,1.5) arc [start angle=90, end angle = 0, radius=.25cm];
      \draw (1,1.25) arc [start angle=0, end angle = -90, radius=.25cm];
      \fill (1,1.25) circle (1.5pt) ; 
      \draw (0.5,1.0) -- (0.74,1.0) (0.5,1.5) -- (0.74,1.5);
      \draw[postaction={mid arrow}] (0.74,1.5) -- (0.75,1.5);
      \draw[postaction={mid arrow}] (0.75,1.0) -- (0.55,1.0);
      \draw (1,1.25) node [right] {$z$};
      \draw (0.8,.5) node {$[b/z]$};
    \end{scope}   
    \end{scope}
  \end{tikzpicture}
  \caption{A $4 \times 2$ inhomogeneous lattice with U-turn boundary conditions and the statistical weights for the two admissible arrow orientations along a horizontal U-turn.}
  \label{fig:uturns}
\end{figure}
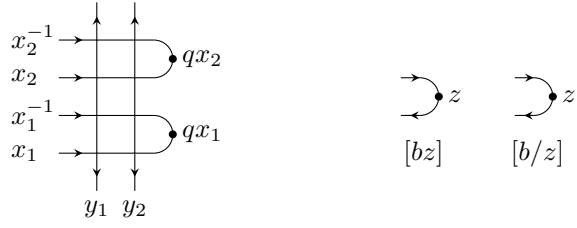

\begin{proposition} The scalar product is given by
  \begin{equation} 
    \Xi_{\text{\rm \tiny D}}(x_1,\dots,x_n;b) =\bigg([q]^n \prod_{i=1}^n[bx_i]
    \prod_{1\le i<j\le 2n}[q w_i/w_j]\bigg) Z_{\text{\rm \tiny U}}(\{x_i\}_{i=1}^n;\{x_i\}_{i=1}^n)
    \label{eqn:resXiD}
  \end{equation}
  where $w_{2i-1} = x_i$ and $w_{2i} = x_i^{-1}$.
\end{proposition}
\begin{proof}
  Notice that in \eqref{eqn:ZA}, we have $\det M = (-1)^n \det M'$ where $M'_{2k-1,j}=M_{2k-1,j}$ and
  \begin{equation}
    M'_{2k,j}=-(M_{2k-1,j}+ M_{2k,j}) =\frac{[q^2]}{[q]}\left(\frac{1}{[q x_k/y_j][q y_j/x_k]}-\frac{1}{[q x_k y_j][q/(y_jx_k)]}\right)=\frac{[q^2]}{[q]}(M_{\text{\rm \tiny U}})_{kj}.
  \end{equation}
  If $y_i$ approaches $x_i$, then the entry $M'_{2i-1,i}$ diverges whereas all other entries of $M'$ remain finite in this limit. The divergence is compensated by the vanishing of the prefactor. The rest of the proof consists in simplifying the prefactors of $\Xi_{\text{\rm \tiny D}}(x_1,\dots,x_n;b)$, which is straightforward.
\end{proof}

\subsubsection{Homogeneous limit}
\paragraph{Even size.} The results of the previous subsection allow us to
determine the special components of the vector $|\phi_{\text{\rm \tiny D}}\rangle$ given in \cref{thm:PhiDSpecCom}. We set the length of the system to $N = 2n$ and consider the homogeneous limit of \eqref{eqn:defXiD}. We show below that is related to $A_{\text{\rm \tiny V}}(2n+1; x^2)$, the generating functions for VSASMs. Kuperberg \cite{kuperberg:02} obtained this generating function as follows: 
\begin{equation}\label{eq:KupU}
  A_{\text{\rm \tiny V}}(2n+1; x^2) = \frac{Z_{\text{\tiny \rm U}}(1,\dots,1;1,\dots,1)}{[q]^{2n^2-n}[q^2]^n([bq]+[b/q])^n}, \quad x=q+q^{-1}.
\end{equation}

Let us define
\begin{equation}
  |\chit\rangle \equiv |\chit(z=1)\rangle=[b][bq]|{\Uparrow\Downarrow}\rangle-[b q][b/q]|{00}\rangle+[b][b/q]|{\Downarrow\Uparrow}\rangle.
  \label{eqn:defChiHom}
\end{equation}  
Using \eqref{eqn:defPhiD}, \eqref{eqn:resXiD} and \eqref{eq:KupU}, we obtain 
\begin{equation}
   (\langle \chit |)^{\otimes n}|\phi_{\text{\rm \tiny D}}\rangle = [b]^n([bq]+[b/q])^nA_{\text{\rm \tiny V}}(2n+1; x^2).
   \label{eqn:chiPhiD}
\end{equation}

In particular, setting $b=q$ in \eqref{eqn:chiPhiD} leads to $(\phi_{\text{\rm \tiny D}})_{\Uparrow\Downarrow\cdots \Uparrow\Downarrow} = A_{\text{\rm \tiny V}}(2n+1;x^2)$, as stated in \cref{thm:PhiDSpecCom}.
Furthermore, it is possible to explicitly compute \eqref{eq:KupU} by taking the homogeneous limit of \eqref{eq:ZU5.8}. The result is an expression for $A_{\text{\rm \tiny V}}(2n+1; t)$ in terms of the determinant of a matrix whose entries are polynomial in $t$.
\begin{proposition}
  \label{prop:AVDet}
  The polynomial $A_{\text{\rm \tiny V}}(2n+1; t)$ is given by
  \begin{equation}\label{eq:AVdet}
    A_{\text{\rm \tiny V}}(2n+1; t) = \Det{i,j=0}{n-1}\left(\sum_{k=0}^{n-1}\binom{i+j+1}{i+k+1}\binom{i+k+1}{2k+1}t^{k}\right).
    \end{equation}
  \begin{proof}
     We abbreviate $t=x^2=(q+q^{-1})^2$. Using the results of \cite{kuperberg:02}, and Tsuchiya's partition function 
     \eqref{eq:Tsu}, we obtain after a simple calculation
     \begin{equation}
     A_{\text{\rm \tiny V}}(2n+1; t) = [q]^{2n(n+1)}\lim_{\substack{x_1,\dots,x_n\to 1\\ y_1,\dots,y_n\to 1}}\frac{\Det{i,j=1}{n}\left(\frac{1}{[q x_i/y_j][q y_j/x_i][q x_iy_j][q/(x_iy_j)]}\right)}{\prod_{1\le i<j \le n}[x_i/x_j][x_i x_j][y_i/y_j][y_i y_j]}.
     \end{equation}
     In order to apply the technique from \cite{behrend:12}, we introduce new variables $X_i=x_{i+1}^2+1/x_{i+1}^2,\,Y_i=y_{i+1}^2+1/y_{i+1}^2$ for $\,i=0,\dots,n-1$. This leads to
     \begin{alignat}{2}
       A_{\text{\rm \tiny V}}(2n+1; t) &= [q]^{2n(n+1)}\lim_{\substack{X_0,\dots,X_{n-1}\to 2\\ Y_0,\dots,Y_{n-1}\to 2}}\frac{\Det{i,j=0}{n-1}f(X_i,Y_j)}{\prod_{0\le i<j\le n-1}(X_i-X_j)(Y_i-Y_j)}\nonumber\\
       & = [q]^{2n(n+1)} \lim_{\substack{X_0,\dots,X_{n-1}\to 2\\ Y_0,\dots,Y_{n-1}\to 2}} 
\ \Det{i,j=0}{n-1}f[X_0, \dots, X_i;Y_0,\dots,Y_j]\nonumber\\
       &=[q]^{2n(n+1)}\Det{i,j=0}{n-1}\left(\{X^iY^j\}f(X+2,Y+2)\right)
     \end{alignat}
      where we applied \eqref{eqn:divideddiff} with $r=2$ to the function
  \begin{equation}
    f(X,Y)=\frac{1}{q^4+q^{-4}-2-(q^2+q^{-2})XY + X^2+Y^2}.
    \label{eqn:deff}
  \end{equation}
  It remains to evaluate the series expansion of $ f(X+2,Y+2)$ in $X$ and $Y$, which yields
  \begin{equation}
   \{X^iY^j\}f(X+2,Y+2)=\frac{1}{(x^2-4)^{i+j+2}} \sum_{k=0}^{n-1}\binom{i+j+1}{i+k+1}\binom{i+k+1}{2k+1}t^{k}.
   \label{eqn:coeffsf}
  \end{equation}
  The prefactor $(x^2-4)^{-i-j-2}= [q]^{-2(i+j+2)}$ can be extracted from the determinant and \eqref{eq:AVdet} readily follows.
  \end{proof}
\end{proposition}
\paragraph{Odd size.} Next, we consider $N=2n+1$ sites and evaluate the component $(\phi_{\text{\rm \tiny D}})_{0\cdots 0}$. This can be done graphically, taking into account the fact that, as shown in \cite{hagendorf:15}, each component of the vector $|\psi_{\text{\rm \tiny D}}\rangle$ is a partition function of the ten-vertex model defined in \cref{sec:Rmat+fusion}. The key ingredient in proving the next proposition is that some vertex weights vanish in the homogeneous limit, as one can see from \cref{fig:10v}.
\begin{proposition}
  In the homogeneous limit, 
  the component of $|\psi_{\text{\rm \tiny D}}\rangle$ along $|0\cdots0 \rangle$ is given by
  \begin{equation}
    \psi_{\text{\rm \tiny D}}(1,\dots,1)\!\!\!_{\underset{N=2n+1}{\underbrace{\scriptstyle0\cdots 0}}}= ([q][q^2])^{n+1/2}[q]^{2n}\psi_{\text{\rm \tiny D}}(1,\dots,1)_{\underset{N'=2n}{\underbrace{\scriptstyle\Uparrow\Downarrow\cdots \Uparrow\Downarrow}}}\,.
  \end{equation}
  \begin{proof}
    Graphically, the component $\psi_{\text{\rm \tiny D}}(1,\dots,1)_{0\cdots 0}$ can be represented as a partition function on an $N\times N$ square:
    \begin{equation}
      \psi_{\text{\rm \tiny D}}(1,\dots,1)_{0\cdots 0} =\quad 
      \begin{tikzpicture}[baseline=-.3cm]
      \foreach \y in {-1.5,-1,.5,1}
      {
        \draw (-1.5,\y) -- (1,\y);
        \draw[postaction={on each segment={mid arrow}}] (-1.5,\y) -- (-2,\y);
        \draw[postaction={on each segment={mid arrow}}] (1,\y) -- (1.5,\y);
        \draw[double] (\y,-1.5) -- (\y,1.);
        \draw[double,densely dotted] (\y,1.) -- (\y,1.5);
        \draw[double,postaction={on each segment={mid arrow}}] (\y,-2) -- (\y,-1.5);  
      }
      \draw (-.25,1.25) node {$\cdots$};
      \draw (-.25,-1.75) node {$\cdots$};
      \draw (1.25,-.125) node {$\vdots$};
      \draw (-1.75,-.125) node {$\vdots$};
      \end{tikzpicture} \ \  .
    \end{equation}
    Here, all horizontal and vertical parameters are set to one. Since the weights of the vertices
    \begin{equation*}
 \begin{tikzpicture}     
     \draw[postaction={on each segment={mid arrow}}] (0,0) -- (.5,0) -- (1,0);
     \draw[densely dotted ,double]  (.5,.5)  --(.5,-.5);
     \begin{scope}[xshift=1.5cm]
     \draw (0.5,0.05) node {and}; 
     \end{scope}
     \begin{scope}[xshift=3cm]
       \draw[postaction={on each segment={mid arrow}}] (1,0) -- (.5,0) -- (0,0);
       \draw[densely dotted,double] (.5,-.5)  -- (.5,0) --  (.5,.5);|
     \end{scope}
 \end{tikzpicture}
\end{equation*}
vanish in the homogeneous limit, the topmost row is fixed to
\begin{equation*} 
 \begin{tikzpicture}     
   \begin{scope}
   \clip (0,-.5) rectangle (1.5,.5);
   \draw[postaction={on each segment={mid arrow}}] (.5,0) -- (0,0);
   \draw[postaction={on each segment={mid arrow}}] (.5,0) -- (1,0);
   \draw[postaction={on each segment={mid arrow}}] (1.5,0) -- (1,0);
   \draw[double,postaction={on each segment={mid arrow}}] (.5,-.5) -- (.5,0);
   \draw[double,postaction={on each segment={mid arrow}}] (1,0) -- (1,-.5);
   \draw[double,postaction={on each segment={mid arrow}}] (1.5,-.5) -- (1.5,0);
   \foreach \x in {.5,1,1.5} \draw [densely dotted,double] (\x,0) -- (\x,.5);
   \end{scope}
   \begin{scope}
     \clip (2.5,-.5) rectangle (4,.5);
     \foreach \x in {2.5,3,3.5} \draw [densely dotted,double] (\x,0) -- (\x,.5);
     \draw[postaction={on each segment={mid arrow}}] (2.5,0) -- (3,0);
     \draw[postaction={on each segment={mid arrow}}] (3.5,0) -- (3,0);
     \draw[postaction={on each segment={mid arrow}}] (3.5,0) -- (4,0);
     \draw[double,postaction={on each segment={mid arrow}}] (2.5,-.5) -- (2.5,0);
     \draw[double,postaction={on each segment={mid arrow}}] (3,0) -- (3,-.5);
     \draw[double,postaction={on each segment={mid arrow}}] (3.5,-.5) -- (3.5,0);
   \end{scope}
   \draw (2,0) node {$\cdots$};
 \end{tikzpicture}
\end{equation*}
and yields a weight $([q][q^2])^{n+1/2}$. This row can thus be removed from the picture.
  The component of interest is thus reduced to a partition function on an
  $(N-1)\times N$ rectangle 
 whose top boundary condition consists of alternating outgoing and ingoing arrows:
   \begin{equation}
      \psi_{\text{\rm \tiny D}}(1,\dots,1)_{0\cdots 0} = \quad ([q][q^2])^{n+1/2}\times 
      \begin{tikzpicture}[baseline=-.5cm]
      \foreach \y in {-1.5,-1,.5}
      {
      \draw[postaction={on each segment={mid arrow}}] (-1.5,\y) -- (-2,\y);
      \draw[postaction={on each segment={mid arrow}}] (1,\y) -- (1.5,\y);
       \draw (-1.5,\y) -- (1,\y);
      }
      \foreach \y in {-1.5,-1,0,.5,1}
      {
      \draw[double] (\y,-1.5) -- (\y,0.5);
      \draw[double,postaction={on each segment={mid arrow}}] (\y,-2) -- (\y,-1.5);  
      }
      \draw[double,postaction={on each segment={mid arrow}}] (-1.5,.5) -- (-1.5,1);
      \draw[double,postaction={on each segment={mid arrow}}] (-1,1) -- (-1,.5);
      \draw[double,postaction={on each segment={mid arrow}}] (1,.5) -- (1,1);
      \draw[double,postaction={on each segment={mid arrow}}] (.5,1) -- (.5,.5);
            \draw[double,postaction={on each segment={mid arrow}}] (0,.5) -- (0,1);
      \draw (-.5,0.75) node {$\cdots$};
      \draw (-.5,-1.75) node {$\cdots$};
      \draw (1.25,-.125) node {$\vdots$};
      \draw (-1.75,-.125) node {$\vdots$};
      \end{tikzpicture} \ \ .
    \end{equation}
    It is also possible to remove the rightmost column.
    To see this, we inspect the vertex in the top right corner. It has two outgoing arrows, and hence the two other edges 
     need to be ingoing arrows. The situation is repeated for each of the next vertices below
     and fixes each one to the configuration
    \begin{equation*}
    \begin{tikzpicture}
      \draw[postaction={on each segment={mid arrow}}] (0,0) -- (.5,0) -- (1,0);
     \draw[postaction={on each segment={mid arrow}},double] (.5,-.5) -- (.5,0) --(.5,.5);
    \draw (1.2,-0.1) node {.};
    \end{tikzpicture}
    \end{equation*}
    This vertex has weight $[q]$ in the homogeneous limit. The column hence produces a weight $[q]^{2n}$ and may be erased, yielding
    \begin{equation}
      \psi_{\text{\rm \tiny D}}(1,\dots,1)_{0\cdots 0} = \quad ([q][q^2])^{n+1/2}[q]^{2n}\times 
      \begin{tikzpicture}[baseline=-.5cm]
      \foreach \y in {-1.5,-1,.5}
      {
      \draw[postaction={on each segment={mid arrow}}] (-1.5,\y) -- (-2,\y);
      \draw[postaction={on each segment={mid arrow}}] (.5,\y) -- (1,\y);
       \draw (-1.5,\y) -- (.5,\y);
      }
      \foreach \y in {-1.5,-1,0,.5}
      {
      \draw[double] (\y,-1.5) -- (\y,0.5);
      \draw[double,postaction={on each segment={mid arrow}}] (\y,-2) -- (\y,-1.5);  
      }
      \draw[double,postaction={on each segment={mid arrow}}] (-1.5,.5) -- (-1.5,1);
      \draw[double,postaction={on each segment={mid arrow}}] (-1,1) -- (-1,.5);
      \draw[double,postaction={on each segment={mid arrow}}] (.5,1) -- (.5,.5);
      \draw[double,postaction={on each segment={mid arrow}}] (0,.5) -- (0,1);
      \draw (-.5,0.75) node {$\cdots$};
      \draw (-.5,-1.75) node {$\cdots$};
      \draw (0.75,-.125) node {$\vdots$};
      \draw (-1.75,-.125) node {$\vdots$};
      \draw (1.35,-0.6) node {.};
      \end{tikzpicture}
      \label{eq:psiDdiag}
    \end{equation}
     The diagram in \eqref{eq:psiDdiag} is the partition function for the ten-vertex model on an $(N-1)\times (N-1)$ homogeneous lattice, which we identify as the component $\psi_{\text{\rm \tiny D}}(1,\dots,1)_{\Uparrow\Downarrow\cdots \Uparrow\Downarrow}$ for the system size $N'=N-1=2n$. This completes the proof.
  \end{proof}
\end{proposition}
Combining this proposition with the definition \eqref{eqn:defPhiD} of $|\phi_{\text{\rm \tiny D}}\rangle$, we obtain $ (\phi_{\text{\rm \tiny D}})_{0\cdots 0}= x^n (\phi_{\text{\rm \tiny D}})_{\Uparrow\Downarrow\cdots \Uparrow\Downarrow}=x^nA_{\text{\rm \tiny V}}(2n+1;x^2)$ 
which is precisely \eqref{eqn:PhiDSpecCompOdd}. This
ends the proof of \cref{thm:PhiDSpecCom}.

%%%%%%%%%
\subsection{Anti-diagonal twist}
\label{sec:ADTwist}
%%%%%%%%%

\subsubsection{Inhomogeneous case: Even size}
\label{sec:evensize}

We turn now to the evaluation of certain special components of the vector $|\psi_{\text{\rm \tiny AD}}\rangle$ for $N=2n$, with the inhomogeneity parameters following the alternating pattern \eqref{eqn:alternatingparams}. To this end, we use the same strategy as for the diagonal twist:
We compute the scalar product 
\begin{equation}
  \Xi_{\text{\rm \tiny AD}}(x_1,\dots,x_n;b) = \langle \chit (x_1,\dots,x_n)|\psi_{\text{\rm \tiny AD}}\rangle
  \label{eqn:defXiAD}
\end{equation}
and evaluate certain components through the specification of the parameter $b$. By now, it is no longer surprising that this scalar product is related to 
a partition function of the six-vertex model found by Kuperberg \cite{kuperberg:02}: $Z_{\text{\rm \tiny UU}}(\{x_i\}_{i=1}^n,\{y_i\}_{i=1}^n;b,c)$. The domain is a
$2n\times2n$ inhomogeneous lattice with UU-turn boundary conditions, as illustrated in \cref{fig:doubleuturns}
for $n=2$.
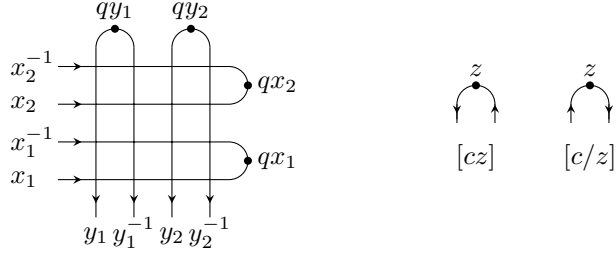
\begin{figure}[h]  
  \centering
  \begin{tikzpicture}
    \begin{scope}
    \draw (0,0) grid[step=.5cm] (1.5,1.5);
    \foreach \x in {0,0.5,1,1.5}
    {
      \draw[postaction={on each segment={mid arrow}}] (-.5,\x)--(0,\x);
      \draw (1.75,\x)--(1.5,\x);
      \draw[postaction={on each segment={mid arrow}}] (\x,0)--(\x,-0.5);
      \draw (\x,1.5)--(\x,1.75);
      
    }
    
    \draw (1.75,1) arc [start angle=-90, end angle = 90, radius=.25cm];
    \fill (2,1.25) circle (1.5pt) ;
    \draw (2,1.25) node[right] {$q x_2$};

    \draw (1.75,0) arc [start angle=-90, end angle = 90, radius=.25cm];
    \fill (2,0.25) circle (1.5pt) ;
    \draw (2,0.25) node[right] {$q x_1$};
        
    \draw (0.5,1.75) arc [start angle=0, end angle = 180, radius=.25cm];
    \fill (0.25,2) circle (1.5pt) ;
    \draw (0.25,2) node[above] {$q y_1$};

    \draw (1.5,1.75) arc [start angle=0, end angle = 180, radius=.25cm];
    \fill (1.25,2) circle (1.5pt) ;
    \draw (1.25,2) node[above] {$q y_2$};
       
     \draw (-1.25,0) node [right]  {$x_1$};
    \draw (-1.25,.5) node  [right] {$x_1^{-1}$};
    \draw (-1.25,1) node [right]  {$ x_2$};
    \draw (-1.25,1.5) node [right]  {$x_2^{-1}$};
    
    \draw (0,-.75) node  {$y_1$};
    \draw (.5,-.7) node {$y_1^{-1}$};
    \draw (1,-.75) node {$y_2$};
    \draw (1.5,-.7) node{$y_2^{-1}$};
    
    \end{scope}
    
    \begin{scope}[xshift=5cm,yshift=1cm]
	  \draw (0.25,0) arc [start angle=0, end angle = 90, radius=.25cm];
	  \draw (0,.25) arc [start angle=90, end angle = 180, radius=.25cm];
	  \draw[postaction={on each segment={mid arrow}}] (-0.25,0) -- (-0.25,-0.20);
	  \draw[postaction={on each segment={mid arrow}}] (0.25,-0.01) -- (0.25,0);
	  \draw (0.25,-0.25) -- (0.25,0) (-0.25,-0.25) -- (-0.25,0);
          \fill (0,0.25) circle (1.5pt) ;
          \draw (0,.25) node [above] {$z$};
          \draw (0,-1.00) node [above] {$[c z]$};
    \end{scope}
    
    \begin{scope}[xshift=6.5cm,yshift=1cm]
          \draw (-0.25,0) arc [start angle=180, end angle = 90, radius=.25cm];
          \draw (0,.25) arc [start angle=90, end angle = 0, radius=.25cm];
	  \draw[postaction={on each segment={mid arrow}}] (0.25,0) -- (0.25,-0.20);
	  \draw[postaction={on each segment={mid arrow}}] (-0.25,-0.01) -- (-0.25,0);
	  \draw (0.25,-0.25) -- (0.25,0) (-0.25,-0.25) -- (-0.25,0);
          \fill (0,0.25) circle (1.5pt) ;
          \draw (0,.25) node [above] {$z$};
          \draw (0,-1.00) node [above] {$[c/z]$};
    \end{scope}
    
  \end{tikzpicture}
  \caption{A $4 \times 4$ inhomogeneous lattice
  with UU-turn boundary conditions and the statistical weights for the two admissible arrow orientations along a vertical U-turn.}
  \label{fig:doubleuturns}
\end{figure}

Here, the weights $\mathfrak a(z), \mathfrak b(z)$ and $\mathfrak c(z)$ of \cref{fig:6v} are assigned to the vertex arrow configurations. Kuperberg \cite{kuperberg:02} showed that the corresponding partition function factorises as\begin{equation}
  Z_{\text{\rm \tiny UU}}(\{x_i\}_{i=1}^n;\{y_i\}_{i=1}^n;b,c)=\left(\prod_{i=1}^n\frac{[q^2/y_i^2]}{[b/y_i]}\right)Z_{\text{\rm \tiny U}}(\{x_i\}_{i=1}^n;\{y_i\}_{i=1}^n;b
  )Z^{(2)}_{\text{\rm \tiny UU}}(\{x_i\}_{i=1}^n;\{y_i\}_{i=1}^n;b,c)
\end{equation}
where 
\begin{align}
  \label{eqn:defZ2UU}
  Z^{(2)}_{\text{\rm \tiny UU}}(\{x_i\}_{i=1}^n;\{y_i\}_{i=1}^n;b,c) = &\frac{\prod_{i,j=1}^n[q x_i/y_j][q y_j/x_i][q x_iy_j][q/(x_iy_j)]}{\prod_{1\le i<j\le n}[x_j/x_i][y_i/y_j]\prod_{1\le i\le j \le n}[1/(x_ix_j)][y_iy_j]}\\
  & \times \Det{i,j=1}{n} 
  M_{\text{\rm \tiny UU}}(\{x_i\}_{i=1}^n;\{y_i\}_{i=1}^n;b,c)_{ij},\nonumber
\end{align}
with the matrix elements
\begin{equation}
  \label{eqn:defMUU}
  M_{\text{\rm \tiny UU}}(\{x_i\}_{i=1}^n;\{y_i\}_{i=1}^n;b,c)_{ij} = \frac{[b/y_j][c x_i]}{[q x_i/y_j]}-\frac{[b/y_j][c/x_i]}{[q /(x_iy_j)]}+\frac{[b y_j][c /x_i]}{[q y_j/x_i]}-\frac{[b y_j]
  [cx_i]
  }{[q x_iy_j]}.
\end{equation}

\begin{proposition} We have the following expression:
\label{prop:XiADZ2}
\begin{align}\label{eq:Xievaluation}
  \Xi_{\text{\rm \tiny AD}}(x_1,\dots,x_n;b) = \frac{(-1)^n \prod_{i=1}^n[q^2x_i^2] Z^{(2)}_{\text{\rm \tiny UU}}(\{x_i\}_{i=1}^n;\{x_i\}_{i=1}^n;b,b^{-1})}{\prod_{i,j=1}^n[q x_i/x_j]\prod_{1\le i<j \le n}[q x_i x_j]\prod_{1\le i\le j\le n}[q/(x_ix_j)]}
\end{align}
\end{proposition}

\begin{proof} The proof has two parts. The first part consists of writing out the scalar product as
\begin{equation}\label{eq:xixi}
  \Xi_{\text{\rm \tiny AD}}(x_1,\dots,x_n;b) = \sum_{\bm h} \psi_{\text{\rm \tiny AD}}(\bm h)\langle \chit(x_1,\dots,x_n)||\bm h\rrangle,
\end{equation}
and discussing which choices for $\bm h$ yield non-vanishing contributions to the sum. The coefficients $\psi_{\text{\rm \tiny AD}}(\bm h)$ are defined in \eqref{eq:ShandPsih}.
They are non-zero only if each height $h_j$ is different from~$1$. Under these conditions, $\psi_{\text{\rm \tiny AD}}(\bm h)$ can be simplified to the expression given in \eqref{eqn:simplePsiAD}.
Moreover, because the vector $|\chit(x_1,\dots, x_n)\rangle$ has zero magnetisation, $\langle \chit(x_1,\dots,x_n)||\bm h\rrangle$ vanishes unless
\begin{equation}
  \sum_{j=1}^{2n} h_j=2n.
  \label{eqn:sumh}
\end{equation}  
If this condition is met, then $\langle \chit(x_1,\dots,x_n)||\bm h\rrangle$
is proportional to $Z_{\text{\rm \tiny A}}(\{x_i\}_{i=1}^n;\{y_i\}_{i=1}^{2n})$ with the
$y_i$
fixed by the height configuration $\bm h$. Notice that if $h_{2j-1}=h_{2j}=2$, then the set $\{y_i\}_{i=1}^{2n}$ contains
the four variables $x_j,qx_j,x_j^{-1},q x_j^{-1}$ 
and the partition function vanishes according to \cref{lemma:wheel}. Combining this with \eqref{eqn:sumh}, we conclude 
that $\langle \chit(x_1,\dots,x_n)||\bm h\rrangle$ is zero unless $\bm h$ meets the conditions
\begin{equation}
  h_{2j-1}+h_{2j}=2, \qquad j=1,\dots,n.
\end{equation} 
For the remaining terms, it is convenient to write $\bm h = \bm h(\bm \epsilon)$ with
\begin{equation}
  h_{2j-1}(\bm \epsilon) = 1+\epsilon_j, \qquad h_{2j}(\bm \epsilon)=1-\epsilon_j,\qquad j=1,\dots,n
\end{equation}
where $\epsilon_1,\dots, \epsilon_n=\pm 1$.
 The expression \eqref{eq:xixi} becomes
\begin{equation}
\Xi_{\text{\rm \tiny AD}}(x_1,\dots,x_n;b) = \sum_{\epsilon_1,\dots,\epsilon_n = \pm1
} \frac{\psi_{\text{\rm \tiny AD}}(\bm h(\bm \epsilon)) Z_{\text{\rm \tiny A}}(\{x_i\}_{i=1}^n;\{x_i^{\epsilon_i}\}_{i=1}^n,\{qx_i^{\epsilon_i}\}_{i=1}^n)}{\prod_{i=1}^n a(x_i^{\epsilon_i})a(qx_i^{\epsilon_i})}.
\label{eqn:XiADEpsilon}
\end{equation}

The second part of the proof is the evaluation of this sum. 
We note that $Z_{\text{\rm \tiny A}}(\{x_i\}_{i=1}^n;\{y_i\}_{i=1}^{2n})$ satisfies a family of {\it reduction relations}. Indeed, for the specialisations $y_j = x_i^{\pm 1}$ and $y_j = qx_i^{\pm 1}$, the determinant in \eqref{eqn:ZA} is singular. In each case, the singularity is compensated by a vanishing prefactor. This allows us to rewrite the specialised $Z_{\text{\rm \tiny A}}$ in terms of the determinant of a matrix of size $2n-1$. In \eqref{eqn:XiADEpsilon}, each of the $y_i$ is set to one of above specialisations. By the repeated application of the reduction, we find after some manipulations that the partition functions in \eqref{eqn:XiADEpsilon} are expressible as products of monomials:
\begin{align}\label{eq:ZAmonom}   Z_{\text{\rm \tiny A}}(\{x_i\}_{i=1}^n;\{x_i^{\epsilon_i}\}_{i=1}^n,\{qx_i^{\epsilon_i}\}_{i=1}^n)&= ([q][q^2])^n    \prod_{1\le i<j\le 2n}\!\!
   [q w_i/w_j]\prod_{i=1}^n[b x_i^{-\epsilon_i}/q][b x_i^{-\epsilon_i}][q^2 x_i^2]  \\ 
   &\times \!\!\prod_{1\le i<j\le n}
   [q^{2\epsilon_i} x_i/x_j][q^{2\epsilon_i} x_ix_j][q^{-\epsilon_i-\epsilon_j} x_i/x_j][q^{-\epsilon_i+\epsilon_j} x_ix_j]. \nonumber
   \end{align}
Furthermore, using \eqref{eqn:simplePsiAD} we have
\begin{align}\label{eq:psihmonom}
  \psi_{\text{\rm \tiny AD}}(\bm h(\bm \epsilon)) = \frac{(-1)^n \prod_{1\le i<j\le n}[q^{\epsilon_i-\epsilon_j}x_i/x_j]^2
  [q^{\epsilon_i+\epsilon_j}x_ix_j]^2\prod_{i=1}^n[q^{2\epsilon_i}x_i^2]}{\prod_{1\le i<j\le 2n} [w_i/w_j]}.
\end{align} 
Each term in \eqref{eqn:XiADEpsilon} is thus a quotient of products of monomials. 
Using the fact that the $\epsilon_i$ are signs, each one can be
recast into a determinant form with the help of a classical identity of Cauchy's \cite{cauchy:41}, rewritten in our notation, 
\begin{equation}
  \Det{i,j=1}{n}\left(\frac{1}{[a_i/b_j]}-\frac{1}{[a_i b_j]}\right) = \left(\prod_{i=1}^n\frac{[a_i^2][b_i]}{[a_i]}\right)\frac{\prod_{1\le i<j \le n}[a_i a_j][b_i b_j][a_i/a_j][b_j/b_i]
  }{\prod_{i,j=1}^n[a_i/b_j][a_i b_j]},
\end{equation}
which we specialise at $a_i = q^{\epsilon_i}x_i$ and $b_j = x_j$.
After some simplifications, we find
\begin{align}
  & \Xi_{\text{\rm \tiny AD}}(x_1,\dots,x_n;b) = [q]^n 
  \prod_{1\le i<j\le 2n}\frac{[q w_i/w_j]}{[w_i/w_j]}\prod_{i=1}^n\frac{[q^2 x_i^2]}{[x_i]}\nonumber\\
 & \times \sum_{ 
 \epsilon_1,\dots,\epsilon_n = \pm 1
 }\Det{i,j=1}{n} 
 \Biggl[(-1)^{\frac{1+\epsilon_i}{2}}\frac{[q^{\epsilon_i}x_i][bx_i^{-\epsilon_i}][bx_i^{-\epsilon_i}/q]}{[q^{2\epsilon_i}x_i^2
 ]}  \times \Biggl(\frac{1}{[q^{\epsilon_i}x_i/x_j]}-\frac{1}{[q^{\epsilon_i}x_ix_j]}\Biggr)
  \Biggr].
\end{align}
Performing the sum over $\epsilon_i = \pm 1$ is equivalent to transforming the $i$-th row of the matrix in the determinant.
One verifies that
\begin{equation}\label{eq:magicalidentity}
  \sum_{\epsilon_i=\pm 1} (-1)^{\frac{1+\epsilon_i}{2}} \frac{[q^{\epsilon_i}x_i][bx_i^{-\epsilon_i}][bx_i^{-\epsilon_i}/q]}{[q^{2\epsilon_i}x_i^2
  ]}
  \left(\frac{1}{[q^{\epsilon_i}x_i/x_j]}-\frac{1}{[q^{\epsilon_i}x_ix_j]}\right) = \frac{[x_j]}{[x_j^2]}
   M_{\text{\rm \tiny UU}}(\{x_i\}_{i=1}^n; 
   \{x_i\}_{i=1}^n; b,b^{-1})_{ij}.
\end{equation}
This leads to
\begin{equation}\label{eq:Xievaluation2}
  \Xi_{\text{\rm \tiny AD}}(x_1,\dots,x_n;b) = [q]^n\!\! \prod_{1\le i<j\le 2n} \frac{[q w_i/w_j]}{[w_i/w_j]}\prod_{i=1}^n\frac{[qw_i^2]}{[w_i^2]} \det M_{\text{\rm \tiny UU}}(\{x_i\}_{i=1}^n; \{x_i\}_{i=1}^n; b,b^{-1})_{ij}
\end{equation}
and to \eqref{eq:Xievaluation} after using \eqref{eqn:alternatingparams} and \eqref{eqn:defZ2UU}.
\end{proof}

\subsubsection{Inhomogeneous case: Odd size} 
In this section, we set the system size to an odd integer, $N = 2n+1$, and the inhomogeneity parameters to the values
\begin{equation}
  w_{2i-1}=x_i, \quad w_{2i}=x_i^{-1}, \quad i=1,\dots,n, \qquad w_{2n+1} = x_{n+1} = 1.
\end{equation}
Let us define the scalar product
\begin{equation}
Z_b = \Big(
\langle \chit(x_1,\dots,x_n)|\otimes \langle{\Downarrow}| \Big) |\psi_{\text{\rm \tiny AD}} \rangle.
\end{equation}
Unlike in \cref{sec:evensize}, we have not found a determinant formula for this scalar product for arbitrary $b$, but only for the special case $b=q^{-1}$.

We proceed by first relating $Z_b$ to a partition function of the ten-vertex model on a domain of size $N' = N+1 = 2n+2$ that we already investigated in \cref{sec:evensize}. For this reason, we add extra labels indicating the system size and thus denote $a(z)$ by $a_N(z)$ and $\psi_{\textrm{\tiny AD}}(\bm h)$ by $\psi_N(\bm h)$. We have the following proposition.
\begin{proposition} We have $Z_b= \lim_{b'\rightarrow q^{-1}} Z_{b,b'}$
with
\begin{subequations}
 \begin{equation}
 \label{eq:Zbeq0}
 Z_{b,b'} = \frac{1}{[q][q^2]} \sum_{\epsilon_1, \dots, \epsilon_n = \pm 1} \frac{\psi_{2n+1}(\bm h(\bm \epsilon))D_{b,b'}(\bm \epsilon)}{\prod_{i=1}^{n+1}[q^2 x_i^{\epsilon_i}][q x_i^{\epsilon_i}]a_{2n+1}(q x_i^{\epsilon_i})a_{2n+1}(x_i^{\epsilon_i})},
 \end{equation}
where $h_{2j-1}(\bm \epsilon) = 1+\epsilon_j, h_{2j}(\bm \epsilon) = 1-\epsilon_j$ for $j=1,\dots,n$, $h_{2n+1}(\bm \epsilon)=2,\,\epsilon_{n+1} = 1$ and
\begin{equation}\label{eq:Dbb}
D_{b,b'}(\bm \epsilon) = \ 
  \begin{tikzpicture}[baseline = 1.5cm]
     
     \foreach \x in {0,.5}
     {  
       \draw[double] (\x,0) -- (\x,3.25);
       \draw[double,xshift=2cm] (\x,0) -- (\x,3.25);
       \draw[double,xshift=3cm] (\x,0) -- (\x,3.25);
       
       \draw[double,postaction={on each segment={mid arrow}}] (\x,-.5) -- (\x,0);
       \draw[double,xshift=2cm,postaction={on each segment={mid arrow}}] (\x,-.5) -- (\x,0);
       \draw[double,xshift=3cm,postaction={on each segment={mid arrow}}] (\x,-.5) -- (\x,0);
     }
     
     \foreach \x in {0,2}
     {  
        \draw [double,xshift=\x cm] (.5,3.25) arc [start angle=0, end angle = 180, radius=.25cm];
        \fill[xshift=\x cm] (0.25,3.5) circle (1.5pt) ;
     }
     	\draw[double] (3,3.25) -- (3,3.5) -- (3.5,3.5) -- (3.5,3.25);
        \fill[xshift=3 cm] (0.25,3.5) circle (1.5pt) ;

     \draw (0.25,3.5) node[above] {$x_1$} ;
     \draw (2.25,3.5) node[above] {$x_n$} ;
     \draw (3.25,3.5) node[above] {$1$} ;
     
     \foreach \x in {0,4}
       \draw[xshift=\x cm] (-.25,1.1) node {$\vdots$};
     
     \foreach \y in {0,3.5}
       \draw[yshift=\y cm] (1.25,-.25) node {$\cdots$};

     \foreach \y in {0,0.5,1.5,2,2.5,3}
     {
     \draw[postaction={on each segment={mid arrow}},yshift=\y cm] (0,0)--(-0.5,0);
     \draw[postaction={on each segment={mid arrow}},yshift=\y cm] (3.5,0)--(4,0);
     \draw[yshift=\y cm]  (0,0)--(3.5,0);
     }
     
     \draw (-.5,0) node [left] {$x_1^{\epsilon_1}$};
     \draw (-.5,0.5) node [left] {$qx_1^{\epsilon_1}$};
     \draw (-.5,1.5) node [left] {$x_n^{\epsilon_n}$};
     \draw (-.5,2) node [left] {$qx_n^{\epsilon_n}$};
     \draw (-.5,2.5) node [left] {$x_{n+1}^{\epsilon_{n+1}} = 1$};
     \draw (-.5,3) node [left] {$qx_{n+1}^{\epsilon_{n+1}} = q$};
     \draw (0,-.79) node {$x_1$};
     \draw (.6,-.75) node  {$x_1^{-1}$};
      \draw (2,-.79) node {$x_n$};
     \draw (2.6,-.75) node  {$x_n^{-1}$};
     \draw (3.0,-.79) node {$1$};
     \draw (3.5,-.79) node  {$1$};
     
     \draw (4.35,1.5) node  {$,$};
     
     \begin{scope}[xshift=5.4cm,yshift=1.5cm]
	 	\draw[double] (-0.25,0) -- (-0.25,0.25) -- (0.25,0.25) -- (0.25,0);
          	\fill (0,0.25) circle (1.5pt) ;
          	\draw (0,.25) node [above] {$x$};
          	\draw[double] (0.25,-0.25) -- (0.25,0) (-0.25,-0.25) -- (-0.25,0);
		\draw (0.45,0) node[right] {$=$};
    \end{scope}
     
     \begin{scope}[xshift=6.85cm,yshift=1.5cm]
          \draw[double] (0.25,0) arc [start angle=0, end angle = 90, radius=.25cm];
          \draw[double] (0,.25) arc [start angle=90, end angle = 180, radius=.25cm];
          \fill (0,0.25) circle (1.5pt) ;
          \draw (0,.25) node [above] {$x$};
          \draw[double] (0.25,-0.25) -- (0.25,0) (-0.25,-0.25) -- (-0.25,0);     
          \draw (0.35,0) node[right] {$\Big|_{b\rightarrow b'}.$};
    \end{scope}

    \end{tikzpicture}
\end{equation}
\end{subequations}
\end{proposition}
\begin{proof}
We perform the computation of $Z_b$ over $(\mathbb C^3)^{\otimes 2(n+1)}$ and find
\begin{alignat}{2}
Z_b &= \Big(\langle \chit(x_1,\dots,x_n)|\otimes \langle{\Downarrow\Uparrow}|
\Big) \Big(|\psi_{\text{\rm \tiny AD}} \rangle\otimes |{\Uparrow}\rangle\Big)  = \nonumber \\
& = \sum_{\bm h} \frac{\psi_{2n+1}(\bm h)}{\prod_{j=1}^{2n+1}\prod_{h=0}^{h_j-1} a_{2n+1}(q^{1-h}w_j)} \times
  \begin{tikzpicture}[baseline=0.85cm]
     \foreach \x in {0,.5}
     {  
       \draw[double] (\x,0) -- (\x,1.75);
       \draw[double,xshift=2cm] (\x,0) -- (\x,1.75);
       \draw[double,postaction={on each segment={mid arrow}}] (\x,-.5) -- (\x,0);
       \draw[double,xshift=2cm,postaction={on each segment={mid arrow}}] (\x,-.5) -- (\x,0);
     }
     \draw[double] (3,0) -- (3,1.75);                
     \draw[double] (3.75,0) -- (3.75,1.75);   
     \draw[double,postaction={on each segment={mid arrow}}] (3,-.5) -- (3,0);
     \draw[double,postaction={on each segment={mid arrow}}] (3.75,-.5) -- (3.75,0);
     \foreach \x in {0,2}
     {  
        \draw [double,xshift=\x cm] (.5,1.75) arc [start angle=0, end angle = 180, radius=.25cm];
        \fill[xshift=\x cm] (0.25,2.0) circle (1.5pt) ;
     }
     \draw (0.25,2.0) node[above] {$x_1$} ;
     \draw (2.25,2.0) node[above] {$x_n$};
     \draw[double,postaction={on each segment={mid arrow}}] (3,2) -- (3,1.5);
     \draw[double] (3.75,1.5) -- (3.75,2);     
     \draw[double,postaction={on each segment={mid arrow}}] (3.75,1.65) -- (3.75,2);     
     \foreach \x in {0,3.5}
       \draw[xshift=\x cm] (-.25,1.1) node {$\vdots$};
     \foreach \y in {0,2.0}
       \draw[yshift=\y cm] (1.25,-.25) node {$\cdots$};
     \foreach \y in {0,0.5,1.5}
     {
     \draw[postaction={on each segment={mid arrow}},yshift=\y cm] (0,0)--(-0.5,0);
     \draw[postaction={on each segment={mid arrow}},yshift=\y cm] (3.0,0)--(3.5,0);
     \draw[yshift=\y cm]  (0,0)--(3.0,0);
     }     
     \draw (-.5,0) node [left] {$y_1$};
     \draw (-.5,0.5) node [left] {$y_2$};
     \draw (-.5,1.5) node [left] {$y_{2n+2}$};
     \draw (0,-.79) node {$x_1$};
     \draw (.6,-.75) node  {$x_1^{-1}$};
      \draw (2,-.79) node {$x_n$};
     \draw (2.6,-.75) node  {$x_n^{-1}$};
     \draw (3.0,-.79) node {$1$};
    \end{tikzpicture}
    \label{eq:Zbeq1}
\end{alignat}
where the $y_j$ are specified in terms of the heights of $\bm h$: $\{y_j\}_{j=1}^{2n+2} = \{w_i, q w_i\}_{i|h_i = 2}$. Each $\bm h = (h_1, \dots, h_{2n+1})$ contributing to \eqref{eq:Zbeq1} satisfies $\sum_{i=1}^{2n+1} h_i = 2n+2$. From this constraint and the result of \cref{lemma:wheel}, we find that the heights satisfy
\begin{equation} h_{2n+1} = 2, \qquad h_{2j-1}+h_{2j} = 2, \qquad j = 1, \dots, n.\end{equation}
The parametrisation $h_{2j-1}(\bm \epsilon) = 1+\epsilon_j$, $h_{2j}(\bm \epsilon) = 1-\epsilon_j$ with $\epsilon_j = \pm 1$
is therefore convenient. For $j=n+1$, $h_{2n+1}=2$, so $\epsilon_{n+1}$ equals $1$ for each term contributing to \eqref{eq:Zbeq1}.

The next step consists in noting that each horizontal line in the diagram in \eqref{eq:Zbeq1} can be extended to the right beyond the last column at the cost of adding a prefactor: 
\begin{equation}
\begin{tikzpicture}[baseline=0.85cm]
     \foreach \x in {0,.5}
     {  
       \draw[double] (\x,0) -- (\x,1.75);
       \draw[double,xshift=2cm] (\x,0) -- (\x,1.75);
       \draw[double,postaction={on each segment={mid arrow}}] (\x,-.5) -- (\x,0);
       \draw[double,xshift=2cm,postaction={on each segment={mid arrow}}] (\x,-.5) -- (\x,0);
     }
     \draw[double] (3,0) -- (3,1.75);                
     \draw[double] (3.75,0) -- (3.75,1.75);   
     \draw[double,postaction={on each segment={mid arrow}}] (3,-.5) -- (3,0);
     \draw[double,postaction={on each segment={mid arrow}}] (3.75,-.5) -- (3.75,0);
     \foreach \x in {0,2}
     {  
        \draw [double,xshift=\x cm] (.5,1.75) arc [start angle=0, end angle = 180, radius=.25cm];
        \fill[xshift=\x cm] (0.25,2.0) circle (1.5pt) ;
     }
     \draw (0.25,2.0) node[above] {$x_1$} ;
     \draw (2.25,2.0) node[above] {$x_n$};
     \draw[double,postaction={on each segment={mid arrow}}] (3,2) -- (3,1.5);
     \draw[double](3.75,1.5) -- (3.75,2);    
     \draw[double,postaction={on each segment={mid arrow}}] (3.75,1.65) -- (3.75,2);     
     \foreach \x in {0,3.5}
       \draw[xshift=\x cm] (-.25,1.1) node {$\vdots$};
     \foreach \y in {0,2.0}
       \draw[yshift=\y cm] (1.25,-.25) node {$\cdots$};
     \foreach \y in {0,0.5,1.5}
     {
     \draw[postaction={on each segment={mid arrow}},yshift=\y cm] (0,0)--(-0.5,0);
     \draw[postaction={on each segment={mid arrow}},yshift=\y cm] (3.0,0)--(3.5,0);
     \draw[yshift=\y cm]  (0,0)--(3.0,0);
     }     
     \draw (-.5,0) node [left] {$y_1$};
     \draw (-.5,0.5) node [left] {$y_2$};
     \draw (-.5,1.5) node [left] {$y_{2n+2}$};
     \draw (0,-.79) node {$x_1$};
     \draw (.6,-.75) node  {$x_1^{-1}$};
      \draw (2,-.79) node {$x_n$};
     \draw (2.6,-.75) node  {$x_n^{-1}$};
     \draw (3.0,-.79) node {$1$};
    \end{tikzpicture} \ \
= \prod_{j=1}^{2n+2} \frac{1}{[q y_j]} \times 
\begin{tikzpicture}[baseline=0.85cm]
     \foreach \x in {0,.5}
     {  
       \draw[double] (\x,0) -- (\x,1.75);
       \draw[double,xshift=2cm] (\x,0) -- (\x,1.75);
       \draw[double,postaction={on each segment={mid arrow}}] (\x,-.5) -- (\x,0);
       \draw[double,xshift=2cm,postaction={on each segment={mid arrow}}] (\x,-.5) -- (\x,0);
     }
     \draw[double] (3,0) -- (3,1.75);                
     \draw[double] (3.5,0) -- (3.5,1.75);   
     \draw[double,postaction={on each segment={mid arrow}}] (3,-.5) -- (3,0);
     \draw[double,postaction={on each segment={mid arrow}}] (3.5,-.5) -- (3.5,0);
     \foreach \x in {0,2}
     {  
        \draw [double,xshift=\x cm] (.5,1.75) arc [start angle=0, end angle = 180, radius=.25cm];
        \fill[xshift=\x cm] (0.25,2.0) circle (1.5pt) ;
     }
     \draw (0.25,2.0) node[above] {$x_1$} ;
     \draw (2.25,2.0) node[above] {$x_n$};
     \draw[double,postaction={on each segment={mid arrow}}] (3,2) -- (3,1.5);
     \draw[double] (3.5,1.5) -- (3.5,2);   
     \draw[double,postaction={on each segment={mid arrow}}] (3.5,1.65) -- (3.5,2);     
     \foreach \x in {0,4.0}
       \draw[xshift=\x cm] (-.25,1.1) node {$\vdots$};
     \foreach \y in {0,2.0}
       \draw[yshift=\y cm] (1.25,-.25) node {$\cdots$};
     \foreach \y in {0,0.5,1.5}
     {
     \draw[postaction={on each segment={mid arrow}},yshift=\y cm] (0,0)--(-0.5,0);
     \draw[postaction={on each segment={mid arrow}},yshift=\y cm] (3.5,0)--(4.0,0);
     \draw[yshift=\y cm]  (0,0)--(3.5,0);
     }     
     \draw (-.5,0) node [left] {$y_1$};
     \draw (-.5,0.5) node [left] {$y_2$};
     \draw (-.5,1.5) node [left] {$y_{2n+2}$};
     \draw (0,-.79) node {$x_1$};
     \draw (.6,-.75) node  {$x_1^{-1}$};
      \draw (2,-.79) node {$x_n$};
     \draw (2.6,-.75) node  {$x_n^{-1}$};
     \draw (3.0,-.79) node {$1$};
     \draw (3.5,-.79) node {$1$};
     \draw (4.35,0.85) node  {$.$};
    \end{tikzpicture}
\end{equation}
Finally, using
\begin{equation}
\begin{tikzpicture}[baseline=-0.1cm]
	\draw[double] (-0.25,-0.25) -- (-0.25,0.25)  (0.25,0.25) -- (0.25,-0.25);
	\draw[double,postaction={on each segment={mid arrow}}]  (-0.25,-0.15) -- (-0.25,0.25)  (0.25,0.25) -- (0.25,-0.25);
	\end{tikzpicture}
 \ \, = \ \, 
\frac{1}{[q][q^2]}
\lim_{b'\rightarrow q^{-1}} \ 
	\begin{tikzpicture}[baseline=-0.1cm]
          	\draw[double] (0.25,-0.25) -- (0.25,0.25) --  (-0.25,0.25) -- (-0.25,-0.25);
	        \fill (0,0.25) circle (1.5pt);
	        \draw (0,.25) node [above] {$1$};
	\end{tikzpicture}
	\end{equation}
leads to the result of the proposition.
\end{proof}
Ultimately, we want to compute
\begin{equation}\label{eq:phioddlimit}
(\phi_{\text{\rm \tiny AD}})_{\Downarrow\Uparrow \cdots \Downarrow \Uparrow \Downarrow} = \lim_{x_1, \dots, x_n \rightarrow 1} \lim_{b,b' \rightarrow q^{-1}}   \frac{Z_{b,b'}}{([q][q^2])^n} =  \lim_{x_1, \dots, x_n \rightarrow 1}  \frac{Z_{q^{-1}}}{([q][q^2])^n},
\end{equation}
so from here onwards, we shall focus on the specialisation $b=b'= q^{-1}$. 
\begin{proposition} We have
\begin{equation}\label{eq:Zqm1}
Z_{q^{-1}} = \frac{[q]^n}{[q^2]} \lim_{x_{n+1}\rightarrow 1} \prod_{1\le i<j\le 2n+2}\frac{[q w_i/w_j]}{[w_i/w_j]} \prod_{j=1}^{n+1}\frac{[q^2 x_j^2]}{[x_j^2]} \Det{i,j=1}{n+1}M_{\text{\rm \tiny UU}}(\{x_i\}_{i=1}^{n+1}; \{x_i\}_{i=1}^{n+1}; 1,1)_{ij}.
\end{equation}
\end{proposition}
\begin{proof}
At $b=b'=q^{-1}$, the diagram $D_{b,b'}$ defined in \eqref{eq:Dbb} becomes a partition function on a lattice of width $N' = N+1 = 2(n+1)$ that we have computed in \cref{prop:ZAspin1}:
\begin{equation}
\lim_{b,b'\rightarrow q^{-1}} D_{b,b'}= Z_{\text{\rm \tiny A}}(\{x_i\}_{i=1}^{n+1};\{x_i^{\epsilon_i}\}_{i=1}^{n+1},\{qx_i^{\epsilon_i}\}_{i=1}^{n+1})\Big|_{b\rightarrow q^{-1}},
\end{equation}
with the specialisation $x_{n+1}=1$. 

We now promote $x_{n+1}$ to a formal parameter and take the limit $x_{n+1}\rightarrow 1$ at the end. This is done by assigning the inhomogeneity parameters $w_{2n+1} = x_{n+1}$ and $w_{2n+2}=x_{n+1}^{-1}$ respectively to the second last and last column of $D_{q^{-1}\!,q^{-1}}$, as well as $y_{2n+1}=x_{n+1}$ and $y_{2n+2}=qx_{n+1}$ to its top two rows.  We can pursue the calculation in terms of the heights profiles $\bm h$ of length $2n+2$, by considering that the last height is always $h_{2n+2} = 1- \epsilon_{n+1} = 0$. With this convention,
$\psi_{2n+1}(\bm h)$ and $a_{2n+1}(z)$ can be written as follows in the limit $x_{n+1}\rightarrow 1$: 
\begin{subequations}
\begin{equation}
\lim_{x_{n+1}\rightarrow 1} \psi_{2n+1}(\bm h) = \frac{2}{[q^2]}\lim_{x_{n+1}\rightarrow 1} \psi_{2n+2}(\bm h) \frac{\prod_{i=1}^{n+1}[x_i]}{\prod_{i=1}^{n}[q^{2\epsilon_i}x_i]}, \qquad
\lim_{x_{n+1}\rightarrow 1} \frac{a_{2n+2}(z)}{a_{2n+1}(z)} = [q z].
\end{equation} 
\end{subequations}
This yields
\begin{equation}
Z_{q^{-1}} = \frac{2}{[q][q^2]^2}\lim_{x_{n+1}\rightarrow 1} \sum_{\epsilon_1, \dots, \epsilon_n = \pm 1} \frac{\prod_{i=1}^{n+1}[x_i]}{\prod_{i=1}^n [q^{2 \epsilon_i} x_i]} \frac{\psi_{2n+2}(\bm h)Z_{\text{\rm \tiny A}}(\{x_i\}_{i=1}^{n+1};\{x_i^{\epsilon_i}\}_{i=1}^{n+1},\{qx_i^{\epsilon_i}\}_{i=1}^{n+1})}{\prod_{i=1}^{n+1}a_{2n+2}(q x_i^{\epsilon_i})a_{2n+2}(x_i^{\epsilon_i})}\Big|_{b\rightarrow q^{-1}}.
\end{equation}
At this point, we apply intermediate results found while proving \cref{prop:XiADZ2}, namely \eqref{eq:ZAmonom} and \eqref{eq:psihmonom}, and find after simplification:
\begin{alignat}{2}
Z_{q^{-1}} = \frac{-2[q]^{n+2}}{[q^2]^2}\lim_{x_{n+1}\rightarrow 1}\prod_{1\le i<j\le 2n+2}\frac{[q w_i/w_j]}{[w_i/w_j]} \prod_{j=1}^{n+1}[q^2 x_j^2]& \sum_{\epsilon_1, \dots, \epsilon_n = \pm 1} \prod_{i=1}^n(-1)^{\frac{1+\epsilon_i}2}\frac{[q^{\epsilon_i}x_i]^2}{[q^{2\epsilon_i}x_i^2]}
 \\&
 \times\Det{i,j=1}{n+1}\Big(\frac{1}{[q^{\epsilon_i} x_i/x_j]}-\frac{1}{[q^{\epsilon_i} x_ix_j]}\Big).\nonumber
\end{alignat}
Performing matrix operations allows us to express the sum in terms of a single determinant. The product over $i$ is first split and inserted separately as prefactors in the rows $i = 1, \dots, n$. Each sum over $\epsilon_i$ is then carried out by taking the corresponding linear combination of the $i$-th row of the matrix in the determinant. Using \eqref{eq:magicalidentity} at $b=1$, we obtain
\begin{equation}
Z_{q^{-1}} = \frac{-2[q]^{n+2}}{[q^2]}\lim_{x_{n+1}\rightarrow 1}\prod_{1\le i<j\le 2n+2}\frac{[q w_i/w_j]}{[w_i/w_j]} \prod_{j=1}^{n}[q^2 x_j^2] \Det{i,j=1}{n+1}M_{ij}
\end{equation}
where 
\begin{equation}
M_{ij} = \left\{\begin{array}{cl}
\displaystyle\frac{[x_j]}{[x_i][x_j^2]} M_{\text{\rm \tiny UU}}(\{x_i\}_{i=1}^{n+1}; \{x_i\}_{i=1}^{n+1}; 1,1)_{ij}, &\quad i = 1, \dots, n, \\[0.2cm]
\displaystyle\frac{1}{[q x_{n+1}/x_j]}-\frac{1}{[q x_{n+1}x_j]},&\quad i = n+1.
\end{array}\right.
\end{equation}
Let us define the matrix $\hat M$ with entries
\begin{equation}
\hat M_{ij} =  \frac{[x_j]}{[x_i][x_j^2]} M_{\text{\rm \tiny UU}}(\{x_i\}_{i=1}^n; \{x_i\}_{i=1}^n; 1,1)_{ij}, \qquad i,j = 1, \dots, n+1.
\end{equation}
This matrix only differs from $M$ in its last row. For $j=n+1$, the matrix entries $M_{ij}$ and $\hat M_{ij}$ both vanish for $x_{n+1} \rightarrow 1$. In that limit, we observe that
\begin{equation}
\lim_{x_{n+1}\rightarrow 1} \frac{M_{n+1,j}}{\hat M_{n+1,j}} = - \frac{[q^2]}{2[q]^2}, \qquad j=1, \dots, n+1,
\end{equation}
from which we deduce that
\begin{equation}
\lim_{x_{n+1}\rightarrow 1} \frac{\displaystyle\Det{i,j=1}{n+1}M_{ij}}{[x_{n+1}]} = - \frac{[q^2]}{2 [q]^2} \lim_{x_{n+1}\rightarrow 1} \frac{\displaystyle\Det{i,j=1}{n+1}M_{\text{\rm \tiny UU}}(\{x_i\}_{i=1}^{n+1}; \{x_i\}_{i=1}^{n+1}; 1,1)_{ij}}{[x_{n+1}]\prod_{j=1}^{n+1}[x_i^2]}.
\end{equation}
This ends the proof of the proposition.
\end{proof}

\subsubsection{Homogeneous limit}

\paragraph{Even size.}
The result of \cref{prop:XiADZ2} allows us to compute two components of the vector $|\phi_{\text{\rm \tiny AD}}\rangle$ for $N=2n$ sites.  From \eqref{eqn:defXiAD}, we have
\begin{equation}
  \Xi_{\text{\rm \tiny AD}}(1,\dots,1;b) = (\langle \chit|^{\otimes n}) |\phi_{\text{\rm \tiny AD}}\rangle
\end{equation}
where $|\chit\rangle$ is defined in
\eqref{eqn:defChiHom}. The left-hand side of this equality can be evaluated by taking the homogeneous limit of \eqref{eq:Xievaluation}.
 We express the result in terms of a class of polynomials defined by Kuperberg~\cite{kuperberg:02}:\footnote{We note
 a missing factor of $([q^2]/[q])^{2n}$ in Kuperberg's definition of $A^{(2)}_{\text{\rm \tiny UU}}(4n;t,y,z)$.}
\begin{equation}
  A^{(2)}_{\text{\rm \tiny UU}}(4n;t,y,z) =  \frac{[q^2]^{2n}Z^{(2)}_{\text{\rm \tiny UU}}(1,\dots,1;1,\dots,1;b,c)}{[q]^{n(2n+1)}([b/q][c/q])^n} 
  \label{eqn:defA2UU}
\end{equation}
where 
\begin{equation}
  t=x^2= \left(\frac{[q^2]}{[q]}\right)^2,\quad y = \frac{[bq]}{[b/q]}, \quad z = \frac{[cq]}{[c/q]}.
\end{equation}
Using these, we find:
\begin{equation}\label{eq:Xihomogeneous}
  (\langle \chit|^{\otimes n}) |\phi_{\text{\rm \tiny AD}}\rangle
  = \left(-\frac{[q][b/q][qb]}{
  [q^2]}\right)^n \!\!A^{(2)}_{\text{\rm \tiny UU}}(4n;t,y,y^{-1})  
  = \left(\frac{xy(x^2-4)}{1+(2+y-x^2)y}\right)^n\!\! A_{\text{\rm \tiny UU}}^{(2)}(4n;x^2,y,y^{-1}).
\end{equation}
We obtain two interesting components from this equation through the specialisation of the parameter $b$ 
(or equivalently of $y$). First, for $b=1$ (and thus $y=-1$), we have
$(\langle \chit|^{\otimes n}) |\phi_{\text{\rm \tiny AD}}\rangle=[q]^{2n}(\phi_{\text{\rm \tiny AD}})_{0\cdots 0}$, which leads to
\begin{equation}
  (\phi_{\text{\rm \tiny AD}})_{0\cdots 0} = (-x)^{-n}
  A^{(2)}_{\text{\rm \tiny UU}}(4n;x^2,-1,-1)=(-x)^{n}\tilde A^{(2)}_{\text{\rm \tiny UU}}(4n;x^2).
\end{equation}
Second,
for $b\to q$ (and thus $y\to \infty$), we obtain $(\langle \chit|^{\otimes n}) |\phi_{\text{\rm \tiny AD}}\rangle=([q][q^2])^{n}(\phi_{\text{\rm \tiny AD}})_{\Uparrow\Downarrow\cdots \Uparrow\Downarrow}$, and after elimination of some factors:
\begin{equation}
  (\phi_{\text{\rm \tiny AD}})_{\Uparrow\Downarrow\cdots \Uparrow\Downarrow} = \lim_{y\to \infty} y^{-n} A^{(2)}_{\text{\rm \tiny UU}}(4n;x^2,y,y^{-1}) = A^{(2)}_{\text{\rm \tiny VHP}}(4n+2;x^2).
\end{equation}
As the next proposition shows, these components can be expressed as the determinant of $n\times n$ matrices with polynomial entries in $x^2$.
 
\begin{proposition} \label{prop:prop510}
The polynomial $A^{(2)}_{\text{\rm \tiny UU}}(4n;t,y,z)$ can be expressed as
\begin{align}
  A^{(2)}_{\text{\rm \tiny UU}}(4n;t,y,z)= &\Det{i,j=0}{n-1}
  \Biggr[\ 
  \sum_{k=0}^{n-1}(t+(1+y)(1+z))
  \binom{i+j}{i+k}\binom{i+k}{2k} t^{k}\label{eq:bigDET}\\
  & +\Biggr((1+z)\binom{i+j}{i+k}\binom{i+k}{2k+1}+(1+y)\binom{i+j}{i+k+1}\binom{i+k+1}{2k+1}\Biggl)t^{k+1} 
  \Biggl].
  \nonumber
\end{align}
\begin{proof}
  The proof is very similar to the one for \cref{prop:AVDet}. We combine \eqref{eqn:defZ2UU}, \eqref{eqn:defMUU} with \eqref{eqn:defA2UU}, and introduce the variables $X_i=x_{i+1}^2+1/x_{i+1}^2,\,Y_i=y_{i+1}^2+1/y_{i+1}^2$ for $\,i=0,\dots,n-1$. We obtain
  \begin{equation}
    A^{(2)}_{\text{\rm \tiny UU}}(4n;t,y,z) = [q]^{2n^2}
    \Det{i,j=0}{n-1} \left(\{X^i Y^j\} (A X + B Y +C) f(X+2,Y+2)\right),
  \end{equation}
  where $f(X,Y)$ is defined in \eqref{eqn:deff}, and
	\begin{equation}
    A = z t-(1+y)(1+z),\quad B = yt-(1+y)(1+z),\quad
    C = (t-4)(t+(1+y)(1+z)).
  \end{equation}
  The matrix elements within the determinant can then be  deduced from the series expansion \eqref{eqn:coeffsf} of the function $f(X+2,Y+2)$.
\end{proof}
\end{proposition}

The determinant formulas in \cref{thm:componentsAD1,thm:componentsAD2} are by combining \eqref{eq:Xihomogeneous} and \eqref{eq:bigDET} and via the respective specialisations $y=-1$ and $y\to \infty$. This ends the proofs of these theorems.

\paragraph{Odd size.} Next, we compute the special component $(\phi_{\text{\rm \tiny AD}})_{\Downarrow\Uparrow \dots \Downarrow \Uparrow \Downarrow}$ for $N=2n+1$ sites. Using \eqref{eq:phioddlimit} and \eqref{eq:Zqm1} and comparing with \eqref{eq:Xievaluation2}, we obtain
\begin{equation}
(\phi_{\text{\rm \tiny AD}})_{\Downarrow\Uparrow \dots \Downarrow \Uparrow \Downarrow} =  \lim_{x_1, \dots, x_{n+1} \rightarrow 1} \frac{\Xi_{\text{\rm \tiny AD}}(x_1,\dots,x_{n+1};1)}{([q][q^2])^{n+1}} = (-1)^{n+1} \tilde A^{(2)}_{\text{\rm \tiny UU}}(4(n+1);x^2).
\end{equation}
Here, we used \eqref{eq:Xihomogeneous} evaluated at $y=-1$ at the last equality. By computing $\langle{\Downarrow \Uparrow \cdots \Downarrow \Uparrow \Downarrow}| F | \phi_{\text{\rm \tiny AD}} \rangle$ in two possible ways, one of which requires \cref{prop:SRantidiagonal}, we see that 
\begin{equation}
(\phi_{\text{\rm \tiny AD}})_{\Uparrow \Downarrow \cdots \Uparrow \Downarrow \Uparrow} = - (\phi_{\text{\rm \tiny AD}})_{\Downarrow\Uparrow \cdots \Downarrow \Uparrow \Downarrow}.
\end{equation}
This ends the proof of the special component \eqref{eq:ududucomp}.

%%%%%%%%%%%%%%%%%%%%%
%
\section{Conclusion}
%
%%%%%%%%%%%%%%%%%%%%%
\label{sec:conclusion}

In this article, we have revealed a rich combinatorial structure in the integrable spin-one XXZ chain for a specific anti-diagonal twist. We have also presented new results for the same chain but with a diagonal twist.
Specifically,
many 
properties of special zero-energy states of the spin-chain Hamiltonians were investigated for both twists. In a suitable normalisation, a number of scalar products and components
are given by polynomials in the anisotropy parameter which are known generating functions for the enumeration of symmetry classes of ASMs. These results are valid for any system size and any anisotropy.
This is in sharp contrast with the well-known spin-$\frac12$ XXZ chain, where a relation to ASMs only seems to arise at the combinatorial point $\Delta=-1/2$.

To obtain the spin-chain results, we have investigated the transfer matrix of the corresponding twisted inhomogeneous nineteen-vertex model. Using elements of the quantum separation of variables technique, we have found a simple eigenvalue of the transfer matrix for the anti-diagonal twist and constructed the corresponding eigenvector. A similar special eigenvector had previously been obtained for the diagonal twist by means of the algebraic Bethe ansatz. With the help of tools from quantum integrability, such as the fusion procedure of $R$-matrices, the boundary Yang-Baxter equation, the algebraic Bethe ansatz and the quantum separation of variables method, we have derived exact finite-size sum rules and components for the special vectors. Surprisingly, these can be expressed in terms of partition functions of the six-vertex model on various lattice domains. Exploiting bijections between the configurations of the six-vertex model and ASMs, we have recovered sum rules and components for the spin-chain zero-energy states in terms of combinatorial quantities. As a side result, we have obtained new determinant and pfaffian formulas for some generating functions of ASM enumeration.

The results of this paper lead to several challenging open questions and directions for future work. Of central interest is the problem of finding explicit formulas for all components of the inhomogeneous special eigenvectors (in the canonical basis). There is hope these can be expressed by contour integral formulas, like those obtained in \cite{razumov:07} and \cite{fonseca:12} respectively for the spin-$\tfrac12$ and higher-spin vertex models at their combinatorial points.

A systematic computation of all components might elucidate the intriguing appearance of partition functions of the six-vertex model within the nineteen-vertex model. 
This feature, which we have encountered many times in this paper as the result of algebraic simplifications, is admittedly poorly understood at the moment.
We expect it could be crucial in 
understanding whether all components of the special eigenvectors have a combinatorial meaning in the homogeneous limit.

Another interesting open problem regards the conjectures about the ground-state eigenvalues of the spin chain presented at the end of \cref{sec:hamXXZ}.
We expect their proofs
to be easier in the subsectors of the Hilbert space where the lattice supersymmetry of the integrable spin-one XXZ chain is present. 
In these subsectors, the spectrum is automatically non-negative and any zero-energy state is a ground state. Furthermore, the dimension of the space of zero-energy states can be computed using cohomology arguments.
The characterisation of the ground-state eigenvalue in sectors without supersymmetry remains a
challenging open problem.

A further avenue for future work is the computation of finite-size correlation functions. For the diagonal twist, we expect
these can be explicitly computed for the inhomogeneous special eigenvector by following the ideas
of \cite{kitanine:01}. In particular, emptiness formation probabilities are the easiest correlation functions to study using the formalism of the QISM. 
The evaluation of their homogeneous limit could be quite challenging but not impossible. 
Their behaviour as $N\to \infty$ is also of great interest.
For the homogeneous nineteen-vertex model and the spin-one XXZ chain with $-2 \le x \le 2$, the long-distance behaviour of correlation functions is expected to be described by an $\mathcal N=1$ superconformal field theory with central charge $c=3/2$ \cite{difrancesco:88,baranowski:90,pearce:91}. 
The two different twists considered in this paper may correspond to different sectors of the superconformal field theory. To our knowledge, this connection is poorly understood and deserves to be clarified.

Finally, we mention two possible directions that venture beyond the integrable spin-one XXZ chain. The first is
the extension of our work to the integrable spin-one XYZ chain\cite{fateev:81} with twisted boundary conditions, a one-parameter deformation of the XXZ case.
Its Hamiltonian is also supersymmetric in certain subsectors of the Hilbert space \cite{hagendorf:15}. This allows us to infer the existence of zero-energy states from the XXZ case by a cohomology argument \cite{hagendorf:13}. 
In a suitable normalisation, their components can be expressed as polynomials in \textit{two} variables whose 
possible combinatorial content remains to be explored.
The second natural
direction to explore is the generalisation of the present work to the integrable XXZ chain with arbitrary integer spin. We expect the fusion procedure will allow us to construct eigenvectors of the inhomogeneous higher-spin transfer matrices using the spin-one results. 

\subsubsection*{Acknowledgements}
This work is supported by the Belgian Interuniversity Attraction Poles Program P7/18 through the network DYGEST (Dynamics, Geometry and Statistical Physics). The authors acknowledge hospitality and support from the Galileo Galilei Institute, Florence, and from the program ``Statistical Mechanics, Integrability and Combinatorics'', where part of this work was done.

\appendix

%%%%%%%%%%%%%%%%%%%%%
%
\section*{Appendices}
%
%%%%%%%%%%%%%%%%%%%%%

%%%%%%%%%%%%%%%%%%%%%
%
\section{Solutions to the boundary Yang-Baxter equation}\label{app:BYBE}
%
%%%%%%%%%%%%%%%%%%%%%

In this appendix, we present solutions to the boundary Yang-Baxter equation for the six- and the nineteen-vertex model. In its vector form,\footnote{The boundary Yang-Baxter equation is often presented in a matrix form involving the so-called $K$-matrix \cite{sklyanin:88}. The matrix and vector form are equivalent. In the latter, the matrix elements of $K(z)$ are encoded in the components of the vector $|\chit(z)\rangle$.} this equation reads
\begin{equation}\label{eq:bybe}
\check R_{12}(z/w)\check R_{23}(zw)\big(|\chit(w)\rangle \otimes |\chit(z)\rangle\big) = \check R_{34}(z/w)\check R_{23}(zw)\big(|\chit(z)\rangle \otimes \chit(w) \rangle \big).
\end{equation}
Its diagrammatic representation for the nineteen-vertex model is given in the left panel of \cref{fig:spin.one.ids}.

For the six-vertex model, the $\check R$-matrix is defined by $\check R^{(1,1)}(z) = P^{(1)}R^{(1,1)}(z)$ where $P^{(1)}$ is the permutation operator on $\mathbb C^2 \otimes \mathbb C^2$: $P^{(1)}(|v_1\rangle \otimes |v_2\rangle) = |v_2\rangle \otimes |v_1\rangle$. One can check that 
\begin{equation}\label{eq:solbybe6}
|\chit^{(1)}(z)\rangle = \left[\sqrt qz b\right] |{\uparrow\downarrow} \rangle + \left[\frac{\sqrt qz}{b}\right] |{\downarrow\uparrow} \rangle
\end{equation}
solves the boundary Yang-Baxter equation.
Here, $b$ is a free parameter which takes the same value for $|\chit^{(1)}(w)\rangle$ and $|\chit^{(1)}(z)\rangle$ in \eqref{eq:bybe}.
The state $|\chit^{(1)}(z)\rangle$ also satisfies a so-called {\it fish equation}: 
\begin{equation}\label{eq:fish6}
  \check R^{(1,1)}(z^{-2})|\chit^{(1)}(z)\rangle = [q z^2]|\chit^{(1)}(z^{-1})\rangle.
\end{equation}

For the nineteen-vertex model, $\check R^{(2,2)}(z) \equiv \check R(z)$ and $|\chit^{(2)}(z)\rangle\equiv|\chit(z)\rangle$ are  respectively defined in \eqref{eqn:Rcheck} and \eqref{eq:chistate}.
An explicit computation shows that they solve the boundary Yang-Baxter equation as well as the fish equation
\begin{equation}\label{eq:fish19}
  \check R^{(2,2)}(z^{-2})|\chit^{(2)}(z)\rangle = [q z^2] [q^2 z^2]|\chit^{(2)}(z^{-1})\rangle.
\end{equation}
This relation is represented diagrammatically in the right panel of \cref{fig:spin.one.ids}.\footnote{According to our conventions for the diagrammatic calculus, these diagrams actually illustrate the dual relations to \eqref{eq:bybe} and \eqref{eq:fish19}, obtained by taking their transpose.}
The state $|\chit^{(2)}(z)\rangle$ can in fact be constructed from the solution \eqref{eq:solbybe6} for the six-vertex model as
\begin{equation}
|\chit^{(2)}(z)\rangle = \frac{1}{[z^2]} Q_{12}Q_{34}V_{12} U_{34} P_{12}^+P_{34}^+\check R^{(1,1)}_{23}(z^2) \big(|\chit^{(1)}(q^{1/2}z)\rangle \otimes |\chit^{(1)}(q^{-1/2}z)\rangle \big) 
\end{equation}
where the operators $Q, U$ and $P^+$ are defined in \cref{sec:Rmat+fusion} and
\begin{equation}
V= (U^{-1})^t= \begin{pmatrix} 1 & 0 & 0 & 0 \\ 0 & \alpha' & \alpha' & 0 \\ 0 & 0 & 0 & 1 \\ 0 & \alpha' & -\alpha' & 0 \end{pmatrix}, \qquad \alpha' = \sqrt{\frac{[q]}{[q^2]}}.
\end{equation}
The relations \eqref{eq:bybe} and \eqref{eq:fish19} for the nineteen-vertex model can be shown to hold as a consequence of the fusion procedure and the local relations for the six-vertex model.

\begin{figure}[h]  
  \centering
  \begin{tikzpicture}
    \begin{scope}
        
    \draw[double] (0.5,1.0) arc [start angle=0, end angle = 180, radius=.25cm];
    \fill (0.25,1.25) circle (1.5pt) ;
    \draw (0.25,1.25) node[above] {\scriptsize$w$};

    \draw[double] (1.5,1.0) arc [start angle=0, end angle = 180, radius=.25cm];
    \fill (1.25,1.25) circle (1.5pt) ;
    \draw (1.25,1.25) node[above] {\scriptsize$z$};
    
    \draw 	(0,-0.20) node {\scriptsize$z$} 
    		(0.5,-0.20) node {\scriptsize$w$}
    		(1.03,-0.14) node {\scriptsize$w^{-1}$}
    		(1.58,-0.14) node {\scriptsize$z^{-1}$};		
    
    \draw[double] (0,0) .. controls (0,0.25) and (0.5,0.25) .. (0.5,0.5)
    	     (0.5,0) .. controls (0.5,0.25) and (0,0.25) .. (0,0.5)
	     (0.5,0.5) .. controls (0.5,0.75) and (1.0,0.75) .. (1.0,1.0)
    	     (1.0,0.5) .. controls (1.0,0.75) and (0.5,0.75) .. (0.5,1.0);
    
    \draw[double] (1,0) -- (1,0.5)  (1.5,0) -- (1.5,0.5)  (0,0.5) -- (0,1)  (1.5,0.5) -- (1.5,1);

    \end{scope}

     \begin{scope}[xshift=2.2cm]
     \draw 	(0,0.5) node {$=$} ;
     \end{scope}

    \begin{scope}[xshift=2.9cm]
        
    \draw[double] (0.5,1.0) arc [start angle=0, end angle = 180, radius=.25cm];
    \fill (0.25,1.25) circle (1.5pt) ;
    \draw (0.25,1.25) node[above] {\scriptsize$z$};

    \draw[double] (1.5,1.0) arc [start angle=0, end angle = 180, radius=.25cm];
    \fill (1.25,1.25) circle (1.5pt) ;
    \draw (1.25,1.25) node[above] {\scriptsize$w$};
    
    \draw 	(0,-0.2) node {\scriptsize$z$} 
    		(0.5,-0.2) node {\scriptsize$w$}
    		(1.03,-0.14) node {\scriptsize$w^{-1}$}
    		(1.58,-0.14) node {\scriptsize$z^{-1}$};		
    
    \draw[double] (1,0) .. controls (1,0.25) and (1.5,0.25) .. (1.5,0.5)
    	     (1.5,0) .. controls (1.5,0.25) and (1,0.25) .. (1,0.5)
	     (0.5,0.5) .. controls (0.5,0.75) and (1.0,0.75) .. (1.0,1.0)
    	     (1.0,0.5) .. controls (1.0,0.75) and (0.5,0.75) .. (0.5,1.0);
    
    \draw[double] (0,0) -- (0,0.5)  (0.5,0) -- (0.5,0.5)  (0,0.5) -- (0,1)  (1.5,0.5) -- (1.5,1);
       
    \end{scope}

    \begin{scope}[xshift=6.5cm,yshift=0.25cm]
    
    \draw[double] (0.5,0.5) arc [start angle=0, end angle = 180, radius=.25cm];
    \fill (0.25,0.75) circle (1.5pt);
    \draw (0.25,0.75) node[above] {\scriptsize$z$};
    
    \draw[double] (0,0) .. controls (0,0.25) and (0.5,0.25) .. (0.5,0.5)
    	     (0.5,0) .. controls (0.5,0.25) and (0,0.25) .. (0,0.5);
	     
    \draw 	(0.56,-0.14) node {\scriptsize$z^{-1}$} 
    		(0,-0.2) node {\scriptsize$z$};
    
    \end{scope}

     \begin{scope}[xshift=8.4cm]
     \draw 	(0,0.5) node {$=\ [q z^2][q^2 z^2]$} ;
     \end{scope}

    \begin{scope}[xshift=9.75cm,yshift=0.25cm]
    
    \draw[double] (0.5,0.5) arc [start angle=0, end angle = 180, radius=.25cm];
    \fill (0.25,0.75) circle (1.5pt);
    \draw (0.25,0.75) node[above] {\scriptsize$z^{-1}$};
    
    \draw[double] (0,0) -- (0,0.5)  (0.5,0) -- (0.5,0.5);
	     
    \draw 	(0,-0.14) node {\scriptsize$z^{-1}$} 
    		(0.5,-0.2) node {\scriptsize$z$};
    
    \end{scope}
    
  \end{tikzpicture}
  \caption{Diagrammatic representations in the nineteen-vertex model of the boundary Yang-Baxter equation and the fish equation. The diagrams for the boundary Yang-Baxter and fish equations of the six-vertex model are obtained by replacing each double strand by a single one and modifying the prefactor in the right panel to $[q z^2]$.}
    \label{fig:spin.one.ids}
\end{figure}
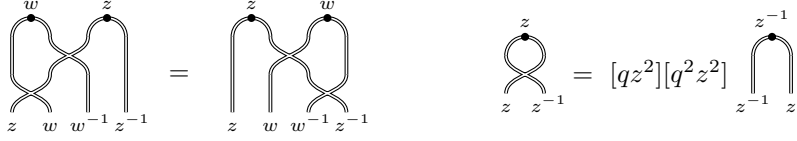

%%%%%%%%%%%%%%%%%%%%%
%
\section{Auxiliary partition functions}
\label{app:ZA}
%
%%%%%%%%%%%%%%%%%%%%%

The objective of this section is to compute $Z_{\text{\rm \tiny A}}(\{x_i\}_{i=1}^n;\{y_i\}_{i=1}^{2n})$. Our strategy is to use the fusion construction of the $R$-matrix of the ten-vertex model to express $Z_{\text{\rm \tiny A}}$ in terms of a partition function of the six-vertex model. The proof we present uses diagrammatic arguments. In particular, we shall depict the solution of the boundary Yang-Baxter equation for the six-vertex model as
\begin{equation}
\langle \chit^{(1)}(z)|  \ \ = \ \ 
\begin{tikzpicture}[baseline=0.5cm]
    \draw (0.5,0.5) arc [start angle=0, end angle = 180, radius=.25cm];
    \fill (0.25,0.75) circle (1.5pt);
    \draw (0.25,0.75) node[above] {$z$};
    \draw (0,0.4) -- (0,0.5)  (0.5,0.4) -- (0.5,0.5);
  \end{tikzpicture}
 \ \, .
\end{equation}
Let us also define the partition function
\begin{equation}
Z_\cap(\{x_i\}_{i=1}^n;\{y_j\}_{j=1}^{n}) = \ \ 
\begin{tikzpicture}[baseline=0.75cm]
	\foreach \x in {0,.5}
     {  
       \draw (\x,0) -- (\x,2);
       \draw[xshift=2cm] (\x,0) -- (\x,2);
       
       \draw[postaction={on each segment={mid arrow}}] (\x,-.5) -- (\x,0);
       \draw[xshift=2cm,postaction={on each segment={mid arrow}}] (\x,-.5) -- (\x,0);
     }
     
     \foreach \x in {0,2}
     {  
        \draw [xshift=\x cm] (.5,2) arc [start angle=0, end angle = 180, radius=.25cm];
        \fill[xshift=\x cm] (0.25,2.25) circle (1.5pt) ;
     }
        
     \draw (0.25,2.25) node[above] {$x_1$} ;
     \draw (2.25,2.25) node[above] {$x_n$} ;
     
     \foreach \x in {0,3}
       \draw[xshift=\x cm] (-.25,1) node {$\vdots$};
     
     \foreach \y in {0,2.25}
       \draw[yshift=\y cm] (1.25,-.25) node {$\cdots$};

     \foreach \y in {0,1.75}
     {
     \draw[postaction={on each segment={mid arrow}},yshift=\y cm] (0,0)--(-0.5,0);
     \draw[postaction={on each segment={mid arrow}},yshift=\y cm] (2.5,0)--(3,0);
     \draw[yshift=\y cm]  (0,0)--(2.5,0);
     }
     
     \draw (-.5,0) node [left] {$y_1$};
     \draw (-.5,1.75) node [left] {$y_{n}$};
     \draw (0,-.797) node {$x_1$};
     \draw (.6,-.75) node  {$x_1^{-1}$};
      \draw (2,-.797) node {$x_n$};
     \draw (2.6,-.75) node  {$x_n^{-1}$};

          \end{tikzpicture}
\end{equation}
where the local weights assigned to vertex configurations are the $\bar {\mathfrak a}(z)$, $\bar {\mathfrak b}(z)$ and $\bar {\mathfrak c}(z)$ given in \cref{fig:6v}.
We have the following proposition.
\begin{lemma}
\label{prop:zcup}
The partition function $Z_\cap$ is given by
\begin{subequations}\label{eq:B.3}
\begin{equation}\label{eq:Zcap}
Z_\cap(\{x_i\}_{i=1}^n;\{y_j\}_{j=1}^{n}) = \frac{[q]^n \prod_{i=1}^n[b/(\sqrt{q}y_i)][qx_i^2]\prod_{i,j=1}^n[x_i/y_j][q^{-1} x_i/y_j][q x_iy_j][x_iy_j]}{\prod_{1\le i<j\le n}[x_i/x_j][y_i/y_j]\prod_{1\le i\le j\le n}[1/(x_ix_j)][qy_iy_j]}  \Det{i,j=1}{n} (M_\cap)_{ij}
\end{equation}
where
\begin{equation}
(M_\cap)_{ij}= \frac1{[x_i/y_j][q^{-1} x_i/y_j]}-\frac1{[x_iy_j][qx_iy_j]}.
\end{equation}
\end{subequations}
\end{lemma}
\begin{proof}
Kuperberg \cite{kuperberg:02} found a formula for a partition function on the same lattice,
but rotated by 90 degrees and defined with the weights $\mathfrak a(z)$, $\mathfrak b(z)$ and $\mathfrak c(z)$. His proof is easily adaptable to the present case. Indeed, $Z_\cap(\{x_i\}_{i=1}^n;\{y_j\}_{j=1}^{n})$ is entirely fixed by the following properties: 
\begin{itemize}
\item[(i)] $Z_\cap(\{x_i\}_{i=1}^n;\{y_j\}_{j=1}^{n})$ is separately symmetric in the $x_i$ and $y_j$.
\item[(ii)] $Z_\cap(\dots,x_i^{-1},\dots;\{y_j\}_{j=1}^n)=[q/x_i^2]/[qx_i^2]Z_\cap(\dots,x_i,\dots;\{y_j\}_{j=1}^n)$.
\item[(iii)] $Z_\cap$ is an odd, centred Laurent polynomial in the variable $y_{n}$ of degree width at most $2(2n-1)$.
\item[(iv)] $Z_\cap$ satisfies the recursion relation 
\begin{equation}
\hspace{-0.2cm}\frac{Z_\cap(\{x_i\}_{i=1}^{n};\{y_j\}_{j=1}^{n-1},y_n = x_n)}{Z_\cap(\{x_i\}_{i=1}^{n-1};\{y_j\}_{j=1}^{n-1})} = [q][q x_n^2][\sqrt{q}y_n/b] \prod_{i=1}^{n-1}[qy_i/x_n][x_ny_i][q y_n/x_i][qx_iy_n].
\end{equation} 
\item[(v)] \eqref{eq:B.3} holds for $n=1$.
\end{itemize}
The verification of each property is completely analogous to the corresponding one in Kuperberg's proof.
\end{proof}

We are now ready to compute $Z_{\text{\rm \tiny A}}(\{x_i\}_{i=1}^n;\{y_i\}_{i=1}^{2n})$. We perform the calculation on $\mathbb C^2 \otimes \mathbb C^2$ instead of $\mathbb C^3$. The state $\langle \chit(z)|$ is then given by
\begin{equation}
[z^2]\langle \chit_{(12)}(z)| = \big(\langle \chi^{(1)}(q^{1/2}z)| \otimes \langle\chi^{(1)}(q^{-1/2}z)|\big)\check R^{(1,1)}_{23}(z^2)P^+_{12}P^+_{34}U^{-1}_{12}V^{-1}_{34} 
\end{equation}
and we draw it as
\begin{equation}
[z^2]\langle \chit_{(12)}(z)| \ = \
\begin{tikzpicture}[baseline=0.25cm]
    \draw (0.5,0.6) arc [start angle=0, end angle = 180, radius=.25cm];
    \fill (0.25,0.85) circle (1.5pt);
    \draw (0.25,0.85) node[above] {\scriptsize$zq^{\frac12}$};
    \draw (0,0.5) -- (0,0.6)  (0.5,0.5) -- (0.5,0.6);
    
    \draw (1.5,0.6) arc [start angle=0, end angle = 180, radius=.25cm];
    \fill (1.25,0.85) circle (1.5pt);
    \draw (1.25,0.85) node[above] {\scriptsize$zq^{-\frac12}$};
    \draw (1,0.5) -- (1,0.6)  (1.5,0.5) -- (1.5,0.6);
    
    \draw (0.5,0) .. controls (0.5,0.25) and (1.0,0.25) .. (1.0,0.5)
    	     (1.0,0) .. controls (1.0,0.25) and (0.5,0.25) .. (0.5,0.5);
	    
    \draw (0,0) -- (0,0.5)   (1.5,0) -- (1.5,0.5);
    
    \fill[gray] (-0.1,-0.1) rectangle (0.6,0);
    \draw (-0.05,-0.3) -- (-0.05,-0.1) -- (0.55,-0.1) -- (0.55,-0.3) -- (-0.05,-0.3);
    \draw (0.25,-0.2) node {$\textrm{\tiny{$_{\bar U}$}}$};
    
    \fill[gray] (0.9,-0.1) rectangle (1.6,0);
    \draw (0.95,-0.3) -- (0.95,-0.1) -- (1.55,-0.1) -- (1.55,-0.3) -- (0.95,-0.3);
    \draw (1.25,-0.2) node {$\textrm{\tiny{$_{\bar V}$}}$};    
    
  \end{tikzpicture} \ .
\end{equation}
In the diagrams, we use the abbreviations $\bar U=U^{-1}$ and $\bar V=V^{-1}$ and draw $P^+$ as a gray rectangle: $P^+ = \ 
\begin{tikzpicture}[baseline=-0.04cm]
	\fill[gray] (0,0) rectangle (0.7,0.1);
\end{tikzpicture}
$ . From \eqref{eq:R12123}, the $R^{(1,2)}$-matrix is drawn in two possible ways:
\begin{equation}\label{eq:Rdiag}
[z/w]R^{(1,2)}_{1(23)}(q^{-1}z/w) =  
\begin{tikzpicture}[baseline=-0.1cm]
	\draw (0.15,0) node[left] {\scriptsize$z$};
	\draw (0.5,-0.55) node[below] {\scriptsize$qw$};
	\draw (1.0,-0.55) node[below] {\scriptsize$w$};
	\draw (0.15,0) -- (1.35,0) (0.5,-0.6) -- (0.5,0.6) (1.0,-0.6) -- (1.0,0.6);
	\fill[gray] (0.4,-0.3) rectangle (1.1,-0.2);
	\draw (0.45,-0.5) -- (0.45,-0.3) -- (1.05,-0.3) -- (1.05,-0.5) -- (0.45,-0.5);
	\fill[white] (0.45,-0.5) rectangle (1.05,-0.3);
	\draw (0.75,-0.4) node {$\textrm{\tiny{$_{\bar U}$}}$};
	
	\draw (0.45,0.5) -- (0.45,0.3) -- (1.05,0.3) -- (1.05,0.5) -- (0.45,0.5);
	\fill[white] (0.45,0.5) rectangle (1.05,0.3);
	\draw (0.75,0.4) node {$\textrm{\tiny{$_{U}$}}$};

\end{tikzpicture} %
\ = 
\begin{tikzpicture}[baseline=-0.1cm]
	\draw (0.15,0) node[left] {\scriptsize$z$};
	\draw (0.5,-0.55) node[below] {\scriptsize$w$};
	\draw (1.0,-0.55) node[below] {\scriptsize$qw$};
	
	\draw (0.15,0) -- (1.35,0) (0.5,-0.6) -- (0.5,0.6) (1.0,-0.6) -- (1.0,0.6);
	\draw (0.45,-0.5) -- (0.45,-0.3) -- (1.05,-0.3) -- (1.05,-0.5) -- (0.45,-0.5);
	\fill[white] (0.45,-0.5) rectangle (1.05,-0.3);
	\draw (0.75,-0.4) node {$\textrm{\tiny{$_{\bar V}$}}$};
	
	\fill[gray] (0.4,0.2) rectangle (1.1,0.3);
	\draw (0.45,0.5) -- (0.45,0.3) -- (1.05,0.3) -- (1.05,0.5) -- (0.45,0.5);
	\fill[white] (0.45,0.5) rectangle (1.05,0.3);
	\draw (0.75,0.4) node {$\textrm{\tiny{$_{V}$}}$};

\end{tikzpicture} \ ,
\end{equation}
where the second diagram is obtained from 
\begin{equation}
R^{(1,2)}_{1(23)}(z) = \big(R^{(1,2)}_{1(23)}(z)\big)^t = \frac{1}{[q z]}V_{23} P^+_{23} R^{(1,1)}_{13}(qz)R^{(1,1)}_{12}(z) V^{-1}_{23}.
\end{equation} As a result, $Z_{\text{\rm \tiny A}}$ is expressible as 
\begin{equation}
Z_{\text{\rm \tiny A}}(\{x_i\}_{i=1}^n;\{y_i\}_{i=1}^{2n}) = \frac{\tilde Z_{\text{\rm \tiny A}}(\{x_i\}_{i=1}^n;\{y_i\}_{i=1}^{2n})}{\prod_{i=1}^n [x_i^2]\prod_{i=1}^n\prod_{j=1}^{2n} [y_j/x_i][y_jx_i]}
\end{equation}
where 
\begin{equation}
\tilde Z_{\text{\rm \tiny A}}(\{x_i\}_{i=1}^n;\{y_i\}_{i=1}^{2n}) = \
\begin{tikzpicture}[baseline=0.95cm]
	\foreach \x in {0,.5,1,1.5,3,3.5,4,4.5}
     	{  
       	\draw (\x,-0.5) -- (\x,2.5);
       	\draw[postaction={on each segment={mid arrow}}] (\x,-1.0) -- (\x,-0.5);
     	}
	\draw[postaction={on each segment={mid arrow}}] (0,0) -- (-0.5,0) (4.5,0) -- (5,0);
	
	\foreach \y in {0,2}
	{
	\draw (0,\y) -- (4.5,\y);
	\draw[postaction={on each segment={mid arrow}}] (0,\y) -- (-0.5,\y) (4.5,\y) -- (5,\y);

	\foreach \x in {0,3}
	{
	\fill[gray] (\x-0.1,-0.2+\y) rectangle (\x+0.6,-0.3+\y);
	\fill[gray] (\x+0.9,0.2+\y) rectangle (\x+1.6,0.3+\y);
	
	\draw (\x-0.05,0.5+\y) -- (\x-0.05,0.3+\y) -- (\x+0.55,0.3+\y) -- (\x+0.55,0.5+\y) -- (\x-0.05,0.5+\y);
	\fill[white] (\x-0.05,0.5+\y) rectangle (\x+0.55,0.3+\y);
	\draw (\x+0.25,0.4+\y) node {$\textrm{\tiny{$_{U}$}}$};
	
	\draw (\x-0.05,-0.5+\y) -- (\x-0.05,-0.3+\y) -- (\x+0.55,-0.3+\y) -- (\x+0.55,-0.5+\y) -- (\x-0.05,-0.5+\y);
	\fill[white] (\x-0.05,-0.5+\y) rectangle (\x+0.55,-0.3+\y);
	\draw (\x+0.25,-0.4+\y) node {$\textrm{\tiny{$_{\bar U}$}}$};
		
	\draw (\x+0.95,0.5+\y) -- (\x+0.95,0.3+\y) -- (\x+1.55,0.3+\y) -- (\x+1.55,0.5+\y) -- (\x+0.95,0.5+\y);
	\fill[white] (\x+0.95,0.5+\y) rectangle (\x+1.55,0.3+\y);
	\draw (\x+1.25,0.4+\y) node {$\textrm{\tiny{$_{V}$}}$};
	
	\draw (\x+0.95,-0.5+\y) -- (\x+0.95,-0.3+\y) -- (\x+1.55,-0.3+\y) -- (\x+1.55,-0.5+\y) -- (\x+0.95,-0.5+\y);
	\fill[white] (\x+0.95,-0.5+\y) rectangle (\x+1.55,-0.3+\y);
	\draw (\x+1.25,-0.4+\y) node {$\textrm{\tiny{$_{\bar V}$}}$};
	}
	}	
	
	\foreach \x in {0,3}
	{
	\draw (\x-0.05,2.5) -- (\x-0.05,2.7) -- (\x+0.55,2.7) -- (\x+0.55,2.5) -- (\x-0.05,2.5);
	\fill[white] (\x-0.05,2.5) rectangle (\x+0.55,2.7);
	\draw (\x+0.25,2.6) node {$\textrm{\tiny{$_{\bar U}$}}$};
	
	\draw (\x+0.95,2.5) -- (\x+0.95,2.7) -- (\x+1.55,2.7) -- (\x+1.55,2.5) -- (\x+0.95,2.5);
	\fill[white] (\x+0.95,2.5) rectangle (\x+1.55,2.7);
	\draw (\x+1.25,2.6) node {$\textrm{\tiny{$_{\bar V}$}}$};
	
	\draw (\x+0.5,2.8) .. controls (\x+0.5,3.05) and (\x+1.0,3.05) .. (\x+1.0,3.3)
    	     	 (\x+1.0,2.8) .. controls (\x+1.0,3.05) and (\x+0.5,3.05) .. (\x+0.5,3.3);
	\draw (0.5+\x,3.3) arc [start angle=0, end angle = 180, radius=.25cm];
	\draw (1.5+\x,3.3) arc [start angle=0, end angle = 180, radius=.25cm];	 
	\draw (\x,2.7) -- (\x,3.3);
	\draw (\x+1.5,2.7) -- (\x+1.5,3.3);
	
	\fill (0.25+\x,3.55) circle (1.5pt);
	\fill (1.25+\x,3.55) circle (1.5pt);
	
	\fill[gray] (\x-0.1,2.7) rectangle (\x+0.6,2.8);
	\fill[gray] (\x+0.9,2.7) rectangle (\x+1.6,2.8);
	}

	\draw (0.25,3.55) node[above] {\scriptsize$x_1q^{\frac12}$};
	\draw (1.25,3.55) node[above] {\scriptsize$x_1q^{-\frac12}$};

	\draw (3.25,3.55) node[above] {\scriptsize$x_nq^{\frac12}$};
	\draw (4.25,3.55) node[above] {\scriptsize$x_nq^{-\frac12}$};	
	
	\draw (-0.5,0) node[left] {\scriptsize$y_1$};
	\draw (-0.5,2) node[left] {\scriptsize$y_{2n}$};
	
	\draw (0,-1.24) node {\scriptsize$qx_1$};
	\draw (0.5,-1.24) node {\scriptsize$x_1$};
	\draw (1,-1.2) node {\scriptsize$x_1^{-1}$};
	\draw (1.6,-1.2) node {\scriptsize$qx_1^{-1}$};	
	
	\draw (3,-1.24) node {\scriptsize$qx_n$};
	\draw (3.5,-1.24) node {\scriptsize$x_n$};
	\draw (4,-1.2) node {\scriptsize$x_n^{-1}$};
	\draw (4.6,-1.2) node {\scriptsize$qx_n^{-1}$};		
	
	\draw (2.25,-0.4) node {$\dots$};
	\draw (2.25,2.5) node {$\dots$};	
	\draw (-0.25,1.1) node {$\vdots$};
	\draw (4.75,1.1) node {$\vdots$};
	
\end{tikzpicture} \ \ .
\label{eq:diagZA}
\end{equation}
Here, we have used the first form for the $R$-matrix \eqref{eq:Rdiag} for odd columns and the second form for even columns. We simplify \eqref{eq:diagZA} by using the following identities:
\begin{subequations}\label{eq:proj.relations}
\begin{equation}
\begin{tikzpicture}[baseline=-0.1cm]
	\draw (0.15,0) node[left] {\scriptsize$z$};
	\draw (0.5,-0.55) node[below] {\scriptsize$qw$};
	\draw (1.0,-0.55) node[below] {\scriptsize$w$};
	\draw (0.15,0) -- (1.35,0) (0.5,-0.6) -- (0.5,0.6) (1.0,-0.6) -- (1.0,0.6);
	\fill[gray] (0.4,-0.3) rectangle (1.1,-0.2);
	\fill[gray] (0.4,0.3) rectangle (1.1,0.2);	
	\draw (0.45,-0.5) -- (0.45,-0.3) -- (1.05,-0.3) -- (1.05,-0.5) -- (0.45,-0.5);
	\fill[white] (0.45,-0.5) rectangle (1.05,-0.3);
	\draw (0.75,-0.4) node {$\textrm{\tiny{$_{\bar U}$}}$};
	
	\draw (0.45,0.5) -- (0.45,0.3) -- (1.05,0.3) -- (1.05,0.5) -- (0.45,0.5);
	\fill[white] (0.45,0.5) rectangle (1.05,0.3);
	\draw (0.75,0.4) node {$\textrm{\tiny{$_{U}$}}$};
\end{tikzpicture} 
\ = \
\begin{tikzpicture}[baseline=-0.1cm]
	\draw (0.15,0) node[left] {\scriptsize$z$};
	\draw (0.5,-0.55) node[below] {\scriptsize$qw$};
	\draw (1.0,-0.55) node[below] {\scriptsize$w$};
	\draw (0.15,0) -- (1.35,0) (0.5,-0.6) -- (0.5,0.6) (1.0,-0.6) -- (1.0,0.6);
	\fill[gray] (0.4,-0.3) rectangle (1.1,-0.2);
	\draw (0.45,-0.5) -- (0.45,-0.3) -- (1.05,-0.3) -- (1.05,-0.5) -- (0.45,-0.5);
	\fill[white] (0.45,-0.5) rectangle (1.05,-0.3);
	\draw (0.75,-0.4) node {$\textrm{\tiny{$_{\bar U}$}}$};
	
	\draw (0.45,0.5) -- (0.45,0.3) -- (1.05,0.3) -- (1.05,0.5) -- (0.45,0.5);
	\fill[white] (0.45,0.5) rectangle (1.05,0.3);
	\draw (0.75,0.4) node {$\textrm{\tiny{$_{U}$}}$};
\end{tikzpicture}\ , \qquad
\begin{tikzpicture}[baseline=-0.1cm]
	\draw (0.15,0) node[left] {\scriptsize$z$};
	\draw (0.5,-0.55) node[below] {\scriptsize$w$};
	\draw (1.0,-0.55) node[below] {\scriptsize$qw$};
	
	\draw (0.15,0) -- (1.35,0) (0.5,-0.6) -- (0.5,0.6) (1.0,-0.6) -- (1.0,0.6);
	\draw (0.45,-0.5) -- (0.45,-0.3) -- (1.05,-0.3) -- (1.05,-0.5) -- (0.45,-0.5);
	\fill[white] (0.45,-0.5) rectangle (1.05,-0.3);
	\draw (0.75,-0.4) node {$\textrm{\tiny{$_{\bar V}$}}$};
	
	\fill[gray] (0.4,0.2) rectangle (1.1,0.3);
	\fill[gray] (0.4,-0.2) rectangle (1.1,-0.3);
	\draw (0.45,0.5) -- (0.45,0.3) -- (1.05,0.3) -- (1.05,0.5) -- (0.45,0.5);
	\fill[white] (0.45,0.5) rectangle (1.05,0.3);
	\draw (0.75,0.4) node {$\textrm{\tiny{$_{V}$}}$};

\end{tikzpicture}
\ = \
 \begin{tikzpicture}[baseline=-0.1cm]
	\draw (0.15,0) node[left] {\scriptsize$z$};
	\draw (0.5,-0.55) node[below] {\scriptsize$w$};
	\draw (1.0,-0.55) node[below] {\scriptsize$qw$};
	
	\draw (0.15,0) -- (1.35,0) (0.5,-0.6) -- (0.5,0.6) (1.0,-0.6) -- (1.0,0.6);
	\draw (0.45,-0.5) -- (0.45,-0.3) -- (1.05,-0.3) -- (1.05,-0.5) -- (0.45,-0.5);
	\fill[white] (0.45,-0.5) rectangle (1.05,-0.3);
	\draw (0.75,-0.4) node {$\textrm{\tiny{$_{\bar V}$}}$};
	
	\fill[gray] (0.4,0.2) rectangle (1.1,0.3);
	\draw (0.45,0.5) -- (0.45,0.3) -- (1.05,0.3) -- (1.05,0.5) -- (0.45,0.5);
	\fill[white] (0.45,0.5) rectangle (1.05,0.3);
	\draw (0.75,0.4) node {$\textrm{\tiny{$_{V}$}}$};

\end{tikzpicture}\ ,
\label{eq:bulkrelations}
\end{equation}
\begin{equation}
\begin{tikzpicture}[baseline=0.25cm]
    \draw (0.5,0.6) arc [start angle=0, end angle = 180, radius=.25cm];
    \fill (0.25,0.85) circle (1.5pt);
    \draw (0.25,0.85) node[above] {\scriptsize$zq^{\frac12}$};
    \draw (0,0.5) -- (0,0.6)  (0.5,0.5) -- (0.5,0.6);
    
    \draw (1.5,0.6) arc [start angle=0, end angle = 180, radius=.25cm];
    \fill (1.25,0.85) circle (1.5pt);
    \draw (1.25,0.85) node[above] {\scriptsize$zq^{-\frac12}$};
    \draw (1,0.5) -- (1,0.6)  (1.5,0.5) -- (1.5,0.6);
    
    \draw (0.5,0) .. controls (0.5,0.25) and (1.0,0.25) .. (1.0,0.5)
    	     (1.0,0) .. controls (1.0,0.25) and (0.5,0.25) .. (0.5,0.5);
	    
    \draw (0,0) -- (0,0.5)   (1.5,0) -- (1.5,0.5);
    
    \fill[gray] (-0.1,-0.1) rectangle (0.6,0);
    \draw (-0.05,-0.3) -- (-0.05,-0.1) -- (0.55,-0.1) -- (0.55,-0.3) -- (-0.05,-0.3);
    \draw (0.25,-0.2) node {$\textrm{\tiny{$_{\bar U}$}}$};
    
    \fill[gray] (0.9,-0.1) rectangle (1.6,0);
    \draw (0.95,-0.3) -- (0.95,-0.1) -- (1.55,-0.1) -- (1.55,-0.3) -- (0.95,-0.3);
    \draw (1.25,-0.2) node {$\textrm{\tiny{$_{\bar V}$}}$};    
    
  \end{tikzpicture} 
  \ = \
  \begin{tikzpicture}[baseline=0.25cm]
    \draw (0.5,0.6) arc [start angle=0, end angle = 180, radius=.25cm];
    \fill (0.25,0.85) circle (1.5pt);
    \draw (0.25,0.85) node[above] {\scriptsize$zq^{\frac12}$};
    \draw (0,0.5) -- (0,0.6)  (0.5,0.5) -- (0.5,0.6);
    
    \draw (1.5,0.6) arc [start angle=0, end angle = 180, radius=.25cm];
    \fill (1.25,0.85) circle (1.5pt);
    \draw (1.25,0.85) node[above] {\scriptsize$zq^{-\frac12}$};
    \draw (1,0.5) -- (1,0.6)  (1.5,0.5) -- (1.5,0.6);
    
    \draw (0.5,0) .. controls (0.5,0.25) and (1.0,0.25) .. (1.0,0.5)
    	     (1.0,0) .. controls (1.0,0.25) and (0.5,0.25) .. (0.5,0.5);
	    
    \draw (0,0) -- (0,0.5)   (1.5,0) -- (1.5,0.5);
    
    \fill[gray] (-0.1,-0.1) rectangle (0.6,0);
    \draw (-0.05,-0.3) -- (-0.05,-0.1) -- (0.55,-0.1) -- (0.55,-0.3) -- (-0.05,-0.3);
    \draw (0.25,-0.2) node {$\textrm{\tiny{$_{\bar U}$}}$};
    
    \draw (0.95,-0.3) -- (0.95,-0.1) -- (1.55,-0.1) -- (1.55,-0.3) -- (0.95,-0.3);
    \draw (1.25,-0.2) node {$\textrm{\tiny{$_{\bar V}$}}$};    
    
    \draw (1,0) -- (1,-0.1) (1.5,0) -- (1.5,-0.1);
    
  \end{tikzpicture} \ .
  \label{eq:toprelation}
\end{equation}
\end{subequations}
Each of these relations is shown algebraically using the elementary property $(P^+)^2 = P^+$ of the projector, its decomposition as $P^+ = B^{-1}R^{(1,1)}(q) = R^{(1,1)}(q) B^{-1}$, as well as either the Yang-Baxter equation \eqref{eq:YB} or the boundary Yang-Baxter equation \eqref{eq:bybe}. 

The relations \eqref{eq:proj.relations} are used multiple times, allowing us to express \eqref{eq:diagZA} without any $P^+$ projectors. Indeed, the second relation in \eqref{eq:bulkrelations} is used to remove projectors in the $V,\bar V$ columns, first at the bottom and progressively towards the top. The projectors attached to the $\langle \chi^{(1)}(x_iq^{-1/2})|$ are then erased using \eqref{eq:toprelation}. We then remove projectors in the $U,\bar U$ columns using the first relation of \eqref{eq:bulkrelations}, starting from the top and progressing downwards. The remaining projectors all lie at the bottom of the diagram and are spurious because $P^+|{\uparrow \uparrow}\rangle = |{\uparrow \uparrow}\rangle$. At this point, one can also remove the operators $U,\bar U,V$ and $\bar V$ using $U\bar U =V\bar V = \bm 1$ and $\bar U|{\uparrow \uparrow}\rangle = \bar V|{\uparrow \uparrow}\rangle =  |{\uparrow \uparrow}\rangle$. We then have:
\begingroup
\allowdisplaybreaks
\begin{alignat}{2}
\tilde Z_{\text{\rm \tiny A}}(\{x_i\}_{i=1}^n;\{y_i\}_{i=1}^{2n}) &= \
\begin{tikzpicture}[baseline=0.45cm]
	\foreach \x in {0,.5,1,1.5,3,3.5,4,4.5}
     	{  
       	\draw (\x,-0.5) -- (\x,1.25);
       	\draw[postaction={on each segment={mid arrow}}] (\x,-0.5) -- (\x,0);
     	}
	\foreach \x in {0,3}{
		\draw (\x+0.5,1.25) .. controls (\x+0.5,1.5) and (\x+1.0,1.5) .. (\x+1.0,1.75)
    	     	 (\x+1.0,1.25) .. controls (\x+1.0,1.5) and (\x+0.5,1.5) .. (\x+0.5,1.75);
		 \draw (0.5+\x,1.75) arc [start angle=0, end angle = 180, radius=.25cm];
		\draw (1.5+\x,1.75) arc [start angle=0, end angle = 180, radius=.25cm];	 
		\draw (\x,1.25) -- (\x,1.75) (\x+1.5,1.25) -- (\x+1.5,1.75);
			\fill (0.25+\x,2) circle (1.5pt);
			\fill (1.25+\x,2) circle (1.5pt);
	}
	\foreach \y in {0,1}
	{
	\draw (0,\y) -- (4.5,\y);
	\draw[postaction={on each segment={mid arrow}}] (0,\y) -- (-0.5,\y);
	\draw[postaction={on each segment={mid arrow}}] (4.5,\y) -- (5,\y);
	}	
	\draw (0.25,2) node[above] {\scriptsize$x_1q^{\frac12}$};
	\draw (1.25,2) node[above] {\scriptsize$x_1q^{-\frac12}$};
	\draw (3.25,2) node[above] {\scriptsize$x_nq^{\frac12}$};
	\draw (4.25,2) node[above] {\scriptsize$x_nq^{-\frac12}$};		
	\draw (-0.5,0) node[left] {\scriptsize$y_1$};
	\draw (-0.5,1) node[left] {\scriptsize$y_{2n}$};
	\draw (0,-0.74) node {\scriptsize$qx_1$};
	\draw (0.5,-0.74) node {\scriptsize$x_1$};
	\draw (1,-0.7) node {\scriptsize$x_1^{-1}$};
	\draw (1.6,-0.7) node {\scriptsize$qx_1^{-1}$};	
	\draw (3,-0.74) node {\scriptsize$qx_n$};
	\draw (3.5,-0.74) node {\scriptsize$x_n$};
	\draw (4,-0.7) node {\scriptsize$x_n^{-1}$};
	\draw (4.6,-0.7) node {\scriptsize$qx_n^{-1}$};	
	\draw (2.25,-0.25) node {$\dots$};
	\draw (2.25,1.5) node {$\dots$};	
	\draw (-0.25,0.6) node {$\vdots$};
	\draw (4.75,0.6) node {$\vdots$};
\end{tikzpicture} \nonumber\\
& = \ 
\begin{tikzpicture}[baseline=0.45cm]
	\foreach \x in {0,.5,1,1.5,3,3.5,4,4.5}
     	{  
       	\draw (\x,0) -- (\x,1.25);
       	\draw[postaction={on each segment={mid arrow}}] (\x,-1) -- (\x,-0.5);
     	}
	\foreach \x in {0,3}{
		\draw (\x+0.5,0) .. controls (\x+0.5,-0.25) and (\x+1.0,-0.25) .. (\x+1.0,-0.5)
    	     	 	 (\x+1.0,0) .. controls (\x+1.0,-0.25) and (\x+0.5,-0.25) .. (\x+0.5,-0.5);
		\draw (0.5+\x,1.25) arc [start angle=0, end angle = 180, radius=.25cm];
		\draw (1.5+\x,1.25) arc [start angle=0, end angle = 180, radius=.25cm];	 
		\draw (\x,0) -- (\x,-0.5) (\x+1.5,0) -- (\x+1.5,-0.5);
			\fill (0.25+\x,1.5) circle (1.5pt);
			\fill (1.25+\x,1.5) circle (1.5pt);
	}
	\foreach \y in {0,1}
	{
	\draw (0,\y) -- (4.5,\y);
	\draw[postaction={on each segment={mid arrow}}] (0,\y) -- (-0.5,\y);
	\draw[postaction={on each segment={mid arrow}}] (4.5,\y) -- (5,\y);
	}	
	\draw (0.25,1.5) node[above] {\scriptsize$x_1q^{\frac12}$};
	\draw (1.25,1.5) node[above] {\scriptsize$x_1q^{-\frac12}$};
	\draw (3.25,1.5) node[above] {\scriptsize$x_nq^{\frac12}$};
	\draw (4.25,1.5) node[above] {\scriptsize$x_nq^{-\frac12}$};		
	\draw (-0.5,0) node[left] {\scriptsize$y_1$};
	\draw (-0.5,1) node[left] {\scriptsize$y_{2n}$};
	\draw (0,-1.24) node {\scriptsize$qx_1$};
	\draw (0.5,-1.24) node {\scriptsize$x_1$};
	\draw (1,-1.2) node {\scriptsize$x_1^{-1}$};
	\draw (1.6,-1.2) node {\scriptsize$qx_1^{-1}$};	
	\draw (3,-1.24) node {\scriptsize$qx_n$};
	\draw (3.5,-1.24) node {\scriptsize$x_n$};
	\draw (4,-1.2) node {\scriptsize$x_n^{-1}$};
	\draw (4.6,-1.2) node {\scriptsize$qx_n^{-1}$};	
	\draw (2.25,-0.5) node {$\dots$};
	\draw (2.25,1.25) node {$\dots$};	
	\draw (-0.25,0.6) node {$\vdots$};
	\draw (4.75,0.6) node {$\vdots$};
\end{tikzpicture}
\nonumber\\ & =\prod_{i=1}^n[q x_i^2] \times
\begin{tikzpicture}[baseline=0.45cm]
	\foreach \x in {0,.5,1,1.5,3,3.5,4,4.5}
     	{  
       	\draw (\x,0) -- (\x,1.25);
       	\draw[postaction={on each segment={mid arrow}}] (\x,-0.5) -- (\x,0);
     	}
	\foreach \x in {0,3}{
		\draw (0.5+\x,1.25) arc [start angle=0, end angle = 180, radius=.25cm];
		\draw (1.5+\x,1.25) arc [start angle=0, end angle = 180, radius=.25cm];	 
			\fill (0.25+\x,1.5) circle (1.5pt);
			\fill (1.25+\x,1.5) circle (1.5pt);
	}
	\foreach \y in {0,1}
	{
	\draw (0,\y) -- (4.5,\y);
	\draw[postaction={on each segment={mid arrow}}] (0,\y) -- (-0.5,\y);
	\draw[postaction={on each segment={mid arrow}}] (4.5,\y) -- (5,\y);
	}	
	\draw (0.25,1.5) node[above] {\scriptsize$x_1q^{\frac12}$};
	\draw (1.25,1.5) node[above] {\scriptsize$x_1q^{-\frac12}$};
	\draw (3.25,1.5) node[above] {\scriptsize$x_nq^{\frac12}$};
	\draw (4.25,1.5) node[above] {\scriptsize$x_nq^{-\frac12}$};		
	\draw (-0.5,0) node[left] {\scriptsize$y_1$};
	\draw (-0.5,1) node[left] {\scriptsize$y_{2n}$};
	\draw (0,-0.74) node {\scriptsize$qx_1$};
	\draw (0.6,-0.7) node {\scriptsize$x_1^{-1}$};
	\draw (1.1,-0.74) node {\scriptsize$x_1$};	
	\draw (1.6,-0.7) node {\scriptsize$qx_1^{-1}$};	
	\draw (3,-0.74) node {\scriptsize$qx_n$};
	\draw (3.6,-0.7) node {\scriptsize$x_n^{-1}$};
	\draw (4.1,-0.74) node {\scriptsize$x_n$};	
	\draw (4.7,-0.7) node {\scriptsize$qx_n^{-1}$};	
	\draw (2.25,-0.25) node {$\dots$};
	\draw (2.25,1.25) node {$\dots$};	
	\draw (-0.25,0.6) node {$\vdots$};
	\draw (4.75,0.6) node {$\vdots$};
\end{tikzpicture}
\nonumber\\ & = \prod_{i=1}^n[q x_i^2] \times
\begin{tikzpicture}[baseline=0.45cm]
	\foreach \x in {0,.5,1,1.5,3,3.5}
     	{  
       	\draw (\x,0) -- (\x,1.25);
       	\draw[postaction={on each segment={mid arrow}}] (\x,-0.5) -- (\x,0);
     	}
	\foreach \x in {0,1,3}{
		\draw (0.5+\x,1.25) arc [start angle=0, end angle = 180, radius=.25cm];
			\fill (0.25+\x,1.5) circle (1.5pt);
	}
	\foreach \y in {0,1}
	{
	\draw (0,\y) -- (3.5,\y);
	\draw[postaction={on each segment={mid arrow}}] (0,\y) -- (-0.5,\y);
	\draw[postaction={on each segment={mid arrow}}] (3.5,\y) -- (4.0,\y);
	}	
	\draw (0.25,1.5) node[above] {\scriptsize$\zeta_1$};
	\draw (1.25,1.5) node[above] {\scriptsize$\zeta_2$};
	\draw (3.25,1.5) node[above] {\scriptsize$\zeta_{2n}$};		
	\draw (-0.5,0.1) node[left] {\scriptsize$q^{-\frac12}y_1$};
	\draw (-0.5,1.1) node[left] {\scriptsize$q^{-\frac12}y_{2n}$};
	\draw (0,-0.71) node {\scriptsize$\zeta_1$};
	\draw (0.5,-0.7) node {\scriptsize$\zeta_1^{-1}$};
	\draw (1,-0.71) node {\scriptsize$\zeta_2$};
	\draw (1.6,-0.7) node {\scriptsize$\zeta_2^{-1}$};	
	\draw (3,-0.71) node {\scriptsize$\zeta_{2n}$};
	\draw (3.5,-0.7) node {\scriptsize$\zeta_{2n}^{-1}$};
	\draw (2.25,-0.25) node {$\dots$};
	\draw (2.25,1.25) node {$\dots$};	
	\draw (-0.25,0.6) node {$\vdots$};
	\draw (3.75,0.6) node {$\vdots$};
\end{tikzpicture} 
\nonumber\\ &= \Big(\prod_{i=1}^n[q x_i^2]\Big) Z_\cap(\{\zeta_i\}_{i=1}^{2n};\{q^{-1/2}y_j\}_{j=1}^{2n}).
\end{alignat}
\endgroup
The second equality is obtained by the repeated application of the Yang-Baxter equation. At the third, we used $\check R(x_i^2) |{\uparrow \uparrow} \rangle = [q x_i^2]|{\uparrow \uparrow} \rangle$. At the penultimate equality, we defined $\zeta_{2i-1} = x_i q^{1/2}, \zeta_{2i} = x_i q^{-1/2}$ for $i = 1, \dots, n$ and used
\begin{equation}
\begin{tikzpicture}[baseline=0.4cm]
\draw (0,0.5) -- (1,0.5) (0.5,0) -- (0.5,1);
\draw (0,0.5) node[left] {\scriptsize$z$};
\draw (0.5,0.05) node[below] {\scriptsize$q^{\frac12}w$};
\end{tikzpicture} \ \ = \ 
\begin{tikzpicture}[baseline=0.4cm]
\draw (0,0.5) -- (1,0.5) (0.5,0) -- (0.5,1);
\draw (0,0.5) node[left] {\scriptsize$q^{-\frac12}z$};
\draw (0.5,0) node[below] {\scriptsize$w$};
\end{tikzpicture}
\end{equation}
at each vertex. The final form of $Z_{\text{\rm \tiny A}}(\{x_i\}_{i=1}^n;\{y_i\}_{i=1}^{2n})$ in \cref{prop:ZAspin1} is obtained by using the result of \cref{prop:zcup}, simplifying the prefactors and defining $\xi_i = q^{-1/2} \zeta_i$ for $i = 1, \dots, 2n$. \hfill $\square$

%\bibliography{biblio}

\begin{thebibliography}{99}

\bibitem{razumov:00}
A.~V. {Razumov} and Y.~G. {Stroganov},
\newblock {\em {Spin chains and combinatorics}},
\newblock J. Phys. A : Math. Gen. {\textbf{34}} (2001)   3185--3190.

\bibitem{bressoudbook}
D.~Bressoud,
\newblock {\em {Proofs and confirmations: the story of the alternating sign
  matrix conjecture}},
\newblock Cambridge University Press, 1999.

\bibitem{razumov:01}
A.~V. {Razumov} and Y.~G. {Stroganov},
\newblock {\em {Spin chains and combinatorics: twisted boundary conditions}},
\newblock J. Phys. A: Math. Gen. {\textbf{34}} (2001)   5335--5340.

\bibitem{batchelor:01}
M.~T. {Batchelor}, J.~{de Gier}  and B.~{Nienhuis},
\newblock {\em {The quantum symmetric XXZ chain at {$\Delta = - 1/2 $},
  alternating-sign matrices and plane partitions}},
\newblock J. Phys. A: Math. Gen. {\textbf{34}} (2001)   L265--L270.

\bibitem{degier:02}
J.~{de Gier}, M.~T. {Batchelor}, B.~{Nienhuis}  and S.~{Mitra},
\newblock {\em {The XXZ spin chain at {$\Delta=-1/2$}: Bethe roots, symmetric
  functions, and determinants}},
\newblock J. Math. Phys. {\textbf{43}} (2002)   4135--4146.

\bibitem{mitra:04}
S.~Mitra, B.~Nienhuis, J.~{de Gier}  and M.~T. Batchelor,
\newblock {\em {Exact expressions for correlations in the ground state of the
  dense O(1) loop model}},
\newblock J. Stat. Mech.  {\textbf{P09010}} (2004).

\bibitem{mitra:04_2}
S.~Mitra and B.~Nienhuis,
\newblock {\em {Exact conjectured expressions for correlations in the dense
  $O(1)$ loop model on cylinders}},
\newblock J. Stat. Mech.  {\textbf{P10006}} (2004).

\bibitem{difrancesco:05_3}
P.~{Di Francesco} and P.~{Zinn-Justin},
\newblock {\em {Around the Razumov-Stroganov conjecture: proof of a
  multi-parameter sum rule}},
\newblock Electr. J. Comb. {\textbf{12}} (2005)  ~R6.

\bibitem{difrancesco:06}
P.~{Di Francesco}, P.~{Zinn-Justin}  and {J.-B.} {Zuber},
\newblock {\em {Sum rules for the ground states of the O(1) loop model on a
  cylinder and the XXZ spin chain}},
\newblock J. Stat. Mech. {\textbf{8}} (2006)  ~11.

\bibitem{razumov:07}
A.~V. Razumov, Yu.~G. Stroganov  and P.~Zinn-Justin,
\newblock {\em {Polynomial solutions of qKZ equation and ground state of XXZ
  spin chain at $\Delta = -1/2$}},
\newblock J. Phys. A : Math. Gen. {\textbf{40}} (2007)   11827.

\bibitem{kitanine:02}
N.~Kitanine, J.~M. Maillet, N.~A. Slavnov  and V.~Terras,
\newblock {\em Emptiness formation probability of the {XXZ} spin-$1/2$
  {H}eisenberg chain at {$\Delta =1/2$}},
\newblock J. Phys. A: Math. Gen. {\textbf{35}} {\textbf{27}} (2002)   L385.

\bibitem{cantini:12_1}
L.~{Cantini},
\newblock {\em {Finite size emptiness formation probability of the XXZ spin
  chain at $\Delta=-1/2$}},
\newblock J. Phys. A: Math. Theor. {\textbf{45}} (2012)   135207.

\bibitem{zinn:09}
P.~Zinn-Justin,
\newblock {\em Combinatorial point for fused loop models},
\newblock Comm. Math. Phys. {\textbf{272}} (2007)   661--682.

\bibitem{fonseca:12}
T.~{Fonseca} and P.~{Zinn-Justin},
\newblock {\em {Higher spin polynomial solutions of quantum
  Knizhnik--Zamolodchikov equation}},
\newblock Comm. Math. Phys. {\textbf{328}} (2012)   1079--1115.

\bibitem{difrancesco:05_4}
P.~{Di Francesco} and P.~{Zinn-Justin},
\newblock {\em {The quantum Knizhnik-Zamolodchikov equation, generalized
  Razumov-Stroganov sum rules and extended Joseph polynomials}},
\newblock J. Phys. A: Math. Gen. {\textbf{38}} (2005)   L815--L822.

\bibitem{yang:04}
X.~{Yang} and P.~{Fendley},
\newblock {\em {Non-local spacetime supersymmetry on the lattice}},
\newblock J. Phys. A: Math. Gen. {\textbf{37}} (2004)   8937--8948.

\bibitem{veneziano:06}
G.~Veneziano and J.~Wosiek,
\newblock {\em A supersymmetric matrix model: {III}. {H}idden {SUSY} in
  statistical systems},
\newblock JHEP {\textbf{11}} (2006)   030.

\bibitem{witten:82}
E.~Witten,
\newblock {\em {Constraints on supersymmetry breaking}},
\newblock Nucl. Phys. B {\textbf{202}} (1982)   253--316.

\bibitem{hagendorf:13}
C.~Hagendorf,
\newblock {\em {Spin chains with dynamical lattice supersymmetry}},
\newblock J. Stat. Phys. {\textbf{150}} (2013)   609--657.

\bibitem{zamolodchikov:81}
A.~B. Zamolodchikov and V.~A. Fateev,
\newblock {\em {A model factorized $S$-matrix and an integrable spin-$1$
  Heisenberg chain}},
\newblock Sov. J. Nucl. Phys. {\textbf{32}} (1981)   298--303.

\bibitem{fateev:81}
V.~A. Fateev,
\newblock {\em {A factorized $S$-matrix for particles of opposite parities and
  an integrable $21$-vertex statistical model}},
\newblock Sov. J. Nucl. Phys. {\textbf{33}} (1981)   761--766.

\bibitem{robbins:00}
D.~P. {Robbins},
\newblock {Symmetry Classes of Alternating Sign Matrices},
\newblock arXiv:math.CO/0008045 2000.

\bibitem{kuperberg:02}
G.~Kuperberg,
\newblock {\em {Symmetry Classes of Alternating-Sign Matrices under One Roof}},
\newblock Ann. Math. {\textbf{156}} (2002)   835--866.

\bibitem{kulish:81}
P.~P. Kulish, N.~Yu. Reshetikhin  and E.~K. Sklyanin,
\newblock {\em {Yang-Baxter equation and representation theory: I}},
\newblock Lett. Math. Phys. {\textbf{5}} (1981)   393--403.

\bibitem{kulish:82}
P.~P. Kulish and E.~K. Sklyanin,
\newblock {\em {Quantum spectral transform method}},
\newblock Lecture Notes in Phys. {\textbf{151}} (1982)   61--119.

\bibitem{kirillov:87}
A.~N. {Kirillov} and N.~Y. {Reshetikhin},
\newblock {\em {Exact solution of the integrable XXZ Heisenberg model with
  arbitrary spin. I. The ground state and the excitation spectrum}},
\newblock J. Phys. A : Math. Gen. {\textbf{20}} (1987)   1565--1585.

\bibitem{hagendorf:15}
C.~Hagendorf,
\newblock {\em The nineteen-vertex model and alternating sign matrices},
\newblock J. Stat. Mech. Theor. Exp. {\textbf{2015}} {\textbf{1}} (2015)
  P01017.

\bibitem{korepin:93}
V.~E. Korepin, N.~M. Bogoliubov  and A.~G. Izergin,
\newblock {\em {Quantum Inverse Scattering Method and Correlation Functions}},
\newblock Cambridge University Press, 1993.

\bibitem{niccoli:13}
G.~Niccoli,
\newblock {\em {Form factors and complete spectrum of XXX antiperiodic higher
  spin chains by quantum separation of variables}},
\newblock J. Math. Phys. {\textbf{54}} (2013)   053516.

\bibitem{niccoli:15}
G.~Niccoli and V.~Terras,
\newblock {\em {Antiperiodic XXZ Chains with Arbitrary Spins: Complete
  Eigenstate Construction by Functional Equations in Separation of Variables}},
\newblock Lett. Math. Phys. {\textbf{105}} (2015)   989--1031.

\bibitem{behrend:12}
R.~E. Behrend, P.~Di Francesco  and P.~Zinn-Justin,
\newblock {\em On the weighted enumeration of alternating sign matrices and
  descending plane partitions},
\newblock J. Comb. Theor. {\textbf{119}} (2012)   331--363.

\bibitem{babujian:82}
H.~M. {Babujian},
\newblock {\em {Exact solution of the one-dimensional isotropic Heisenberg
  chain with arbitrary spins S}},
\newblock Phys. Lett. A {\textbf{90}} (1982)   479--482.

\bibitem{babujian:83}
H.~M. Babujian,
\newblock {\em Exact solution of the isotropic {H}eisenberg chain with
  arbitrary spins: Thermodynamics of the model},
\newblock Nucl. Phys. B {\textbf{215}} {\textbf{3}} (1983)   317--336.

\bibitem{takhtajan:82}
L.~A. Takhtajan,
\newblock {\em {The picture of low-lying excitations in the isotropic
  Heisenberg chain of arbitrary spins}},
\newblock Phys. Lett. A {\textbf{87}} (1982)   479--482.

\bibitem{izergin:92}
A.~G. {Izergin}, D.~A. {Coker}  and V.~E. {Korepin},
\newblock {\em {Determinant formula for the six-vertex model}},
\newblock J. Phys. A : Math. Gen. {\textbf{25}} (1992)   4315--4334.

\bibitem{cauchy:41}
A.~L. Cauchy,
\newblock {\em {M\'emoire sur les fonctions altern\'ees et sur les sommes
  altern\'ees}},
\newblock Exercices Anal. et Phys. Math. 2 (1841)   151--159.

\bibitem{elkies:92}
N.~Elkies, G.~Kuperberg, M.~Larsen  and J.~Propp,
\newblock {\em {Alternating-Sign Matrices and Domino Tilings (Part II)}},
\newblock J. Alg. Comb. {\textbf{1}} (1992)   219--234.

\bibitem{tsuchiya:98}
O.~Tsuchiya,
\newblock {\em {Determinant formula for the six-vertex model with reflecting
  end}},
\newblock J. Math. Phys. {\textbf{39}} (1998)   5946--5951.

\bibitem{kitanine:01}
N.~Kitanine,
\newblock {\em Correlation functions of the higher spin {XXX} chains},
\newblock J. Phys. A: Math. Gen. {\textbf{34}} (2001)   8151.

\bibitem{difrancesco:88}
P.~Di Francesco, H.~Saleur  and J.-B. Zuber,
\newblock {\em {Generalized Coulomb-gas formalism for two-dimensional critical
  models based on $\text{SU}(2)$ coset construction}},
\newblock Nucl. Phys. B {\textbf{300}} (1988)   393--432.

\bibitem{baranowski:90}
D.~Baranowski and V.~Rittenberg,
\newblock {\em {The operator content of the ferromagnetic and antiferromagnetic
  spin-1 Zamolodchikov-Fateev quantum chain}},
\newblock J. Phys. A: Math. Gen. {\textbf{23}} (1990)   1029.

\bibitem{pearce:91}
A.~Kl\"umper, M.~T. Batchelor  and P.~A. Pearce,
\newblock {\em Central charges of the 6- and 19-vertex models with twisted
  boundary conditions},
\newblock J. Phys. A: Math. Gen. {\textbf{24}} {\textbf{13}} (1991)   3111.

\bibitem{sklyanin:88}
E.~K. Sklyanin,
\newblock {\em Boundary conditions for integrable quantum systems},
\newblock J. Phys. A: Math. Gen. {\textbf{21}} {\textbf{10}} (1988)   2375.

\end{thebibliography}
%\bibliographystyle{mybibstyle}

\end{document}